\documentclass[10pt]{article}
\usepackage[utf8]{inputenc}

\usepackage%
[%
left=3cm,%
right=3cm,%
top=3cm, %
bottom=3cm,%
a4paper%
]{geometry} %

\usepackage{latexsym, amssymb, amscd, amsthm, amsxtra, amsmath, amsthm, mathtools}
\usepackage{subcaption}
\usepackage{graphicx}
\usepackage{color}
\usepackage{natbib}
\usepackage{ifpdf}
\usepackage{multirow}
\usepackage{kotex}
\usepackage{hyperref}
\usepackage[linesnumbered,ruled]{algorithm2e}
\usepackage{enumerate}
\usepackage{epstopdf, epsfig}
\usepackage{bm}  %
\usepackage{tikz-cd}
\usepackage{apptools}
\usepackage{authblk}
\usepackage{dsfont}

\newtheorem{thm}{Theorem}
\newtheorem{cor}[thm]{Corollary}
\newtheorem{lem}[thm]{Lemma}
\newtheorem{prop}[thm]{Proposition}

\theoremstyle{definition}
\newtheorem{defn}{Definition}
\newtheorem*{defn*}{Definition}
\newtheorem{example}{Example}

\theoremstyle{remark}
\newtheorem{remark}{Remark}
\newtheorem*{remark*}{Remark}

\AtAppendix{\counterwithin{thm}{section}}

\newcommand{\norm}[1]{\left\Vert#1\right\Vert}
\newcommand{\abs}[1]{\left\vert#1\right\vert}

\newcommand{\Real}{\mathbb R}

\newcommand{\Xv}{\mathbf{B}(X)}

\def\E{\mathbb{E}}
\def\tr{\mbox{trace}}

\def\av{\mathbf a}
\def\bv{\mathbf b}

\def\ev{\mathbf e}

\def\gv{\mathbf g}

\def\jv{\mathbf j}
\def\gv{\mathbf g}

\def\rv{\mathbf r}

\def\tv{\mathbf t}
\def\uv{\mathbf u}
\def\vv{\mathbf v}
\def\wv{\mathbf w}
\def\xv{\mathbf x}
\def\yv{\mathbf y}
\def\zv{\mathbf z}

\def\Av{\mathbf A}

\def\Iv{\mathbf I}

\def\Mv{\mathbf M}

\def\Pv{\mathbf P}

\def\Sv{\mathbf S}

\def\Uv{\mathbf U}

\def\Xv{\mathbf X}

\def\Zv{\mathbf Z}

\newcommand{\deltav}{{\bm \delta}}

\newcommand{\etav}{\mbox{\boldmath{$\eta$}}}

\newcommand{\lambdav}{{\bm \lambda}}
\newcommand{\thetav}{\bm \theta}
\newcommand{\muv}{{\bm \mu}}

\newcommand{\xiv}{\mbox{\boldmath{$\xi$}}}

\newcommand{\piv}{\mbox{\boldmath{$\pi$}}}

\newcommand{\Sigmav}{{\bm \Sigma}}
\newcommand{\Lambdav}{\mbox{\boldmath{$\Lambda$}}}

\newcommand{\Cc}{\mathcal{C}}

\newcommand{\Fc}{\mathcal{F}}

\newcommand{\Lc}{\mathcal{L}}
\newcommand{\Mc}{\mathcal{M}}

\newcommand{\Sc}{\mathcal{S}}

\newcommand{\Xc}{\mathcal{X}}
\newcommand{\Yc}{\mathcal{Y}}
\newcommand{\Zc}{\mathcal{Z}}

\def\1v{\mathbf 1}
\def\0v{\mathbf 0}
\def\Id{ \mathbf{I}}

\def\P{\mathbb{P}}
\def\Lap{\mathrm{Lap}}

\title{Differentially Private Multivariate Statistics with an Application to Contingency Table Analysis\thanks{This work was supported by the National Research Foundation of Korea (NRF) grant funded by the Korea government (MSIT) (No.2019R1A2C2002256) and Institution of Information \& communications Technology Planning \& Evaluation(IITP) grant funded by the Korea government (MSIT) (No.2022-0-00937, Solving the problem of increasing the usability and usefulness synthetic data algorithm for statistical data).}}

\author[]{Minwoo Kim$^{\dagger}$}
\author[]{Jonghyeok Lee\thanks{Equal contribution.}~}
\author[]{Seung Woo Kwak}
\author[]{Sungkyu Jung\thanks{Corresponding author. Email address: \texttt{sungkyu@snu.ac.kr}}}
\affil[]{Department of Statistics, Seoul National University, \\ Seoul 08826, South Korea}

\date{\today}

\begin{document}

\maketitle

\begin{abstract} 
Differential privacy (DP) has become a rigorous central concept for privacy protection in the past decade. We use Gaussian differential privacy (GDP) in gauging the level of privacy protection for releasing statistical summaries from data. The GDP is a natural and easy-to-interpret differential privacy criterion based on the statistical hypothesis testing framework. 
The Gaussian mechanism is a natural and fundamental mechanism that can be used to perturb multivariate statistics to satisfy a $\mu$-GDP criterion, where $\mu>0$ stands for the level of privacy protection. Requiring a certain level of differential privacy inevitably leads to a loss of statistical utility. We improve ordinary Gaussian mechanisms by developing rank-deficient James-Stein Gaussian mechanisms for releasing private multivariate statistics, and show that the proposed mechanisms have higher statistical utilities. Laplace mechanisms, the most commonly used mechanisms in the pure DP framework, are also investigated under the GDP criterion. We show that optimal calibration of multivariate Laplace mechanisms requires more information on the statistic than just the global sensitivity, and derive the minimal amount of Laplace perturbation for releasing $\mu$-GDP contingency tables. Gaussian mechanisms are shown to have higher statistical utilities than Laplace mechanisms, except for very low levels of privacy. The utility of proposed multivariate mechanisms is further demonstrated using differentially private hypotheses tests on contingency tables. Bootstrap-based goodness-of-fit and homogeneity tests, utilizing the proposed rank-deficient James--Stein mechanisms, exhibit higher powers than natural competitors. 
\end{abstract}


\section{Introduction} \label{section: intro}
Privacy has become a crucial concern in modern statistical analyses and machine learning.
Data sets collected in recent years have grown exponentially in size and often contain sensitive personal information.
Consequently, it has become increasingly important to prevent privacy disclosure when publicly releasing data sets containing personal information, or even when releasing statistical summaries from the data sets. In this paper, we investigate private release of multivariate statistics satisfying a pre-specified degree of privacy protection, governed by differential privacy. Differential privacy \citep{dwork2006calibrating,dwork2006our}, or DP, and its variants \citep{mironov2017renyi,dong2022gaussianDP} provide formal criteria for data privacy and can be achieved by randomly perturbing the statistics to be released.

Let $S = (x_1, \dots, x_n)  \in \Xc^n$ be a data set collected from  a sample space $\Xc$, and let $\thetav(S) \in \Real^p$ denote a multivariate summary statistic computed from the data set $S$. (We also call the map $\thetav: \Xc^n \to \Real^p$, defined for any $n \in \mathbb{N}$, a statistic.)
Our development applies to any form of  multivariate statistics, from the multivariate mean vector to regression coefficient estimators.
Since the most common types of data release from data sets containing sensitive personal information are \emph{contingency tables} and \emph{histograms},
we will pay special attention to the case that $\thetav$ is a contingency (or frequency) table.

Let $\Mv: \mathcal{X}^n \to \Real^p$ be a randomized map, called a \emph{mechanism}, that outputs a randomly perturbed
$\thetav(S)$.
A commonly used mechanism for differential privacy is an additive mechanism that adds a random noise $\xi$ to the given statistic, e.g.,
\begin{equation}\label{eq:GaussianMechanism}
  \Mv(S) = \thetav(S) + \xi, \quad \xi \sim N_p(\0v, \sigma^2 \Id_p),
\end{equation} 
where $\Id_p$ is the $p\times p $ identity matrix.
We choose to use Gaussian differential privacy (GDP) of \cite{dong2022gaussianDP} to measure the degrees of privacy protection in a mechanism $\Mv$, and to calibrate the amount of noise infused in $\Mv$. In the GDP framework, the level of privacy protection is summarized by a single parameter $\mu >0$, and we say a randomized mechanism $\Mv$ is $\mu$-GDP if determining whether any specific piece of personal information is present in the data set $S$, based on a single output $\Mv(S)$, is as difficult as testing
$$
H_0 : X \sim N(0,1) \quad {\rm vs.} \quad H_1 : X \sim N(\mu,1)
$$
based on a single observation of $X$. It is known that for the Gaussian mechanism (\ref{eq:GaussianMechanism}) to satisfy $\mu$-GDP (for a pre-specified $\mu > 0$), the amount of random noise in $\xi$ should satisfy $\sigma \ge \mu^{-1}\Delta_2(\thetav; \Xc^n)$, where $\Delta_2(\thetav; \Xc^n)$, called \emph{$L_2$-sensitivity} of $\thetav$, is the largest possible $L_2$-discrepancy of $\thetav$ when one individual's record in the data is replaced by someone else's. See Section \ref{subsec:1.1} for a brief review of GDP.

In this work, we propose and compare various refinements on multivariate additive mechanisms stemming from  adding Gaussian or Laplace-distributed noise, that are calibrated to satisfy $\mu$-GDP for each given  $\mu > 0$.
In particular, we propose \emph{rank-deficient James--Stein post-processed Gaussian mechanisms}, denoted $\Mv_{rJS}$, and show that $\Mv_{rJS}$ has a smaller loss of statistical utility than ordinary Gaussian mechanisms of (\ref{eq:GaussianMechanism}) in Section \ref{section: mv-mech}. The statistical utility of a mechanism $\Mv$ is measured via the mean squared error of $\Mv$ when compared to $\thetav$, i.e., $L_2(\Mv) = E \| \Mv(S) - \thetav(S)\|_2^2$. The rank-deficient James--Stein mechanism $\Mv_{rJS}$ improves upon Gaussian mechanisms in two ways.

First, $\Mv_{rJS}$ exploits the set $\Sc_{\thetav }:= \{\thetav(S) - \thetav(S'): S, S' \in \Xc^n, S \sim S'\} \subset \Real^p$ (the notation $S\sim S'$ means that data sets $S$ and $S'$ differ in exactly one record) rather than using only the maximum potential difference
$\Delta_2(\thetav; \Xc^n) = \max_{\vv \in \Sc_{\thetav }} \|\vv\|_2$
between data sets, which was used for calibration of (\ref{eq:GaussianMechanism}).
We show that if ${\rm span}(\Sc_{\thetav })$ has dimension less than $p$, then adding Gaussian noise with a carefully chosen rank-deficient covariance matrix not only satisfies $\mu$-GDP but also has larger statistical utility than (\ref{eq:GaussianMechanism}).
 The rank-deficient mechanisms are especially useful for releasing private contingency tables, since for a frequency table $\thetav$ (with no overlapping categories), ${\rm span}(\Sc_{\thetav })$ has dimension $p-1$, and is indeed a proper subspace of $\Real^p$.

Another improvement of a Gaussian mechanism (\ref{eq:GaussianMechanism}) is obtained by perceiving  (\ref{eq:GaussianMechanism}) for fixed $S$ as a multivariate mean $\thetav(S)$ estimation problem with one observation $\Mv(S)$. If the dimension $p$ is not small, the celebrated James--Stein shrinkage estimator provides smaller mean squared error than the ordinary least squares estimator (which is $\Mv(S)$ itself in our context)  \citep{james1961estimation, efron2012large,balle2018improving}.
The rank-deficient James--Stein mechanism $\Mv_{rJS}$ exploits both rank-deficiency of ${\rm span}(\Sc_\theta)$ and the James--Stein shrinkage phenomenon, and we show that $\Mv_{rJS}$ can be calibrated to satisfy $\mu$-GDP, and has the largest statistical utility among natural competitors of the same privacy level. 

In addition, we provide the optimal calibration for  Laplace mechanisms of the form $\Mv(S) = \thetav(S) + \xi$, $\xi = (\xi_1,\ldots,\xi_p)$, where each $\xi_i$ follows independently Laplace distribution with scale parameter $b$. In Section~\ref{sec:Lap_mech}, we develop several important technical lemmas that can be used to calibrate \emph{multivariate} Laplace mechanisms with respect to the $\mu$-GDP criterion, which turns out to be substantially different than the case of univariate statistics. 
Optimal calibration of Laplace mechanisms heavily depends on the sensitivity space $\Sc_{\thetav}$, and we provide optimal calibration results for the case where $\Delta_1(\thetav) =\max_{\vv \in \Sc_{\thetav} } \|\vv\|_1$ is the only information on $\thetav$, and also for the case that $\thetav$ is a contingency table. 

In Section \ref{sec:comparison}, we show that the statistical utility of Laplace mechanisms is {typically} lower than that of Gaussian mechanisms, when both mechanisms satisfy $\mu$-GDP. This makes sense because the GDP criterion uses the trade-off between two Gaussian distributions. However, perhaps surprisingly, Laplace mechanisms have higher utility than Gaussian mechanisms for low privacy regimes with larger $\mu$.

The proposed mechanisms and their calibration results are demonstrated by an application of releasing differentially private contingency (or frequency) tables, and we numerically compare the statistical utilities of various mechanisms in Section~\ref{sec:app_conting_tab}. We also provide private hypothesis testing procedures for testing goodness-of-fit and homogeneity among different populations, using the differentially private contingency tables we have developed. The test procedure assumes that the only available information on the data set is $\Mv(S)$, a randomly perturbed contingency table, and uses a parametric bootstrap procedure in the computation of p-values. 
In a numerical study, we reveal our proposed test procedures control the type I error rates at a given significance level, and the power increases as the sample size increases. Our test procedure based on the rank-deficient James--Stein mechanism $\Mv_{rJS}$ shows the highest empirical power among competitors. 

We note that there are other types of randomized mechanisms that can be used to satisfy a DP criterion. 
Prominent examples are the exponential mechanism \citep{mcsherry2007mechanism}, staircase mechanism \citep{kairouz2014extremal}, randomized responses \citep{dwork2006calibrating}, and $K$-norm mechanisms \citep{awan2021structure}. Compositions of these mechanisms can also satisfy a given privacy criterion \citep{dwork2014algorithmic,abadi2016deep,dong2021central,dong2022gaussianDP}. 
\cite{balle2018improving} proposed to use post-processing in the form of shrinkage, as we investigate in this work. While \cite{balle2018improving} demonstrated the advantage of the shrinkage empirically and also in a Bayesian perspective, we extend it to the case of shrinking towards the mean while simultaneously utilizing the structure of the sensitivity space $\Sv_{\thetav}$, and show its advantage more formally. 
Exploiting $\Sv_{\thetav}$ for maximizing statistical utilities of randomized mechanisms has been considered in   \cite{karwa2017sharing,avella2021privacy,awan2021structure} as well. The trade-off between differential privacy and statistical utility is studied in \cite{wasserman2010statistical, duchi2018minimax, cai2021cost}. Several authors including \cite{wasserman2010statistical,gaboardi2016DPchisq, wang2017revisiting, kifer2016new,son2022parametric} proposed differentially private procedures for contingency table analysis (or for releasing private histogram counts). 
All of these previous works, except \citep{dong2021central,dong2022gaussianDP}, are developed under the original $\epsilon$-DP criterion, or a  relaxed $(\epsilon,\delta)$-DP criterion. On the other hand, we use the GDP criterion, since  it enables a straightforward interpretation of privacy parameter $\mu$ as the effect size in a hypothesis testing framework. 

 In Section~\ref{subsec:1.1}, we provide a concise review of differential privacy, viewed in terms of statistical hypothesis testing. Multivariate Gaussian mechanisms and Laplace mechanisms are developed and studied in Sections~\ref{section: mv-mech} and \ref{sec:Lap_mech}, respectively, and their statistical utilities are compared in Section \ref{sec:comparison} with respect to the $L_r$-loss, $r \ge 1$. Applications to contingency table analysis are discussed in Section~\ref{sec:app_conting_tab}. Some important remarks about multivariate Laplace mechanisms are given in Section~\ref{sec:discussion}. Technical details and proofs of theoretical results are contained in the Appendix. 

\subsection{Differential privacy based on hypothesis testing} \label{subsec:1.1}

Let $\Xc$ be the sample space in which individual records of the data lie. For any two data sets $S = (x_1, \cdots, x_n) \in \Xc^n$ and $S' = (x_1', \cdots, x_n') \in \Xc^n$ of equal sample size, we say $ S $ and $ S' $ are neighbors if they differ by a single record, i.e., if there exists an $i \in \{1,\ldots,n\}$ such that $x_i \neq x_i'$ and $x_\iota = x_\iota'$ for all $\iota \neq i$.
A mechanism $M : \Xc^n \rightarrow \Yc $ takes as input a data set $ S $ and releases some
randomly perturbed data summary $ M(S)$ in some abstract space $\Yc$.
Having the data set $S$ fixed, $M(S)$ is a random variable taking values in $\Yc$. A mechanism $M$ is said to be differentially private at level $\varepsilon>0$, or $\varepsilon$-DP \citep{dwork2006calibrating}, if
for any neighboring data sets $S$, $S'$ and for any set $A$ of possible mechanism outputs
\begin{equation}\label{eq:epsilon-DP-definition}
   e^{-\varepsilon} \le  \frac{\P(M(S) \in A)}{\P(M(S') \in A)} \le e^\varepsilon,
\end{equation}
in which the probability is taken over the randomness of $M$, with $S$ and $S'$ being fixed.
For smaller values of the privacy parameter $\varepsilon$, it  is difficult to infer whether a specific private information is contained in the underlying data set $S$,  just based on a perturbed observation of $M(S)$. Note that the relaxed privacy criterion of $(\varepsilon, \delta)$-DP \citep{dwork2006our}, with an additional privacy parameter $\delta>0$, is satisfied for a mechanism $M$ if $
    \P(M(S) \in A) \le e^\varepsilon \P(M(S') \in A) + \delta$
    for all neighboring data sets $S$, $S'$ and for any $A$.

Differential privacy criteria can be equivalently described by the statistical hypothesis testing framework \citep{wasserman2010statistical, geng2015staircase,kairouz2015composition,dong2022gaussianDP}. For a fixed pair of neighboring data sets $S, S'$, denote the distribution of $M(S)$ (or $M(S')$) by $P$ (or $Q$, respectively). Consider testing
    $$ H_0: \mbox{data set is } S \quad vs. \quad H_1: \mbox{data set is } S',$$
or equivalently,
\begin{equation}\label{eq:test}
     H_0: M \sim P \quad vs. \quad H_1: M \sim Q,
\end{equation}
based on a single observation $M$, where $M = M(S) \sim P $ under $H_0$ and $M = M(S') \sim Q$ under $H_1$. %
The difference between $P$ and $Q$ may be described by the smallest type II error rate of level-$\alpha$ tests, for each $\alpha \in (0,1)$. This amounts to quantifying the trade-off between the level of a test and its power. (Recall that requiring more strict control of type I errors (i.e., smaller level) for a test results in smaller power of the test.) Let $\phi: \Yc \to [0, 1]$ be the decision function of a (randomized) test procedure for hypotheses \eqref{eq:test}, and write $\alpha_\phi
    = \E_{M \sim P}[\phi(M)]$ for the type I error rate, and $\beta_\phi
    = 1 - \E_{M \sim Q}[\phi(M)]$ for the type II error rate, of $\phi$.
Then the trade-off function $T(P,Q): [0,1] \to [0,1]$ between distributions $P$ and $Q$ (or between hypotheses \eqref{eq:test}) is defined by
$$
T(P,Q)(\alpha) = \inf_\phi\{\beta_\phi : \alpha_\phi \le \alpha\}.
$$
Note that for any $\alpha \in (0,1)$, %
$T(P,Q)(\alpha)$ is the type II error rate of the most powerful test of level $\alpha$.

For two distribution pairs $(P,Q)$ and $(F,G)$, if $T(P,Q) \ge T(F,G)$
(i.e., $T(P,Q)(\alpha) \ge T(F,G)(\alpha)$ for all $\alpha \in(0,1)$),
then $P$ and $Q$ are harder to distinguish than $F$ and $G$ are. That is, the indistinguishability between $M(S)$ and $M(S')$ represented by the corresponding trade-off function $T(M(S), M(S')) := T(P,Q)$ can be bounded by other trade-off functions. Based on this observation, \cite{dong2022gaussianDP} proposed a general criterion of differential privacy based on trade-off functions.
A function $f$ is called a trade-off function if there exist distributions $F,G$ such that $f = T(F,G)$.
Generally, any $f:[0,1] \to [0,1]$ is a trade-off function if it is convex, continuous, non-increasing, $f(x) \le 1 - x$ and is symmetric about the line $x=y$.
A mechanism $M$ is said to be $f$-differentially private (or $f$-DP) for a trade-off function $f$ if
\begin{equation}\label{eq:none}
  \inf_{S \sim S'} T(M(S), M(S')) \ge f,
\end{equation}
where the infimum is taken over the family of all neighboring pairs $S$ and $S'$ in $\mathcal X^n$. We write $S\sim S'$ if $S$ and $S'$ are neighbors.

The $(\varepsilon, \delta)$-DP criterion can be represented by the $f$-DP framework.
In particular, for any $\varepsilon >0$, $\delta \ge 0$, a mechanism $M$ is $(\varepsilon, \delta)$-DP if and only if $M$ is $f_{\varepsilon, \delta}$-DP, where
\begin{equation}\label{eq:DPtof-DP}
 f_{\varepsilon, \delta}(\alpha)
= \max\{
0,
1-\delta-e^{\varepsilon}\alpha,
e^{-\varepsilon}(1-\delta-\alpha)
\},
\end{equation}
as shown in \cite{wasserman2010statistical, kairouz2015composition,dong2022gaussianDP}.

A simple one-parameter family of trade-off functions is obtained by univariate Gaussian distributions differing only by the mean parameter. For $\mu \ge 0$, let $G_\mu$ be the trade-off function between $N(0,1)$ and $N(\mu,1)$. Using the Neyman--Pearson lemma, one can check that for $\alpha \in (0,1)$
\begin{equation}\label{eq:Gmu}
    G_\mu(\alpha) = T(N(0,1), N(\mu,1))(\alpha) = \Phi ( \Phi^{-1} ( 1-\alpha) - \mu),
\end{equation}
where $\Phi$ is the distribution function of the standard normal distribution \citep{dong2022gaussianDP}.

\begin{defn}[\cite{dong2022gaussianDP}] \label{def:GDP}
For a $\mu \ge 0$,  A mechanism $M$ is $\mu$-Gaussian differentially private (or $\mu$-GDP) if
$$\inf_{S \sim S'} T(M(S), M(S')) \ge G_\mu.$$
\end{defn}

If $M$ is $\mu$-GDP, distinguishing any two neighboring data sets based solely on an observed value of $M$ is at least as hard as distinguishing $N(0,1)$ and $N(\mu,1)$. Since
$0 \le \mu < \mu'$ if and only if $G_\mu > G_{\mu'}$,
the smaller the privacy parameter $\mu$ is, the greater privacy protection is guaranteed for a mechanism satisfying $\mu$-GDP.

Note that an ordinary (multivariate) Gaussian mechanism in the form of
\begin{equation}\label{M_2}
    \Mv_G(S) \coloneqq \thetav(S) + \xiv, \quad \xiv \sim N_p(\0v, \sigma^2 \Id_p),
\end{equation}
(where $\thetav(S) \in \Real^p$) is $\mu$-GDP for any $\mu \ge \mu_0 := \Delta_2 / \sigma$, where $\Delta_2 = \sup_{S \sim S'} \| \thetav(S)  - \thetav(S')\|_2$ is called the global $L_2$-sensitivity of the multivariate statistic $\thetav(S)$. (See Theorem~\ref{thm:G-mech} in Appendix~\ref{sec:appendix_lemmas_trade-off}.)
That is, the level of privacy protection of a Gaussian mechanism $\Mv_G$ in \eqref{M_2} is precisely described by the privacy parameter $\mu_0$ under the GDP criterion. On the other hand, the Gaussian mechanism (\ref{M_2}) does not satisfy $\varepsilon$-DP for any $\epsilon>0$, and there is no single set of privacy parameters $(\varepsilon, \delta)$ representing the degrees of privacy protection of a Gaussian mechanism under the $(\varepsilon, \delta)$-DP criterion. To be precise, the Gaussian mechanism (\ref{M_2})
satisfies $(\epsilon, \delta)$-DP for any combination of privacy parameters $(\varepsilon, \delta) \in (0,\infty)^2$ satisfying
$\Phi(\tfrac{\Delta_2}{2\sigma} - \tfrac{\varepsilon \sigma}{\Delta_2}) - e^\varepsilon \Phi(-\tfrac{\Delta_2}{2\sigma} - \tfrac{\varepsilon \sigma}{\Delta_2}) \le \delta$, where $\Phi$ is the standard normal distribution function \citep[Theorem 8,][]{balle2018improving}. Therefore, when statistics to be released are perturbed by Gaussian noise, using GDP criterion in gauging the privacy protection is more convenient (and is easier to interpret) than using $(\varepsilon, \delta)$-DP criterion.
For the rest of the article, we use the GDP criterion to control the level of privacy protection in randomized mechanisms.

Verifying whether a multivariate mechanism $\Mv$ satisfies $\mu$-GDP can be done via directly comparing trade-off functions $T(\Mv(S), \Mv(S'))$ with $G_\mu$. In Appendix~\ref{sec:appendix_lemmas_trade-off}, we provide several tools that can be used to work with trade-off functions between multivariate distributions.

As seen in (\ref{M_2}), the amount of noise in a randomized mechanism can be calibrated using the global $L_2$-sensitivity of the multivariate statistic $\thetav$.
The global $L_r$-sensitivity, or simply $L_r$-sensitivity, for $r > 0$, is
\begin{equation}\label{eq:Lr-sensitivity}
   \Delta_r = \Delta_r(\thetav; \Xc^n) = \sup_{S\sim S'} \norm{ \thetav(S) - \thetav(S')}_r,
\end{equation}
and  measures the maximum amount of changes in $\thetav$, in the $\ell_r$-norm sense, when one record in $S$ is replaced by any potential record in $\Xc$, for any $S$.
As an instance, if $\Xc = [0,1]^p$, the vector average $\thetav(S) = \bar{\xv}_n = \tfrac{1}{n}\sum_{i=1}^n \xv_i$, where $S = \{\xv_1,\ldots,\xv_n\} \subset \Xc^n$, has the $L_r$-sensitivity $\Delta_r = \tfrac{1}{n}\| \1v_p\|_r = (p^{1/r}) / n$.

\section{Improving Gaussian multivariate mechanisms} \label{section: mv-mech}

In this section, we propose rank-deficient James--Stein Gaussian mechanism $\Mv_{rJS}$ and investigate under which situations $\Mv_{rJS}$ leads to smaller loss of statistical utility than the ordinary Gaussian mechanism (\ref{eq:GaussianMechanism}).

The statistical utility of a randomized mechanism $\Mv$, obtained by a perturbation of $\thetav$, is measured
via the mean squared error of $\Mv$ compared to $\thetav$. For a fixed data set $S$, the $L_2$-cost of $\Mv(S)$ is defined to be
\begin{equation}\label{eq:L2cost_def}
    L_2(\Mv(S)) \coloneqq \E_{\Mv | S} \norm{ \Mv(S) - \thetav(S) }_2^2.
\end{equation}
Borrowing terminology from the theory of estimation, for mechanisms $\Mv$ and $\Mv'$ (both perturbing $\thetav$, we say $\Mv$ \textit{dominates} (in the $\ell_2$ sense) $\Mv'$ if for any $S \in \Xc^n$
\begin{align}\label{eq:dominance}
    L_2(\Mv(S)) \le L_2(\Mv'(S))
\end{align}
(with strict inequality for some $S$),  and $\Mv$ \textit{strictly dominates} $\Mv'$ if, for any $S \in \Xc^n$,     $L_2(\Mv(S)) < L_2(\Mv'(S))$.  Therefore,  for any pre-specified privacy level $\mu >0$, if both $\Mv$ and $\Mv'$ are $\mu$-GDP, and $\Mv$ dominates $\Mv'$, then $\Mv$ is said to have higher statistical utility than $\Mv'$.
Note that one may alternatively view the data set $S$ as a random sample, and use the unconditioned $L_2$-cost of $\Mv$, $L_2(\Mv) := \E \norm{ \Mv(S) - \thetav(S) }_2^2 = \E_S(  L_2(\Mv(S)) )$. If $\Mv$ dominates $\Mv'$, then $L_2(\Mv) \le L_2(\Mv')$, while the converse is generally not true. (If the distribution of $S$ leads that the strict inequality in (\ref{eq:dominance}) holds with a positive probability, then $L_2(\Mv) < L_2(\Mv')$.) Therefore, it is sufficient to compare the conditional $L_2$-costs (\ref{eq:L2cost_def}) of mechanisms, for all possible data sets $S \in \Xc^n$.

\subsection{Rank-deficient Gaussian mechanism}\label{sec:rank-defi_Gauss}

When the \emph{only} information we have on a multivariate statistic $\thetav$ is the $L_2$-sensitivity, the ordinary Gaussian mechanism
\begin{equation}\label{eq:gaussianMech_again}
 \Mv_G(S) = \thetav(S) + N_p(\0v, \frac{\Delta_2^2}{\mu^2}\Iv_p),
\end{equation}
calibrated to satisfy $\mu$-GDP (for any $\thetav$ with $\Delta_2(\thetav; \Xc^n) = \Delta_2$), provides the largest statistical utility among a class of $\mu$-GDP additive Gaussian mechanisms, as explained below.

Let $\Mc_{\mu}^{\Delta_2}$, for $\mu>0$ and $\Delta_2 > 0$, be the collection of all additive Gaussian mechanisms satisfying $\mu$-GDP for any statistic $\thetav$ of $\Delta_2(\thetav; \mathcal X^n) \le \Delta_2$, that are of the form $\Mv_\Sigmav(S) := \thetav(S) + N_p(\0v, \Sigmav)$, where $\Sigmav$ is a  $p\times p$ symmetric non-negative definite matrix.
Note that the statistical utility, or the $L_2$-cost, of $\Mv_\Sigmav$  is $\E \| \Mv_{\Sigmav}(S) - \thetav(S)\|^2 =  {\rm trace}(\Sigmav)$. Thus, the smaller trace$(\Sigmav)$, the larger statistical utility. However, in order to ensure that $\Mv_{\Sigmav}(S)$ is $\mu$-GDP, $\Sigmav$ can not be too small, as the following lemma states.

\begin{lem}\label{lem:iso_Gauss} Let $\mu > 0$, $\Delta_2 >0$ and $\Sigmav$ be a $p\times p$ symmetric non-negative definite matrix. The additive Gaussian mechanism $\Mv_\Sigmav$ is $\mu$-GDP for any $\thetav$ with $\Delta_2(\thetav; \Xc^n) \le \Delta_2$ if and only if $\lambda_{\min}(\Sigmav) \ge \Delta_2^2/\mu^2$, where $\lambda_{\min}(\Sigmav)$ is the smallest eigenvalue of $\Sigmav$.
\end{lem}

A consequence of Lemma~\ref{lem:iso_Gauss} is that each mechanism in the class $\Mc_\mu^{\Delta_2}$ corresponds to a covariance matrix $\Sigmav$ satisfying $\lambda_{\min}(\Sigmav) \ge \Delta_2^2/\mu^2$. Particularly, {$\Mv_G \in \Mc_\mu^{\Delta_2}$. Moreover, for any mechanism $\Mv_\Sigmav \in \Mc_\mu^{\Delta_2}$, we have
$$L_2(\Mv_\Sigmav) = \tr(\Sigmav) \ge p \lambda_{\min}(\Sigmav)\ge p\Delta_2^2/\mu^2 = L_2( \Mv_G(S)) $$
for any data set $S$, and both equalities hold if and only if $\Sigmav = \Delta_2^2/\mu^2 \Iv_p$.
Thus, the ordinary Gaussian mechanism $\Mv_G$
strictly dominates any other mechanism in $\Mc_\mu^{\Delta_2}$.

We note that when considering only the class of Gaussian mechanisms with isotropic noises, a stronger statement can be made:

\begin{lem}\label{lem:iso_Gauss_stronger}
Let $\mu >0$, $\Delta_2>0$ and $\Sigmav = \sigma^2 \Iv_p$. Let $\thetav: \Xc^n \to \Real^p$ be a multivariate statistic with $\Delta_2(\thetav; \Xc^n) = \Delta_2$. Then, $\Mv_\Sigmav$ is $\mu$-GDP if and only if $\sigma^2  \ge \Delta_2^2/\mu^2$.
\end{lem}

Nevertheless, if more information about the multivariate statistic $\thetav$ is known to us, one can improve upon the ordinary Gaussian mechanism by exploiting the additional information. We assume that, for a statistic $\thetav$, the set $\mathcal{S}_{\thetav} = \{\thetav(S) - \thetav(S'): S, S' \in \Xc^n, S\sim S'\}$, which may be called the \emph{sensitivity space} of $\thetav$ \citep{awan2021structure}, is known to us.

For a multivariate statistic $\thetav$ taking values in $\Real^p$, we say $\mathcal{S}_{\thetav}$ is \textit{rank-deficient} if $d_{\thetav}:={\rm dim}({\rm span} (\mathcal{S}_{\thetav})) < p$.
If $\mathcal{S}_{\thetav}$ is rank-deficient, then for any $S$, when  one record of $S$ is replaced by any element in $\Xc$, the value of the statistic $\Pi_{\mathcal{S}_{\thetav}}^\perp \thetav(S)$, where $\Pi_{\mathcal{S}_{\thetav}}^\perp$ is the orthogonal projection operator onto the orthogonal complement of ${\rm span} (\mathcal{S}_{\thetav})$, does not vary.
Thus, in view of privacy protection, there is no need to add a perturbation along
the orthogonal complement of
${\rm span} (\mathcal{S}_{\thetav})$,
as $\thetav(S)$ and $\thetav(S')$ (for $S\sim S'$) are already indistinguishable along the subspace. This observation motivates us to consider a rank-deficient Gaussian mechanism, defined as follows.

Suppose $\thetav$ is given, and $\mathcal{S}_{\thetav}$ is {rank-deficient}. Let $\Pv_{\thetav}$ be the $p \times p$ matrix of orthogonal projection whose range is ${\rm span} (\mathcal{S}_{\thetav})$. For any orthogonal basis $\uv_1,\ldots, \uv_{d_{\thetav}}$ for the subspace ${\rm span}(\mathcal{S}_{\thetav})$, $\Pv_{\thetav} = \sum_{i=1}^{d_{\thetav}} \uv_i\uv_i^\top$.
A multivariate additive Gaussian mechanism $\Mv_r$ given by
\begin{equation}\label{M_3}
    \Mv_r(S) \coloneqq \thetav(S) +  N_p\left( \0v, \sigma^2 \Pv_{\thetav}\right),
\end{equation}
for some $\sigma >0$, is called a rank-deficient Gaussian mechanism.

\begin{thm} \label{thm:rank-deficient-G-mech}
Let $\mu >0$. For a multivariate statistic $\thetav$ with rank-deficient $\Sc_{\thetav}$, the rank-deficient Gaussian mechanism $\Mv_r$ is $\mu$-GDP if and only if $\sigma \ge \Delta_2(\thetav; \Xc^n) / \mu$.
\end{thm}

In the following, $\Mv_r$ with $\sigma = \Delta_2(\thetav; \Xc^n) / \mu$ is referred to as the $\mu$-GDP rank-deficient Gaussian mechanism.
It is easily checked that for any $\mu>0$, $L_2(\Mv_r(S)) = {d_{\thetav}\Delta_2^2(\thetav)}/{\mu^2}$ (for any $S$) and $\Mv_r$  strictly dominates $\Mv_{G}$, while both mechanisms satisfy $\mu$-GDP.

The advantage of $\Mv_r$ over $\Mv_G$ is realized when $\Sc_{\thetav}$ is rank-deficient. In the next example we exhibit that when the statistic $\thetav$ is a frequency table (or the vector of histogram counts, or a contingency table), $\Sc_{\thetav}$ is indeed rank-deficient.

\begin{example}\label{ex:contingency-table}
Suppose that the sample space $\Xc$ is partitioned into $p$ parts (or categories), denoted by $C_1,\ldots,C_p$. For $S = (x_1,\ldots,x_n) \in \Xc^n$, let $\thetav(S)$ be the frequency table $\thetav(S) = (\theta_1(S),\ldots, \theta_p(S))^\top$, where $\theta_i(S)= \#\{j : x_j \in C_i\}$  is the count of records in $S$ belonging to the category $C_i$. Since the categories do not overlap, if one record in $S$ is replaced by any element in $\Xc$, resulting in a neighboring data set $S'$, then
$    \thetav(S) - \thetav(S') = \ev_i - \ev_j$,
for some $i,j \in \{1,\ldots,p\}$, where $\{\ev_1,\ldots, \ev_p\}$ is the standard basis of $\Real^p$. That is, for any $S\sim S'$, $\thetav(S')$ differs from $\thetav(S)$ at most two coordinates, and the sensitivity space of $\thetav$ is $\Sc_{\thetav} = \{\ev_i - \ev_j : i,j = 1,\ldots, p \}$. Thus, for any $p$ and for any partition, $\Delta_r(\thetav) = 2^{1/r}$ for any $r>0$; in particular,
$$\Delta_2(\thetav; \Xc^n) = \sqrt{2}, \quad \Delta_1(\thetav) = 2.$$
Moreover, $\Sc_{\thetav}$ is rank-deficient since ${\rm span}(\Sc_{\thetav})$ is orthogonal to $\1v_p = (1,\ldots,1)^\top \in \Real^p$. Therefore, the rank-deficient Gaussian mechanism for frequency table $\thetav$, $$\Mv_r(S) = \thetav(S) + N_p(\0v, \sigma^2 (\Id_p - \frac{1}{p}\1v_p\1v_p^\top) ),$$
where $\sigma  = \Delta_2(\thetav; \Xc^n) / \mu = \sqrt{2}/ \mu$, is $\mu$-GDP (for any $\mu$) and strictly dominates the ordinary Gaussian mechanism.
\end{example}

\subsection{James--Stein shrinkage mechanisms}\label{sec:2.2}
The mechanism output $\Mv_G(S)$ of an ordinary Gaussian mechanism can be viewed as the maximum likelihood estimator of $\thetav(S) \in \Real^p$ (for each given $S$) from the model $(\Mv_G(S))_i \sim_{\rm ind} N(\theta_i(S), \sigma^2)$ $(i=1,\ldots,n)$.
It is well-known that, for $p \ge 3$, the celebrated James--Stein shrinkage estimator \citep{james1961estimation} dominates $\Mv_G(S)$ under the squared error loss in estimation of $\thetav(S)$. In this section, we propose two types of mechanisms derived by applications of James--Stein shrinkage, confirm that these new mechanisms can be calibrated to satisfy $\mu$-GDP (for any $\mu>0$), and investigate under which situations the James--Stein mechanisms have high statistical utilities.

The James--Stein shrinkage Gaussian mechanisms are defined by shrinking the output of ordinary Gaussian mechanism towards a common value. For notational convenience, for a given $\sigma > 0$, write the James--Stein shrinkage operator (shrinking all elements of a vector toward zero) by $m_0(\cdot; \sigma): \Real^p \to \Real^p$, which is defined by
$$
m_0(\xv; \sigma) = \left(1 - \frac{(p-2)\sigma^2}{\left\| \xv \right\|_2^2}\right) \xv.$$
Similarly, write the operator shrinking to the average $\bar{x} = p^{-1}\1v_p^\top \xv$ by $m_a(\cdot; \sigma): \Real^p \to \Real^p$,
\begin{equation}\label{eq:shrink_to_ave}
  m_a(\xv; \sigma) =\bar{x}\1v_p + \left( 1 - \frac{(p-3)\sigma^2}{\|\xv_c\|_2^2} \right)\xv_c,
\end{equation}
where $\xv_c = \xv - \bar{x} \1v_p$.

\begin{defn}\label{def:JS}\label{eq:M_4}\label{M_5}
Let $\thetav: \Xc^n \to \Real^p$ be a multivariate statistic, and $\sigma > 0$ be given.
\begin{itemize}
  \item[(i)] The James--Stein Gaussian mechanism shrinking to zero is $\Mv_{JS0}: \Xc^n \to \Real^p$, defined by  $\Mv_{JS0}(S) = m_0( \Mv_G(S) ; \sigma)$,
where $\Mv_G(S) = \thetav(S) + N_p(\0v, \sigma^2 \Iv_p)$.
  \item[(ii)]
 The James--Stein Gaussian mechanism shrinking to average   is $\Mv_{JS}: \Xc^n \to \Real^p$, defined by $\Mv_{JS}(S) = m_a( \Mv_G(S) ; \sigma)$.
\end{itemize}
\end{defn}

Both $\Mv_{JS0}(S)$ and $\Mv_{JS}(S)$ depend on the data set $S$ only through the output $\Mv_G(S)$ of the ordinary Gaussian mechanism. In the data privacy literature, such mechanism is said to be obtained via \emph{post-processing} $\Mv_G(S)$. If a randomized mechanism $\Mv$ satisfies a privacy criterion (e.g., it is $\mu$-GDP for some $\mu>0$), then any information obtained via a post-processing of (the output of) $\Mv$ retains at least the same level of privacy protection guaranteed for $\Mv$ \citep{dwork2014algorithmic,dong2022gaussianDP}; see Lemma~\ref{lem:post-processing} in Appendix~\ref{sec:appendix_lemmas_trade-off}. This leads that the James--Stein mechanisms are $\mu$-GDP for any $\mu \ge \Delta_2(\thetav; \Xc^n)/\sigma$. Equivalently, for any $\mu>0$, both $\Mv_{JS0}(S)$ and $\Mv_{JS}(S)$ with $\sigma \ge \Delta_2(\thetav; \Xc^n)/\mu$   satisfy $\mu$-GDP. We say that a James--Stein mechanism is calibrated to satisfy $\mu$-GDP if $\sigma = \Delta_2(\thetav; \Xc^n)/\mu$.

We remark that post-processed Gaussian mechanisms, including $\Mv_{JS0}$, have been considered in the literature; see, e.g., \cite{balle2018improving}. Nevertheless, to the best of our knowledge, our work is the first attempt of calibrating the amount of perturbation with respect to the GDP criterion, and of investigating $\Mv_{JS}$, the James--Stein Gaussian mechanism shrinking to average.

The conditional $L_2$-costs of $\Mv_{JS0}$ and $\Mv_{JS}$ are given as follows.

\begin{prop}\label{prop:M45_costs} For a multivariate statistic $\thetav: \Xc^n \to \Real^p$ with $\Delta_2(\thetav; \Xc^n) = \Delta_2 > 0$, let $\Mv_{JS0}$ and $\Mv_{JS}$ be the James--Stein Gaussian mechanisms with $\sigma>0$, as defined in Definition~\ref{def:JS}. For any data set $S \in \Xc^n$,
\begin{equation}\label{eq:M4cost}
        L_2(\Mv_{JS0}(S)) = \sigma^2  p - \sigma^2 (p-2)^2 \left\{ \frac{1}{2} e^{-\frac{1}{2}\tau(S)} \sum_{\beta = 0}^{\infty} \left( \frac{\tau(S) }{2} \right)^\beta \frac{1}{\beta! (\frac{1}{2}p + \beta - 1)} \right\},
\end{equation}
where $ \tau(S)  = {\|\thetav(S)\|_2^2}/{\sigma^2}$, and
    \begin{align} \label{eq:M5cost}
        L_2(\Mv_{JS}(S)) = \sigma^2 p - \sigma^2(p-3)^2 \left\{ \frac{1}{2} e^{-\frac{1}{2}\tau_{\1v, \perp}(S)} \sum_{\beta = 0}^{\infty} \left( \frac{\tau_{\1v, \perp}(S)}{2} \right)^\beta \frac{1}{\beta! \{\frac{1}{2}(p - 1) + \beta - 1\}} \right\} ,
    \end{align}
    where $\tau_{\1v, \perp}(S) =  {\|(\Id_p - \frac{1}{p} \1v_p \1v_p^\top)\thetav(S)\|_2^2}/{\sigma^2}$.
\end{prop}

Proposition~\ref{prop:M45_costs} allows us to directly compare the conditional $L_2$-costs of $\Mv_{JS0}$ and $\Mv_{JS}$ with the $L_2$-cost of the ordinary Gaussian mechanism $\Mv_G$, for each $S$.

\begin{thm}
\label{thm:JS-mech}
Let $\mu>0$. For a multivariate statistic $\thetav: \Xc^n \to \Real^p$ with $\Delta_2(\thetav; \Xc^n) = \Delta_2 > 0$, let $\Mv_G$ be the ordinary Gaussian mechanism, and $\Mv_{JS0}$  and $\Mv_{JS}$ be the James--Stein Gaussian mechanisms, all calibrated to satisfy $\mu$-GDP (i.e., with $\sigma = \Delta_2 / \mu$).
\begin{enumerate}
    \item[(i)] If $p \ge 3$, $\Mv_{JS0}$ strictly dominates $\Mv_G$.
    \item[(ii)] If $p \ge 4, \Mv_{JS}$ strictly dominates $\Mv_G$.
\end{enumerate}
\end{thm}

Unlike the ordinary and rank-deficient Gaussian mechanisms, $\Mv_G$ and $\Mv_r$, the conditional $L_2$-costs of $\Mv_{JS}$ and $\Mv_{JS0}$, given $S$, depend on the value of $\thetav(S)$. As a consequence, there is in general no simple expression for $L_2(\Mv_{JS}) = \E_S \left( L_2(\Mv_{JS})\right)$, which depends on the distribution of the statistic $\thetav(S)$. Nevertheless,
under the conditions of Theorem~\ref{thm:JS-mech}, $L_2(\Mv_{JS}) < L_2(\Mv_{G})$ holds for any distribution of $\thetav(S)$.
In Section~\ref{sec:app_conting_tab}, we numerically compare the unconditioned $L_2$-costs of various mechanisms, when $\thetav(S)$ is a random frequency table following a multinomial distribution.

\subsection{Rank-deficient James–Stein Gaussian mechanism}\label{sec:2.3}

We have demonstrated that the ordinary Gaussian mechanisms can be improved by either utilizing the structure of the sensitivity space $\Sc_{\thetav}$ of $\thetav$ or using a James--Stein shrinkage of mechanism outputs. In this subsection, we combine those two ideas and propose rank-deficient James–Stein Gaussian mechanisms, $\Mv_{rJS}$.  The mechanism $\Mv_{rJS}$ is obtained by first perturbing the statistic $\thetav(S)$ only for the span of the sensitivity space, then applying a James--Stein shrinkage for the lower dimensional subspace.

\begin{defn}\label{def:rJS}
Let $\thetav: \Xc^n \to \Real^p$ be a multivariate statistic, and
Let $\Uv_{\thetav}$ and $\Uv_1$ be orthonormal basis matrices for ${\rm span}(\Sc_{\thetav})$ and the orthogonal complement of ${\rm span}(\Sc_{\thetav})$, respectively, so that $[\Uv_1 \; \Uv_{\thetav}]$ is a $p \times p$ orthogonal matrix. The mechanism, perturbing $\thetav$,
\begin{equation}\label{eq:rJS}
 \Mv_{rJS}(S) = \Uv_1
                 \Uv_1^\top \thetav(S)  +
                \Uv_{\thetav}  m_a ( \Uv_{\thetav}^\top \Mv_r(S) ; \sigma)
,
\end{equation}
where $\Mv_r(S) = \thetav(S) + N_p(\0v, \sigma^2 \Uv_{\thetav}\Uv_{\thetav}^\top)$ is the rank-deficient Gaussian mechanism defined in (\ref{M_3}), and $m_a(\cdot ; \sigma)$ is the James--Stein shrinkage map (\ref{eq:shrink_to_ave}), is called a rank-deficient James--Stein mechanism.
\end{defn}

In the definition above, the mechanism $\Mv_{rJS}$ not only depends on $\sigma$ but also depends on the choice of orthonormal basis (the columns of $\Uv_{\thetav}$) for ${\rm span}(\Sc_{\thetav})$. As we shall see shortly, the level of privacy protection guaranteed by $\Mv_{rJS}$ is not affected by the choice of $\Uv_{\thetav}$, but the conditional $L_2$-cost of $\Mv_{rJS}(S)$ varies for different choices of $\Uv_{\thetav}$, even if $\thetav(S)$ is fixed.

For any choice of $\Uv_{\thetav}$ that does not directly depend on $\thetav(S)$ or $S$, $\Mv_{rJS}(S)$ depends on data set $S$ \emph{only} through $\Mv_r$. This can be seen by noting that  $\Uv_1^\top\Uv_{\thetav} = \0v$, which in turn leads that  $\Uv_1^\top \Mv_r(S) = \Uv_1^\top \thetav(S)$ for any data set $S$. Since $\Mv_{rJS}(S)$ is obtained via post-processing $\Mv_r$, if $\Mv_r$ is $\mu$-GDP for some $\mu>0$, so is $\Mv_{rJS}$. By Theorem~\ref{thm:rank-deficient-G-mech},  $\Mv_{rJS}$, with any choice of $\Uv_{\thetav}$, is $\mu$-GDP if $\sigma \ge \Delta_2(\thetav; \Xc^n)/\mu$.

On the other hand, the conditional $L_2$-cost of $\Mv_{rJS}$ depends on $\Uv_{\thetav}$. In particular, for any given $S$, $L_2(\Mv_{rJS}(S))$ is exactly Equation (\ref{eq:M5cost}) with $p$ replaced by $d_{\thetav} = {\rm dim}({\rm span}(\Sc_{\thetav}))$ and $\tau_{\1v, \perp}(S)$ replaced by
$\widetilde{\tau}(S) =
\| (\Id_{d_{\thetav}} - \frac{1}{d_{\thetav}} \1v_{d_{\thetav}} \1v_{d_{\thetav}}^\top) \Uv_{\thetav}^\top \thetav(S) \|_2^2 / {\sigma^2}$. (See Proposition~\ref{lem:M6_cost} in the appendix for the explicit expression of $L_2(\Mv_{rJS}(S))$.)
The next result is obtained by directly comparing $L_2(\Mv_{rJS}(S))$ with $L_2(\Mv_r(S)) = d_{\thetav} \sigma^2$, for all $S \in \Xc^n$.

\begin{thm}
\label{thm:rJS-mech}
Let $\mu>0$. Suppose $\thetav: \Xc^n \to \Real^p$ is a multivariate statistic with rank-deficient $\Sc_{\thetav}$, i.e. $d_{\thetav} < p$,  and  $\Delta_2(\thetav; \Xc^n) = \Delta_2 > 0$. Let $\Uv_{\thetav}$ (needed to define $\Mv_{rJS}$) consist of any orthonormal basis of ${\rm span}(\Sc_{\thetav})$.
Suppose $\Mv_G$, $\Mv_r$ and $\Mv_{rJS}$ are calibrated to satisfy $\mu$-GDP (i.e., with $\sigma = \Delta_2 / \mu$). Then, $\Mv_r$ strictly dominates $\Mv_G$.
  If $d_{\thetav} \ge 4$, $\Mv_{rJS}$ strictly dominates each of $\Mv_r$ and $\Mv_G$.
\end{thm}

We remark that one can try to choose a particular $\Uv_{\thetav}$ that minimizes $L_2(\Mv_{rJS}(S))$, but doing so requires using the \emph{unperturbed} information in $\| (\Id_{d_{\thetav}} - \frac{1}{d_{\thetav}} \1v_{d_{\thetav}} \1v_{d_{\thetav}}^\top) \Uv_{\thetav}^\top \thetav(S) \|_2$. Such a mechanism depends on $S$  through not only $\Mv_r(S)$ but also $\thetav(S)$, and determining the level of its privacy protection seems a difficult task.

It is natural to ask whether any mechanism among the James--Stein shrinkage mechanisms $\Mv_{JS0}$, $\Mv_{JS}$ and $\Mv_{rJS}$ dominates another. The answer is no. The conditional $L_2$-costs of these mechanisms depend on $\|\thetav(S)\|_2$,
$\|(\Id_p - \frac{1}{p} \1v_p \1v_p^\top)\thetav(S)\|_2$ or  $\| (\Id_{d_{\thetav}} - \frac{1}{d_{\thetav}} \1v_{d_{\thetav}} \1v_{d_{\thetav}}^\top) \Uv_{\thetav}^\top \thetav(S) \|_2$ (as well as $d_{\thetav}$). As a result, $L_2(\Mv_{JS0}(S)) <  L_2(\Mv_{JS}(S))$ is true for some $S$, but is false for other choices of $S$. Likewise, although  $L_2(\Mv_{JS}(S)) >  L_2(\Mv_{rJS}(S))$  for {most} data sets, it is not true for some $S$ (for some $\thetav$ and $\Xc^n$). To compare these mechanisms we simplify the situation by considering an asymptotic direction in which the sample size $n$ of data set $S$ increases, while the dimension $p$ of statistic $\thetav$ and the privacy level $\mu$ are fixed.

In Theorem \ref{thm:L2cost-conv} below, we consider any sequence of data sets $S_1,S_2,\ldots, S_n,\ldots$ of increasing sample sizes, which are considered deterministic. Note that the $L_2$-sensitivity of $\thetav$ may depend on the sample size $n$. For instance, for $\thetav(S) = \bar{\xv}_n$ with $\Xc = [0,1]^p$, $\Delta_2(\thetav; \Xc^n) = \sqrt{p}/n$.

\begin{thm} \label{thm:L2cost-conv}
Let $\mu>0$, and let $S_1,\ldots,S_n,\ldots$ be any sequence of data sets of increasing sample sizes such that $S_n \in \Xc^n$ for each $n$.
For each $n$, $\thetav: \Xc^n \to \Real^p$ is a multivariate statistic, and the randomized mechanisms $\Mv_G$, $\Mv_r$, $\Mv_{JS0}$, $\Mv_{JS}$ and $\Mv_{rJS}$ satisfy $\mu$-GDP; that is, $\sigma_n := \Delta_2(\thetav; \Xc^n)/ \mu$ is used for all mechanisms with sample size $n$. In the statements below, $\Uv_{\thetav}$ is any  basis matrix chosen for both $\Mv_{r}$ and $\Mv_{rJS}$.
\begin{itemize}
  \item[(i)] If $p \ge 3$ and $\|\thetav(S_n)\|_2 /\Delta_2(\thetav; \Xc^n) \to \infty$ as $n \to\infty$, then
    $$
    \E \left\{ \|\Mv_{JS0}(S_n) - \Mv_{G}(S_n)\|_2^2/\sigma_n^2 \right\}
    \to 0, \quad \mbox{as }~ n \to \infty.
    $$
  \item [(ii)] If $p \ge 4$ and $\|(\Id_p - \frac{1}{p} \1v_p \1v_p^\top)\thetav(S)\|_2/\Delta_2(\thetav; \Xc^n) \to \infty$ as $n \to\infty$, then
    $$
    \E \left\{ \|\Mv_{JS}(S_n) - \Mv_{G}(S_n)\|_2^2/\sigma_n^2 \right\} \to 0, \quad \mbox{as }~ n \to \infty.
    $$
  \item [(iii)] If $d_{\thetav} \ge 4$ and
  $\| (\Id_{d_{\thetav}} - \frac{1}{d_{\thetav}} \1v_{d_{\thetav}} \1v_{d_{\thetav}}^\top) \Uv_{\thetav}^\top \thetav(S) \|_2/\Delta_2(\thetav; \Xc^n) \to \infty$ as $n \to\infty$, then
      $$
    \mathbb E \left\{ \|\Mv_{rJS}(S_n) - \Mv_{r}(S_n)\|_2^2/\sigma_n^2 \right\} \to 0, \quad \mbox{as }~ n \to \infty.
    $$
\end{itemize}
\end{thm}

In Theorem~\ref{thm:L2cost-conv}, we have assumed that for any $S$, $\Mv_{JS0}(S)$ and $\Mv_{JS}(S)$ are obtained by post-processing $\Mv_G(S)$ (and $\Mv_{rJS}(S)$ by $\Mv_r(S)$). The conditions such as $\|\thetav(S)\|_2/\Delta_2(\thetav; \Xc^n) \to \infty$ hold for most applications. For example, when $\thetav(S) = \bar{\xv}_n$ with $\Xc = [0,1]^p$, $\|\thetav(S)\|_2 \le \sqrt{p}$ is bounded but
$\Delta_2(\thetav; \Xc^n) = \sqrt{p}/n$ decreases, and the condition is satisfied. Since for fixed $\mu$, $\sigma_n \asymp \Delta_2(\thetav; \Xc^n)$, Theorem~\ref{thm:L2cost-conv}(i) yields that
 $\E \|\Mv_{JS0}(S_n) - \Mv_{G}(S_n)\|_2^2 = o(n^{-2})$. As an another example, let $\thetav$ be the frequency table defined in Example \ref{ex:contingency-table}. For such cases, $\Delta_2(\thetav; \Xc^n) = \sqrt{2}$ is constant over varying $n$, but $\|\thetav(S)\|_2 \ge \sqrt{\|\thetav(S)\|_1} = \sqrt n$, thus the condition is satisfied. The difference between $\Mv_{JS0}(S_n)$ and $\Mv_{G}(S_n)$ gets smaller for larger $n$, when appropriately scaled.

From Theorem~\ref{thm:L2cost-conv}, we conclude that while each of $\Mv_{JS0}$ and $\Mv_{JS}$ has strictly smaller $L_2$-cost than the $L_2$-cost of $\Mv_G$ for each $n$, the difference becomes smaller for large $n$. Moreover, we can compare $\Mv_{rJS}$ with either  $\Mv_{JS0}$ or $\Mv_{JS}$ in the asymptotic setting, as follows.
\begin{cor}\label{cor:tothm:L2cost-conv} Under the conditions of Theorem~\ref{thm:L2cost-conv}, assume further that $\Sc_{\thetav}$ is rank-deficient, i.e., $d_{\thetav} < p$. Then, for any sequence $S_n$, and for any basis matrix $\Uv_{\thetav}$ chosen for $\Mv_{rJS}$,
\begin{equation*}\label{eq:cortoL2cost-cov}
  \limsup_{n\to\infty} {L_2(\Mv_{rJS}(S_n))}/{\sigma_n^2} < \liminf_{n\to\infty} {L_2(\Mv_{JS}(S_n))}/{\sigma_n^2},
\end{equation*}
which also holds if $ \Mv_{JS}$ is replaced by $\Mv_{JS0}$.
\end{cor}

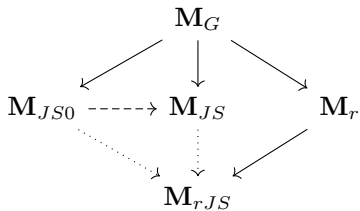
\begin{figure}
    \centering
    \begin{tikzcd}
     & \Mv_G \arrow[ld] \arrow[d] \arrow[rd] & \\
    \Mv_{JS0} \arrow[r, dashed] \arrow[rd, dotted] & \Mv_{JS} \arrow[d, dotted] & \Mv_r \arrow[ld] \\
     & \Mv_{rJS} &
    \end{tikzcd}
    \caption{Notation $\Mv \longrightarrow \Mv'$ is used when the statistical utility of $\Mv'$ is greater than that of $\Mv$  (under appropriate conditions) for finite sample sizes. The dotted arrow is used when the comparison is made for large sample sizes, and the dashed arrow is used when the comparison was made empirically; see Section  \ref{sec:app_conting_tab}.}
    \label{fig:relationship}
\end{figure}

Figure \ref{fig:relationship} summarizes our discussion in this section. Overall, $\Mv_{rJS}$ shows the largest statistical utility (or the smallest $L_2$-cost) among the Gaussian mechanisms we have considered. We have also considered a rank-deficient James--Stein Gaussian mechanism, \emph{shrinking to zero} (given by replacing $m_a$ with $m_0$ in \eqref{eq:rJS}), but because its empirical $L_2$-cost is typically larger than that of $\Mv_{rJS}$, we have chosen to only discuss $\Mv_{rJS}$. Other forms of post-processing can be utilized to further reduce the mechanisms' $L_2$-cost. See Section~\ref{sec:contingency_table_release} for an example.

\section{Calibration of multivariate Laplace mechanisms for \texorpdfstring{$\mu$}{mu}-GDP}
\label{sec:Lap_mech}
 Laplace mechanisms are additive mechanisms perturbing a statistic $\thetav$ with a Laplace-distributed noise, and have been commonly used as a general-purpose mechanism satisfying $\varepsilon$-DP \citep{dwork2014algorithmic}. Let $\mathrm{Lap}(\delta, b)$ stand for the Laplace distribution with location parameter $\delta \in \Real$ and scale parameter $b > 0$. For a random vector $\Xv = (X_1, \dots, X_p)$, we write $\Xv \sim \Lap_p(\deltav, b)$,
where $\deltav = (\delta_1, \cdots, \delta_p)^\top \in \Real^p$, if $X_i$ follows $\mathrm{Lap}(\delta_i, b)$ independently for $i = 1,\ldots,p$. For a multivariate statistic $\thetav$, the additive mechanism that adds a $\Lap_p(\0v, b)$-distributed noise,
\begin{equation}\label{eq:laplace_mechanism}
 \Mv_{\Lap}(S; b) = \thetav(S) +\Lap_p(\0v, b),
\end{equation}
is called multivariate Laplace mechanism.  It is well-known that a multivariate Laplace mechanism  $\Mv_{\Lap}(\cdot; b)$ with $b \ge \Delta_1 / \varepsilon$ satisfies $\varepsilon$-DP \citep{dwork2006calibrating}, where $\Delta_1 = \Delta_1(\thetav ; \Xc^n)$ is the $L_1$-sensitivity of $\thetav$ (see (\ref{eq:Lr-sensitivity})).

In this section, we calibrate the amount of noise in multivariate Laplace mechanisms with respect to $\mu$-GDP, for a fair comparison with Gaussian mechanisms.
We show that, in contrast to the case of Gaussian mechanisms, the optimally-calibrated scale parameter $b$ for \emph{univariate} Laplace mechanisms, with respect to $\mu$-GDP criterion, turns out to be sub-optimal when used for some \emph{multivariate} Laplace mechanisms. An optimal calibration of the amount of Laplace noise is achieved for the special case that $\thetav$ is a multivariate frequency table.

\subsection{Universal conversion and improved Laplace mechanisms}\label{subsec:Lap_univ_conv}

A naive approach to determining the level of privacy protection with respect to GDP criterion in a Laplace mechanism $\Mv_{\Lap}(\cdot; b)$  is to use the fact that $\Mv_{\Lap}(\cdot; b)$ is $\varepsilon$-DP for any $\varepsilon \ge \Delta_1 / b$, and then to determine $\mu(\varepsilon)$ such that any $\varepsilon$-DP mechanism satisfies $\mu(\varepsilon)$-GDP.  Such a conversion from $\varepsilon$-DP to $\mu$-GDP can be obtained by comparing the trade-off function  $T_{\varepsilon \mbox{-} \mathrm{DP}}$ of $\varepsilon$-DP, where
\begin{equation}\label{eq:trade_off_eps_del}
    T_{\varepsilon \mbox{-} \mathrm{DP}}(\alpha) =
    \begin{cases}
        1- e^\varepsilon \alpha & \mbox{if }\: 0 \le \alpha < \frac{1}{1 + e^{\varepsilon}}, \\
        e^{-\varepsilon}(1-\alpha) & \mbox{otherwise},
    \end{cases}
\end{equation}
with the trade-off function $G_\mu$ (Equation~(\ref{eq:Gmu})) corresponding to $\mu$-GDP. %

\begin{lem}[Universal conversion] \label{lem:eps-DP_to_mu-GDP_conversion}
(i) Let $\Mv$ be a mechanism satisfying $\varepsilon$-DP, for $\varepsilon>0$. Then $\Mv$ is $\mu$-GDP for any $\mu \ge 2\Phi^{-1}\{e^\epsilon / (e^\epsilon + 1)\}$, where $\Phi$ is the standard normal distribution function.

(ii) Let $\mu >0$ be given. For any $\varepsilon \le \varepsilon_{\mu}:= \log\{\Phi(\mu/2)/\Phi(-\mu/2)\}$, a mechanism satisfying $\varepsilon$-DP is $\mu$-GDP.
\end{lem}

The universal conversion in Lemma \ref{lem:eps-DP_to_mu-GDP_conversion} can be used in determining the level $\mu$ of GDP criterion for any mechanism that is known to be $\varepsilon$-DP for some $\varepsilon>0$. In particular, a Laplace mechanism $\Mv_{\Lap}(\cdot; b)$, satisfying $(\Delta_1 / b)$-DP, is $\mu$-GDP for any $\mu \ge 2 \Phi^{-1}\{e^{\Delta_1 / b} / (e^{\Delta_1 / b} + 1)\}$.
Conversely, if a desired level $\mu >0 $ is prespecified, then choosing the scale parameter $b$ of $\Mv_{\Lap}(\cdot; b)$ to be %
$$b_{\mu}^{\varepsilon} := \Delta_1 / \varepsilon_{\mu} = \Delta_1 / \log\{\Phi(\mu/2)/\Phi(-\mu/2)\},$$  guarantees that $\Mv_{\Lap}(\cdot ; b)$ is $\mu$-GDP, since
\begin{equation}\label{eq:conversion_universal_lap}
  \inf_{S \sim S'} T(\Mv_{\Lap}(S; b_{\mu}^{\varepsilon}),
\Mv_{\Lap}(S'; b_{\mu}^{\varepsilon}))
\ge
T_{\varepsilon_{\mu} \mbox{-} \mathrm{DP}} \ge G_{\mu},
\end{equation}
in which the first inequality is given by the fact that $\Mv_{\Lap}(\cdot ; b)$ is $(\Delta_1 /b)$-DP, and the second inequality is from Lemma~\ref{lem:eps-DP_to_mu-GDP_conversion}(ii).

\begin{figure}
    \centering
    \begin{subfigure}{0.45\textwidth}
        \centering
        \includegraphics[width=.95\textwidth]{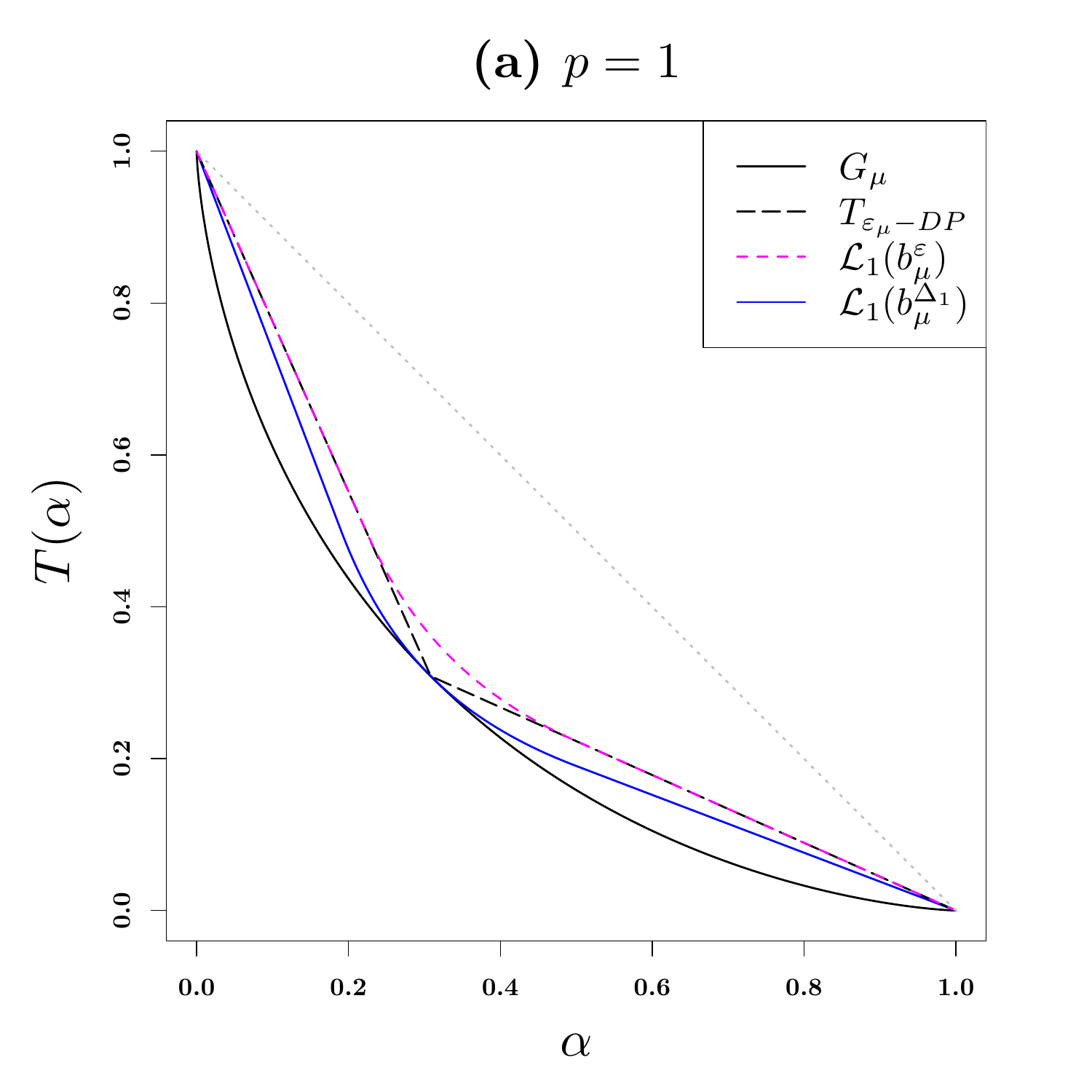}
        \label{fig:lap_tf_1dim}
    \end{subfigure}
    \begin{subfigure}{0.45\textwidth}
        \centering
        \includegraphics[width=.95\textwidth]{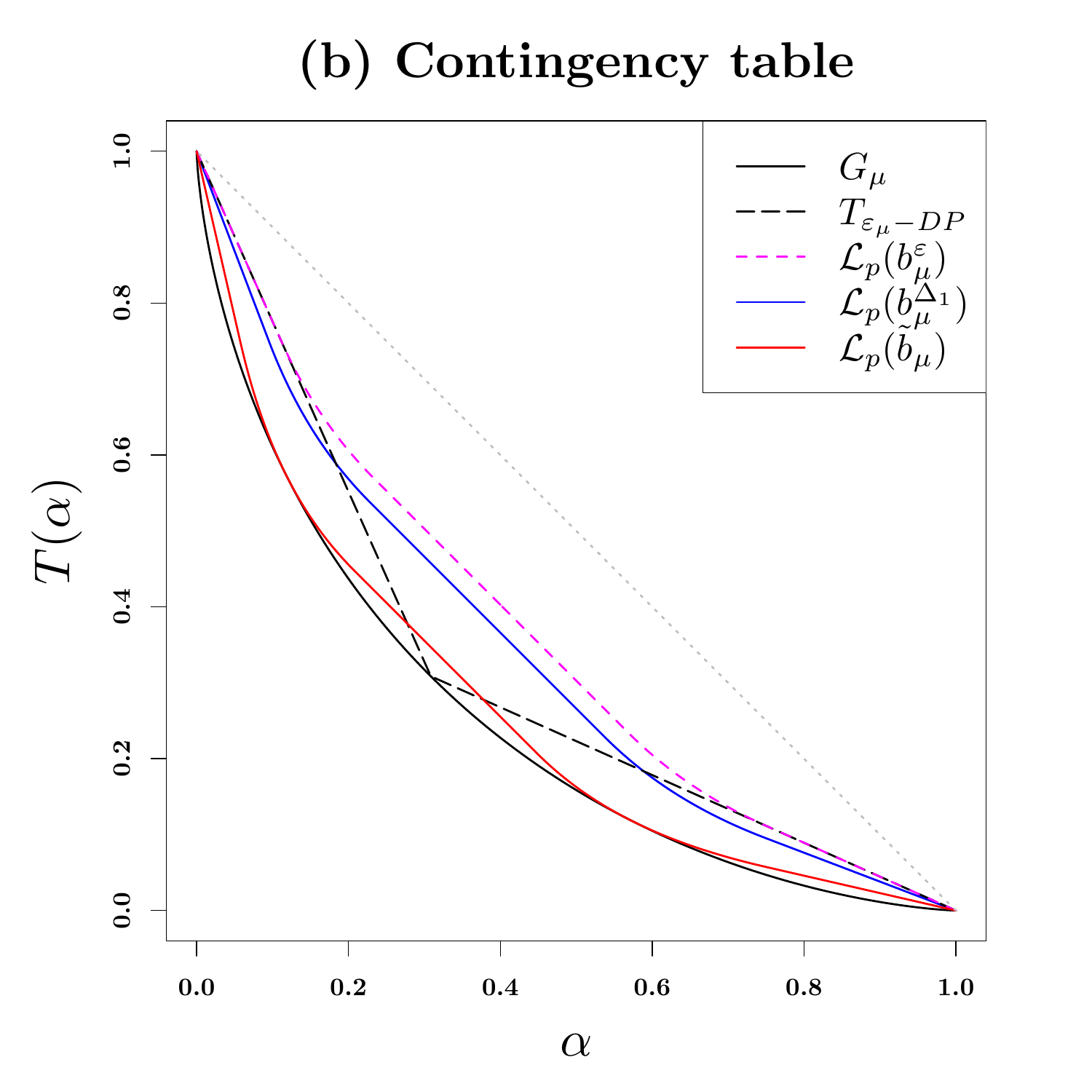}
        \label{fig:lap_tf_table}
    \end{subfigure}
    \caption{%
     The graphs of trade-off functions discussed in Section~\ref{sec:Lap_mech} for $\mu = 1$.
     The Laplace mechanism $\Mv_{\Lap}(\cdot; b_\mu^{\Delta_1})$ is optimal in the univariate case (a), but is  sub-optimal   for a multivariate case (b), in which  $\Mv_{\Lap}(\cdot; \tilde{b}_\mu)$ is optimally calibrated to $\mu$-GDP. 
    }
    \label{fig:lap_tf}
\end{figure}

Such a two-step calibration of $b$ (by $b_{\mu}^{\varepsilon}$) for Laplace mechanisms is by no means optimal. As we shall see shortly, one can choose the scale parameter $b$ strictly smaller than $b_{\mu}^{\varepsilon}$, while still satisfying $\mu$-GDP. Intuitively, this is because the equalities in  (\ref{eq:conversion_universal_lap}) are not simultaneously attained.
To visually confirm this fact, the three trade-off functions in (\ref{eq:conversion_universal_lap}) are plotted in Figure \ref{fig:lap_tf}, for the univariate case  on the left panel and for a multivariate case on the right panel.
Let
\begin{equation}\label{eq:L_p}
  \Lc_p(b) : = \Lc_p(b; \thetav, \Xc^n) = \inf_{S \sim S'} T(\Mv_{\Lap}(S; b),
\Mv_{\Lap}(S'; b ))
\end{equation}
 denote the trade-off function representing the worst-case Laplace trade-offs, where $p$ is the dimension of the statistic $\thetav(S) \in \Real^p$.
It can be observed in Figure \ref{fig:lap_tf} that
there exists a unique $\alpha_0 \in (0,1)$ such that
$T_{\varepsilon_{\mu} \mbox{-} \mathrm{DP}}(\alpha_0) = G_{\mu}(\alpha_0)$, but  $ \Lc_p( b_{\mu}^{\varepsilon})(\alpha_0) > T_{\varepsilon_{\mu} \mbox{-} \mathrm{DP}}(\alpha_0)$ (for both univariate and multivariate cases).

In the two-step calibration above, we have used the $L_1$-sensitivity $\Delta_1(\thetav; \Xc^n)$ of ${\thetav}$ only for calibration with respect to the $\varepsilon$-DP criterion. Naturally, an improvement over the two-step calibration can be obtained by directly using the $L_1$-sensitivity with respect to GDP criterion. We discuss the univariate case (i.e., the case where $p = 1$) first, and then move on to the multivariate case.

\subsubsection{Calibration of univariate Laplace mechanisms}\label{sec:Lap_univariate_sub}
Let $p = 1$. 
For this univariate case, 
we write $\theta$ and $M_{\rm Lap}(\cdot ; b)$ for $\thetav$ and $\Mv_{\rm Lap}(\cdot ; b)$. In calibration of the scale parameter $b$ of $M_{\rm Lap}(\cdot ; b)$ to satisfy $\mu$-GDP (for a given $\mu>0$), we find the smallest $b$ satisfying
$$ \Lc_1(b) \ge G_\mu,$$
where the trade-off function $\Lc_1(b)$ is defined in (\ref{eq:L_p}).
The expression of the trade-off function $\Lc_1(b)$ is useful for such purpose.

\begin{lem}\label{lem:univariate_laplace_tradeoff} For any $b>0$ and $\Delta_1 = \Delta_1(\theta; \Xc^n) = \sup_{S \sim S'} | \theta(S) - \theta(S')|$, the value of the trade-off function $\Lc_1(b):[0,1] \to [0,1]$ at $\alpha \in [0,1]$ is, for $\delta = \Delta_1 / b$,
\begin{align}\label{eq:Lap_tradeoff_1dim}
    \Lc_1(b)(\alpha) =
    \begin{cases}
      1 - e^\delta \alpha, & \text{if $\alpha < \frac{1}{2}e^{-\delta}$};\\
      \frac{e^{-\delta}}{4\alpha}, & \text{if $\frac{1}{2}e^{-\delta} \leq \alpha < \frac{1}{2}$};\\
      e^{-\delta}(1-\alpha), & \text{otherwise.}
    \end{cases}
\end{align}
\end{lem} 
Using the exact expressions of $\Lc_1(b)$ and $G_\mu$  given in (\ref{eq:Lap_tradeoff_1dim}) and (\ref{eq:Gmu}) respectively, one can directly compare $\Lc_1(b)(\alpha)$ with $G_\mu(\alpha)$ at $\alpha \in (0,1)$. Such comparison leads to the following result on the \emph{optimal} calibration of univariate Laplace mechanism for $\mu$-GDP. For $\mu > 0$ and $\Delta_1 > 0$, define
$$b_\mu^{\Delta_1}
     := {\Delta_1}/[{-2\log\{2\Phi(-\tfrac{\mu}{2})\}}].$$

\begin{prop}
\label{prop:1dim_lap_stronger}
Let $\theta: \Xc^n \to \Real$ be any univariate statistic with   $\Delta_1(\theta; \Xc^n) = \Delta_1$, and $\Delta_1 \ge 0$.
\begin{itemize}
  \item[(i)] For any $\mu > 0$, the Laplace mechanism with scale parameter $b$, $M_{\rm Lap}(\cdot ; b)$, is $\mu$-GDP if and only if
      $  b \ge b_\mu^{\Delta_1}.$ Moreover, $b_\mu^{\Delta_1} < b_\mu^\varepsilon$.
  \item[(ii)] Conversely, for any $b > 0$, $M_{\rm Lap}(\cdot ; b)$ is $\mu$-GDP if and only if $\mu \ge -2\Phi^{-1}\left\{ \frac{1}{2}{\exp(\frac{-\Delta_1}{2b})} \right\}$.
\end{itemize}
 \end{prop}

Proposition~\ref{prop:1dim_lap_stronger} tells us that $b_\mu^{\Delta_1} = \inf\{ b >0 : \Lc_1(b) \ge G_{\mu} \}$, and as a result, we may say that a univariate Laplace mechanism with scale parameter $b_\mu^{\Delta_1}$ is optimally calibrated to satisfy $\mu$-GDP. As visualized in Figure~\ref{fig:lap_tf}(a), the graph of $\Lc_1(b_\mu^{\Delta_1})$ sits just above the graph of $G_\mu$; any $b < b_\mu^{\Delta_1}$ results in  $\Lc_1(b)(\alpha) < G_\mu(\alpha)$ for some $\alpha\in (0,1)$.

 Recall that in Lemma~\ref{lem:iso_Gauss_stronger} we have established a result on an optimal calibration for a class of \emph{multivariate} Gaussian mechanisms. Since Proposition~\ref{prop:1dim_lap_stronger} above states an optimality result for univariate Laplace mechanisms similar to the conclusion of Lemma~\ref{lem:iso_Gauss_stronger}, one may conjecture that the above result generalizes to the case of multivariate Laplace mechanisms. This conjecture is \emph{partly} true, as we discuss in the next subsection.

\subsubsection{Calibration of multivariate Laplace mechanisms} \label{sec:Lap_multivariate_sub}

A simple-minded approach of calibrating multivariate Laplace mechanism, $\Mv_{\rm Lap}(\cdot ; b)$, for $p \ge 2$ is to directly compare the exact expression for the trade-off function $\Lc_p(b; \thetav, \Xc^n)$ with $G_\mu$, as done for the $p = 1$ case in Section~\ref{sec:Lap_univariate_sub}. 
Unfortunately, the exact expression for the trade-off function $\Lc_p(b; \thetav, \Xc^n)$ is extremely dependent on the shape of the sensitivity space $\Sc_{\thetav}$, and turns out to be quite messy even for $p = 2$. 
Evaluating $\Lc_p(b; \thetav, \Xc^n) = \inf_{S \sim S'} T(\Mv_{\Lap}(S; b),
\Mv_{\Lap}(S'; b ))$ requires an evaluation of the trade-off function   between two multivariate Laplace distributions,
since for any $S,S' \in \Xc^n$, $b>0$,
\begin{align}
 T(\Mv_{\Lap}(S; b), \Mv_{\Lap}(S'; b ))  &  = T({\rm Lap}_p(\thetav(S), b), {\rm Lap}_p(\thetav(S'),b)) \nonumber \\
   & = T({\rm Lap}_p(\0v, 1), {\rm Lap}_p(\deltav,1)),  \label{eq:LapMulti}
\end{align}
where  $\deltav = \tfrac{1}{b}(\thetav(S) - \thetav(S'))$, due to the invariance of trade-off functions under simultaneous affine transformations (see Lemma~\ref{lem:Tf-linearity}). A key observation that can be used  for further simplifying (\ref{eq:LapMulti}) is the following.
\begin{lem} \label{lem-high2one}
    Let $p \ge 1$. For any $\deltav = (\delta_1, \dots, \delta_p) \in \Real^p$, we have
          \begin{equation}\label{eq:trade-off_multiLap}
       T(\Lap_p(\0v_p, 1), \Lap_p(\deltav, 1))
    \ge T(\Lap(0, 1), \Lap(\|\deltav\|_1, 1)).
      \end{equation}
    \begin{itemize}
      \item[(i)] If the number of nonzero elements in $\deltav$ is at most one, then the equality in (\ref{eq:trade-off_multiLap}) holds for all $\alpha \in (0,1)$.
      \item[(ii)] If at least two elements in $\deltav$ are nonzero, then strict inequality  in (\ref{eq:trade-off_multiLap}) holds for some $\alpha \in (0,1)$. That is,
      there exists an interval $(a,b)$, $0<a<b<1$, such that for any $\alpha \in(a,b)$, $T(\Lap_p(\0v_p, 1), \Lap_p(\deltav, 1)) (\alpha) > T(\Lap(0, 1), \Lap(\|\deltav\|_1, 1)) (\alpha)$.
    Moreover, $(a, b)$ contains $\alpha_0 = \tfrac{1}{2}e^{-\|\deltav\|_1/2}$ which satisfies that $T(\Lap(0, 1), \Lap(\|\deltav\|_1, 1))(\alpha_0) = \alpha_0$.
    \end{itemize}
\end{lem}

\begin{remark}
The trade-off function between multivariate \emph{Gaussian} distributions can be represented exactly by a corresponding trade-off function between univariate Gaussian distributions: For any $p\ge 1$, and for any $\deltav \in \Real^p$,
 $T(N_p(\0v_p, 1), N_p(\deltav, 1))
    = T(N(0, 1), N(\|\deltav\|_2, 1))$ (see Lemma~\ref{lem:mv-Tf-normal}). In contrast, the trade-off function between multivariate Laplace distribution can only be lower-bounded by a trade-off function between univariate Laplace distributions, as stated in Lemma~\ref{lem-high2one}.
\end{remark}

As a result of Lemma~\ref{lem-high2one} and using (\ref{eq:LapMulti}), the trade-off function $\Lc_p(b; \thetav, \Xc^n)$ of a multivariate Laplace mechanism can be lower-bounded by a trade-off function between \emph{univariate} Laplace distributions, as shown in the equations below:
\begin{align*}
\Lc_p(b; \thetav, \Xc^n)  & =  \inf_{S \sim S'} T(\Mv_{\Lap}(S; b),
\Mv_{\Lap}(S'; b ))\\
   &  \ge \inf_{S \sim S'} T(\Lap(0,1),\Lap(\|\thetav(S) - \thetav(S')\|_1/b,1)) \\
   & = T(\Lap(0,1),\Lap( \sup_{S \sim S'}\|\thetav(S) - \thetav(S')\|_1/b,1))\\
   & = T(\Lap(0,1),\Lap( {\Delta_1(\thetav, \Xc^n)}/{b},1)).
\end{align*}

The above observation allows us to generalize the optimal calibration result on univariate Laplace mechanisms  in Proposition~\ref{prop:1dim_lap_stronger} to multivariate cases. Recall that for $\mu > 0$ and $\Delta_1 > 0$,
$b_\mu^{\Delta_1}
     = {\Delta_1}/[{-2\log\{2\Phi(-\tfrac{\mu}{2})\}}].$

\begin{thm} \label{thm:improved_lap_mech}
Let $p \ge 1$, and $\Delta_1 > 0$.
\begin{itemize}
  \item[(i)] For any $\mu > 0$, the multivariate Laplace mechanism with scale parameter $b$, $\Mv_{\Lap}(\cdot ; b)$, is $\mu$-GDP for any $\thetav$ with $\Delta_1(\thetav, \Xc^n) \le \Delta_1$ if and only if
      $  b \ge b_\mu^{\Delta_1}.$
  \item[(ii)] Conversely, for any $b > 0$, $\Mv_{\rm Lap}(\cdot ; b)$ is $\mu$-GDP for any   $\thetav$ with $\Delta_1(\thetav, \Xc^n) \le \Delta_1$ if and only if $\mu \ge -2\Phi^{-1}\left\{ \frac{1}{2}{\exp(\frac{-\Delta_1}{2b})} \right\} $.
  \item[(iii)] Suppose $p \ge 2$. Then for any given $\thetav$ with $\Delta_1(\thetav, \Xc^n) \le \Delta_1$, $\Mv_{\Lap}(\cdot ; b)$, is $\mu$-GDP if  $  b \ge b_\mu^{\Delta_1}.$ However, there exists $\thetav$ such that $\Delta_1(\theta, \Xc^n) \le \Delta_1$ and $\Mv_{\Lap}(\cdot ; b)$ is $\mu$-GDP for some $b < b_\mu^{\Delta_1}.$
\end{itemize}
\end{thm}

In Theorem~\ref{thm:improved_lap_mech} above, parts (i) and (ii) provide an optimal calibration of Laplace mechanism (with respect to GDP criterion) when the \emph{only} information we have on the multivariate statistic $\thetav$ is the $L_1$-sensitivity. Part (iii), when compared to Proposition~\ref{prop:1dim_lap_stronger}, reveals a subtle difference between the calibrations of univariate and multivariate Laplace mechanisms. Simply put, the Laplace mechanism $\Mv_{\rm Lap}(\cdot ; b_\mu^{\Delta_1})$ is optimal in the univariate case, but is sub-optimal for some multivariate statistics $\thetav$. This conclusion can be visually verified in Figure ~\ref{fig:lap_tf}(b), drawn for the case where the multivariate statistic $\thetav$ is a frequency table: The trade-off function $\Lc_p(b_\mu^{\Delta_1})$
is strictly larger than $G_\mu$, suggesting that calibrating $b$ smaller than $b_\mu^{\Delta_1}$ is possible.

Whether the multivariate Laplace mechanism $\Mv_{\rm Lap}(\cdot ; b_\mu^{\Delta_1})$ is optimal or not depends on the statistic $\thetav$, and it is natural to investigate the sensitivity space $\Sc_{\thetav}$ for further improvement. 
In the next subsection, we take a full advantage of the sensitivity space in calibrating the Laplace mechanisms for the case where $\thetav$ is a frequency table.

\subsection{Laplace mechanisms for frequency tables}
\label{sec:lap_mech_table}

In this subsection, we consider the case where the multivariate statistic $\thetav$ is a frequency table. As shown in Example~\ref{ex:contingency-table}, the sensitivity space of $\thetav$ is
$$\Sc_{\thetav} = \{\thetav(S) - \thetav(S'): S, S'
\in \Xc^n, S\sim S'\} = \{\ev_i - \ev_j : i,j = 1,\ldots, p \}.$$
Since any element in $\Sc_{\thetav}$ involves only two coordinates of $\Real^p$, for any $p$, the trade-off function $\Lc_p(b; \thetav, \Xc^n)$ of a multivariate Laplace mechanism perturbing the frequency table $\thetav(S)$ with Laplace scale parameter $b$ is exactly a trade-off function between \emph{two-dimensional} Laplace distributions. In particular, if $\thetav$ is a frequency table, then for any $b>0$,
\begin{align}
 \Lc_p(b; \thetav, \Xc^n)  & =  \inf_{S \sim S'} T(\Mv_{\Lap}(S; b),
\Mv_{\Lap}(S'; b )) \nonumber \\
   & = \inf_{S \sim S'} T(\Lap_p(\0v,b), \Lap_p ( \thetav(S) - \thetav(S'), b)) \label{eq:freq_1} \\
   & = T(\Lap_p(\0v,b), \Lap_p ( \ev_1 - \ev_2 , b))\nonumber \\
   & = T(\Lap_2(\0v,b), \Lap_2 ( \1v , b)), \nonumber
\end{align}
since taking $\thetav(S) - \thetav(S') = \ev_1 - \ev_2 = (1,-1,0,\ldots,0)^\top \in \Sc_{\thetav}$ achieves the minimum of (\ref{eq:freq_1}).
To distinguish the general trade-off function $\Lc_p(b; \thetav, \Xc^n)$ (used for any $\thetav$) from the special case of frequency tables, write $\Lc_{\rm freq}(b) =  \Lc_p(b; \thetav_0, \Xc^n)= T(\Lap_2(\0v,b), \Lap_2 ( \1v , b))$ when  $\thetav_0$ is a frequency table.

An explicit expression of $\Lc_{\rm freq}(b)$ can be obtained by the Neyman--Pearson lemma: For $0 < \alpha \le \frac{1}{4}e^{-\frac{1}{b}} \left( 2 + \frac{1}{b} \right)$,
\begin{align}\label{eq:Lap_tradeoff_commondelta}
    \Lc_{\rm freq}(b) (\alpha) =
    \begin{cases}
      -e^{\frac{2}{b}}\alpha + 1
      & \text{if $0 \leq \alpha < \frac{1}{4}e^{-\frac{2}{b}}$},\\
      \frac{1}{4}e^{-C_b(\alpha)} \left( 3+C_b(\alpha) \right)
      & \text{if $\frac{1}{4}e^{-\frac{2}{b}} \leq \alpha < \frac{1}{4}e^{-\frac{1}{b}} \left( 1 + \frac{1}{b} \right)$},\\
      -\alpha + \frac{1}{2}e^{-\frac{1}{b}}(2+\frac{1}{b})
      &\text{if $\frac{1}{4}e^{-\frac{1}{b}}(1 + \frac{1}{b}) \leq \alpha \le \frac{1}{4}e^{-\frac{1}{b}}(2 + \frac{1}{b})$},
    \end{cases}
\end{align}
and $\Lc_{\rm freq}(b) (\alpha)$ for $\alpha \ge \frac{1}{4}e^{-\frac{1}{b}} \left( 2 + \frac{1}{b} \right)$ is given by the symmetry about the line $x=y$.
In (\ref{eq:Lap_tradeoff_commondelta}), $C_b(\alpha)  = W(4 \alpha e^{1 + \frac{2}{b}}) - 1$, where $W$ is the Lambert $W$ function \citep{lambert1758observationes} restricted to $[0,\infty)$, so that $W$ is the inverse of the function $w \mapsto we^w$ \citep[Section 4.13,][]{olver2010nist}.
An expression for the trade-off function between more general two-dimensional Laplace distributions, $T(\Lap_2(\deltav_1,b),  \Lap_2 ( \deltav_2 , b))$, for any $\deltav_1, \deltav_2 \in \Real^2$, can be found in Appendix \ref{sec:appendix_proofs1}.

Using the expression (\ref{eq:Lap_tradeoff_commondelta}), one can directly compare $\Lc_{\rm freq}(b)(\alpha)$ with $G_\mu(\alpha)$ at $\alpha \in (0,1)$. For $\mu > 0$, let
\begin{equation} \label{eq:bmu-contingency}
    \tilde b_{\mu} = \inf{\{b>0 \mid \Lc_{\rm freq}(b) \geq G_{\mu}\}}.
\end{equation}

Choosing the scale parameter $b$ of multivariate Laplace mechanisms to be $b = \tilde b_{\mu}$ is optimal when $\thetav$ is a frequency table, in the sense of the following theorem.
\begin{thm} \label{thm:...}
Let $p \ge 2$, and suppose that the multivariate statistic $\thetav: \Xc^n \to \Real^p$ is a frequency table, defined in Example~\ref{ex:contingency-table}. The multivariate Laplace mechanisms $\Mv_{\Lap}(\cdot; b)$ is $\mu$-GDP if and only if $b \ge \tilde{b}_{\mu}$.
\end{thm}

The optimality of $\Mv_{\Lap}(\cdot; \tilde b_{\mu})$ for frequency tables can be visually confirmed in Figure~\ref{fig:lap_tf}, in which the trade-off function $\Lc_{\rm freq}(\tilde b_{\mu}) = \Lc_p( \tilde b_{\mu})$ is directly above $G_\mu$.

The optimal scale parameter $ \tilde b_{\mu}$ in (\ref{eq:bmu-contingency}) has no closed-form expression, and is only evaluated numerically. We propose a bisection algorithm of numerically approximating $\tilde b_{\mu}$ for given $\mu>0$; see Algorithm \ref{alg:bisection}.  A key observation in designing such an algorithm is that $\tilde b_{\mu}$ is the unique root of the equation $g_\mu(b) = 0$, where $g_\mu(b) = \min_{\alpha \in I_b} \left\{ \Lc_{\rm freq}(b)(\alpha) - G_{\mu}(\alpha) \right\} = 0$, $I_b = [ \frac{1}{4}e^{-\frac{2}{b}} , \frac{1}{4} e^{-\frac{1}{b}} (1 + \frac{1}{b}) ]$. This is valid since trade-off functions are symmetric about the line $x = y$ and the function $\Lc_{\rm freq}(b)$ is linear on the complement of $I_b$.

\begin{algorithm}[tb]
	\SetAlgoLined\DontPrintSemicolon
	\KwResult{Return numerically approximated $ \tilde b_{\mu} $.}
	Initialize $ b^{U} = b_{\mu}^{\varepsilon} $ and $ b^{L} = 0 $\;
	\While{until the stopping criterion is met}{
		$ b \gets (b^U + b^L)/2 $\;
		$ l \gets \min_{\alpha \in I} (\Lc_{\rm freq}(b)(\alpha) - G_{\mu}(\alpha)) $
		where
		$I = \left[\frac{1}{4}e^{-\frac{2}{b}},	\frac{1}{4} e^{-\frac{1}{b}} (1 + \frac{1}{b}) \right]$\;
		\uIf{$ \vert l \vert < \epsilon_{tol} $}{
			Stopping criterion is met\;
			\Return $ b $\;
		}
		
		\uElseIf{$l > 0$}{
			$ b^U \gets b $\;
		}

		\ElseIf{$l < 0$}{
			$ b^L \gets b $\;
		}
	}
	\caption{Bisection algorithm for finding $ \tilde b_{\mu}$.}\label{alg:bisection}
\end{algorithm}

\section{Comparison between Gaussian and Laplace mechanisms}\label{sec:comparison}

In this section, we compare the statistical utilities of $\mu$-GDP Gaussian and Laplace mechanisms, both perturbing a multivariate statistic $\thetav$. For a fair comparison, we use the $L_r$-cost for all $r \ge 1$, as opposed to just using the $L_2$-cost, as done for comparison among Gaussian mechanisms.

We first compare the ordinary Gaussian mechanism $\Mv_G$ (\ref{eq:gaussianMech_again}) and the multivariate Laplace mechanism $\Mv_{\Lap}(\cdot, b_{\mu}^{\Delta_1})$ (\ref{eq:laplace_mechanism}), since for any $\mu>0$  and for any statistic $\thetav$, both mechanisms are $\mu$-GDP (see Lemma~\ref{lem:iso_Gauss_stronger} and Theorem \ref{thm:improved_lap_mech}). Write $\Mv_{\Lap}$ for $\Mv_{\Lap}(\cdot, b_{\mu}^{\Delta_1})$ for notational simplicity.

Since the conditional $L_r$-costs of $\Mv_G(S)$  and $\Mv_{\Lap}(S)$ do not depend on the data set $S$, we can say that $\Mv_G$ strictly dominates  $\Mv_{\Lap}$ in the $\ell_r$ sense if $L_r(\Mv_{\Lap}) > L_r(\Mv_G)$. To investigate under which conditions the Gaussian mechanism strictly dominates the Laplace mechanism, the relative efficiency
\begin{equation} \label{eq:L_r-cost-ratio}
{\rm RE}_r(\Mv_G , \Mv_{\Lap} ) := \frac{L_r(\Mv_{\Lap})}{L_r(\Mv_G)}
= \left(\frac{\Delta_1}{\Delta_2}\right)^r  2^{\frac{r}{2}} \Gamma\left(\frac{r}{2} + 1\right) \left(\frac{\mu}{-2 \log (2 \Phi (-\frac{\mu}{2}))}\right)^r,
\end{equation}
where $\Delta_r = \Delta_r(\thetav; \Xc^n)$ for $r = 1,2$, can be used. In particular, the Gaussian mechanism dominates the Laplace mechanism if ${\rm RE}_r(\Mv_G , \Mv_{\Lap} ) > 1$. Observe from (\ref{eq:L_r-cost-ratio}) that the relative efficiency depends on  three parameters: privacy parameter $\mu>0$, the ratio of sensitivities $\Delta_1/ \Delta_2\ge 1$, and $r \ge 1$. The dimension $p$ affects the relative efficiency  only through the ratio $\Delta_1/ \Delta_2$.
For each given $r$, ${\rm RE}_r(\Mv_G , \Mv_{\Lap} )$ increases (i.e., the Gaussian mechanism becomes more efficient than the Laplace mechanism) as  $\Delta_1/ \Delta_2$ increases or $\mu$ decreases.  (These claims are based on Lemma \ref{lem:re} in Appendix~\ref{sec:detailsforComparionsSection}, which collects some facts about ${\rm RE}_r(\Mv_G , \Mv_{\Lap} )$.)

Figure \ref{fig:Lap_vs_Gauss} plots the region (of $\mu, \Delta_1/\Delta_2, r$) under which the Gaussian mechanism dominates the Laplace mechanism. We point out that for high privacy situations with smaller privacy parameter $\mu$, the Gaussian mechanism dominates the Laplace mechanism. In particular, if $\mu < 2.55$ (roughly), ${\rm RE}_r(\Mv_G , \Mv_{\Lap} ) > 1$ for any $r \ge 1$ and $\Delta_1/\Delta_2 \ge 1$. Note that the ratio $\Delta_1/\Delta_2$ is typically greater than $1$; if $\thetav$ is a frequency table,  $\Delta_1/\Delta_2 = \sqrt{2}$, and if $\thetav$ is a multivariate mean, $\Delta_1/\Delta_2 = \sqrt{p}$ (for $\Xc = [0,1]^p$). For such typical cases, the Gaussian mechanism dominates the Laplace mechanism for a wide range of $\mu$, and the Laplace mechanism is useful only under very low privacy situations, as shown in Figure \ref{fig:Lap_vs_Gauss}(b).

\begin{figure}[t]
    \centering
    \begin{subfigure}{0.45\textwidth}
        \centering
        \includegraphics[width=.95\textwidth]{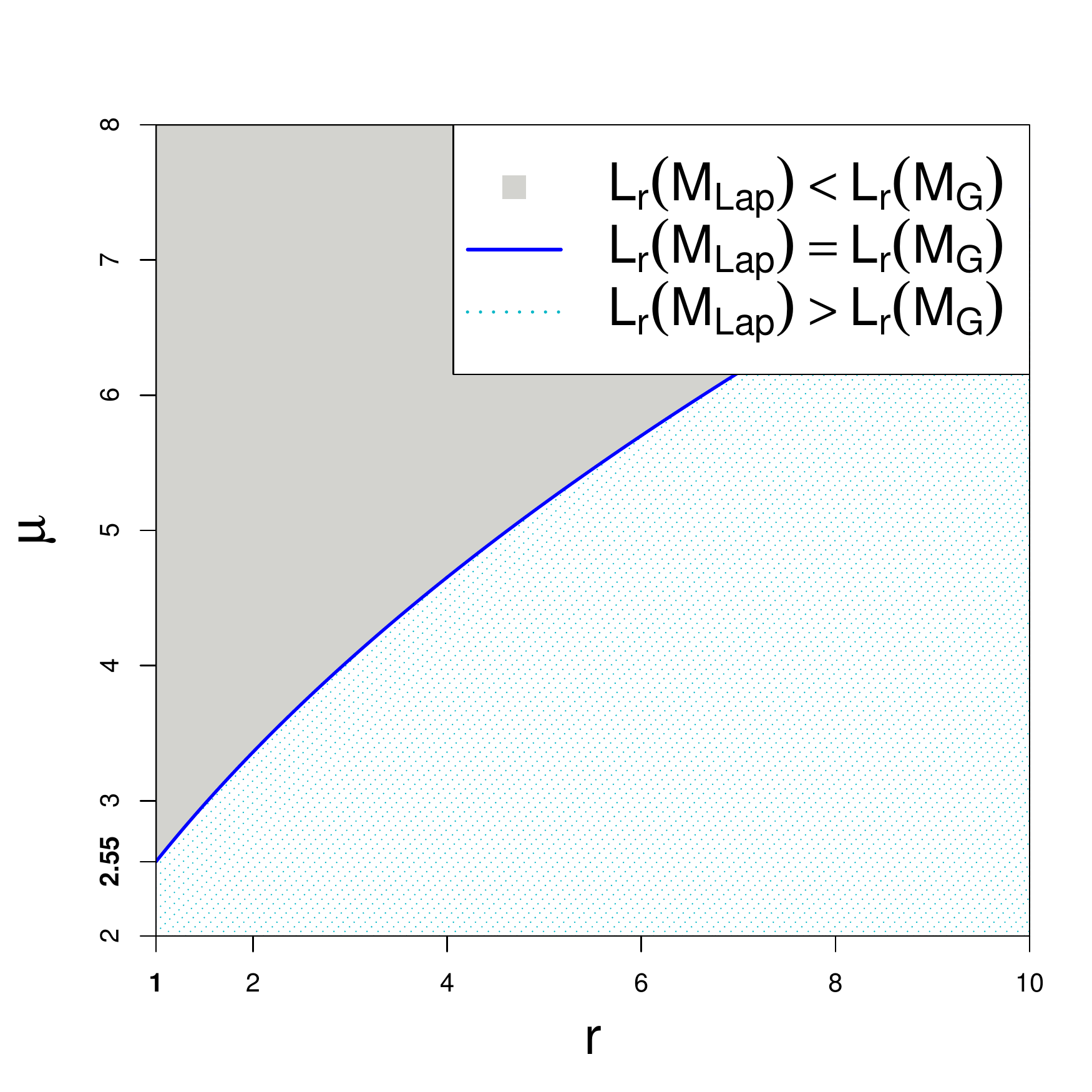}
        \caption{$\Delta_1/\Delta_2 = 1$.}
        \label{fig:mu_r_Lap_vs_Gauss}
    \end{subfigure}
    \begin{subfigure}{0.45\textwidth}
        \centering
        \includegraphics[width=.95\textwidth]{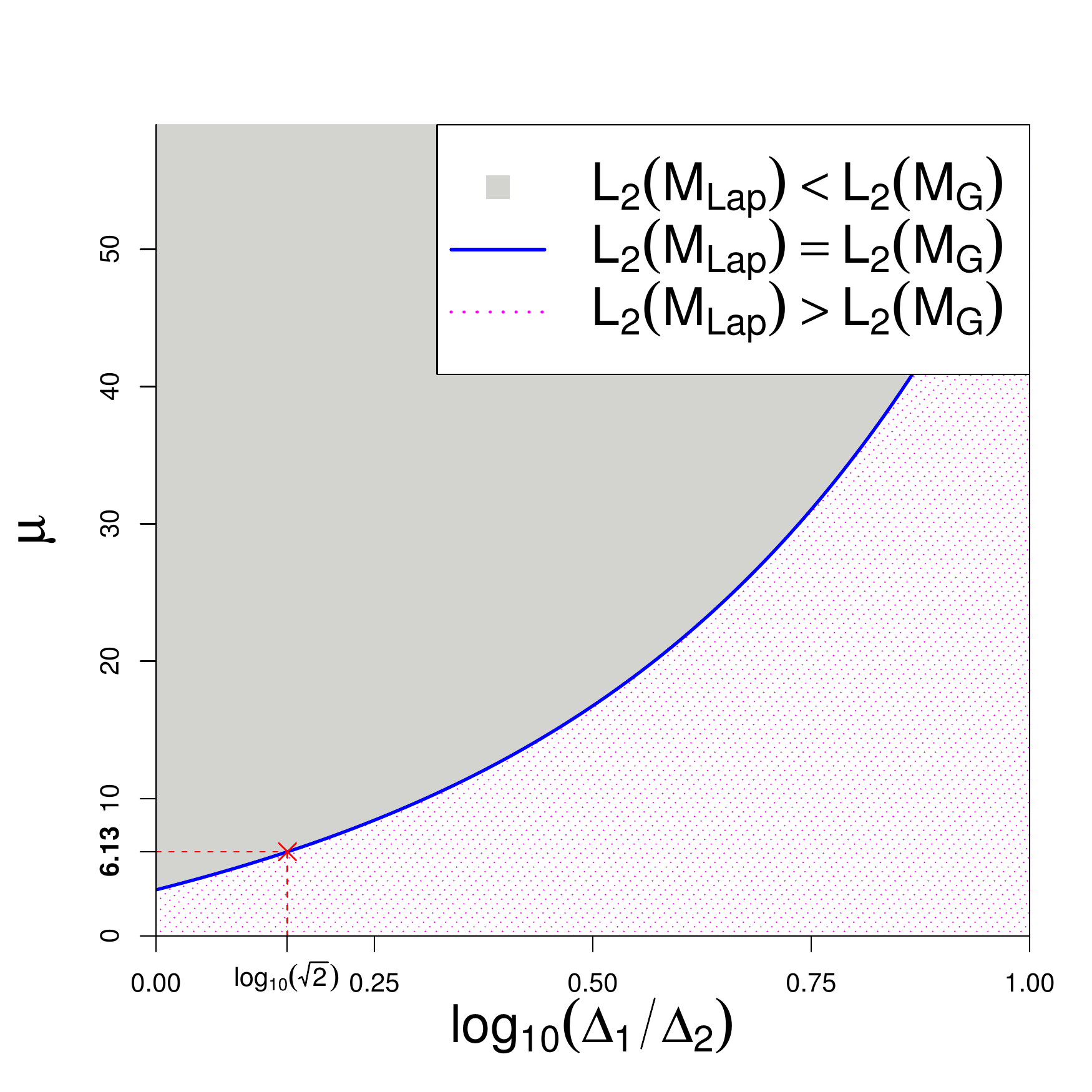}
        \caption{$r = 2$.}
        \label{fig:mu_sens_Lap_vs_Gauss}
    \end{subfigure}
    \caption{The partition of the set of parameters ($\mu, \Delta_1/\Delta_2, r$), depending on whether the Gaussian mechanism dominates the Laplace mechanism or not. %
    }
    \label{fig:Lap_vs_Gauss}
\end{figure}

Next, we focus on the case where $\thetav$ is a frequency table, so that one can use a much smaller scale parameter $\tilde{b}_\mu$ than $b_{\mu}^{\Delta_1}$ used in (\ref{eq:L_r-cost-ratio}), for the $\mu$-GDP Laplace mechanism. Figure \ref{fig:L2_g_lap} plots the relative efficiency of the Gaussian mechanism over the Laplace mechanisms when the scale parameter $b$ is set to be $\tilde{b}_\mu$, $b_{\mu}^{\Delta_1}$ or $b_\mu^\epsilon$, over a range of privacy parameter $\mu>0$, for $r=2$. The ordinary Gaussian mechanism $\Mv_G$ dominates the optimally calibrated $\Mv_{\Lap}(\cdot; \tilde{b}_\mu)$, for high privacy situations (corresponding to $\mu < 3.9$, approximately). We note that rank-deficient James--Stein mechanism $\Mv_{rJS}$ typically has a much smaller $L_2$-cost than $L_2(\Mv_G)$. Thus, unless the privacy protection level is very low (i.e., $\mu$ is large), Gaussian mechanisms have advantages over Laplace mechanisms, when all mechanisms satisfy the $\mu$-GDP criterion.

\begin{figure}%
	\centering
	\includegraphics[width=0.5\linewidth]{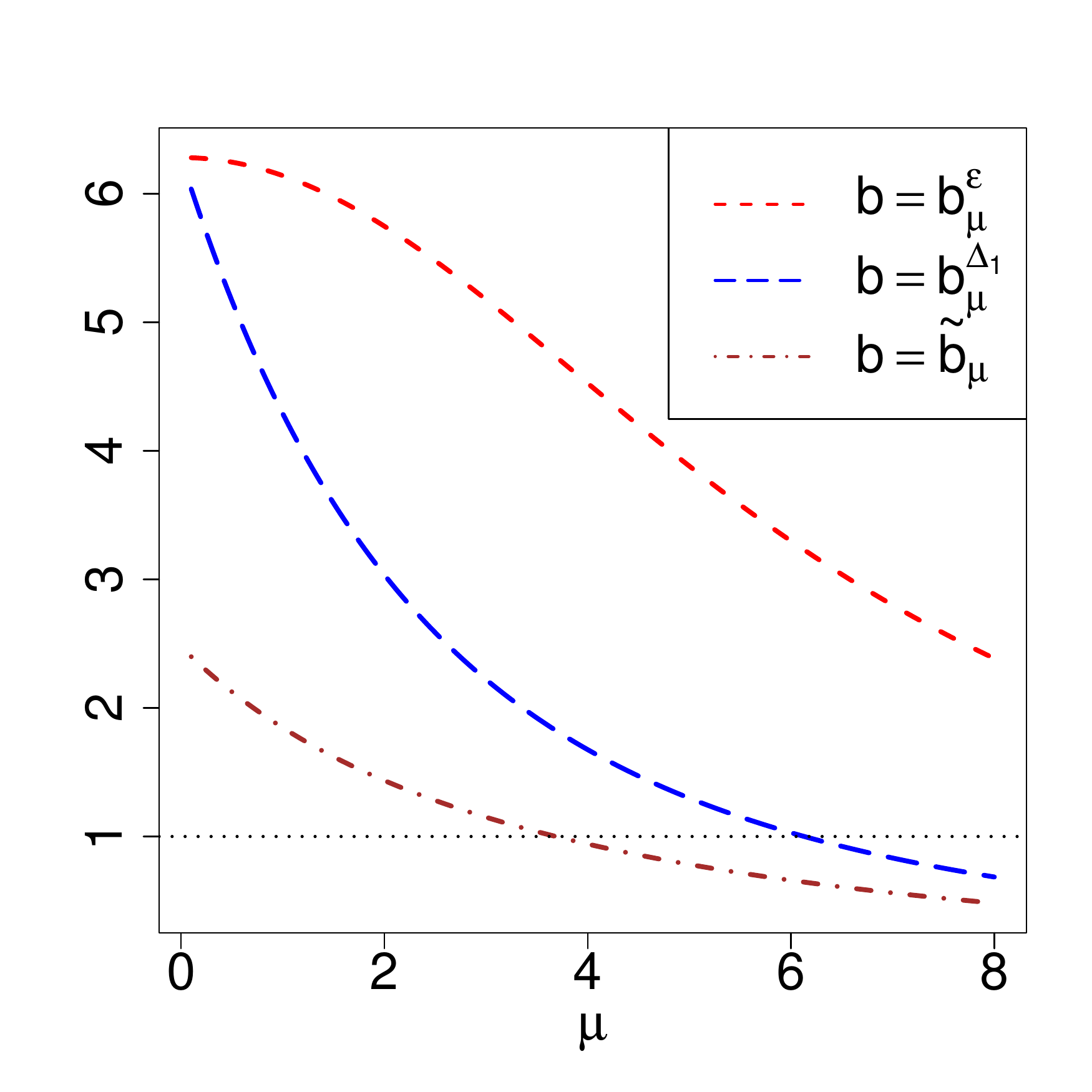}
	\caption{%
 The relative efficiency ${\rm RE}_2(\Mv_G, \Mv_{\Lap}(\cdot; b))$ (for three choices of $b$) as a function of $\mu$. The Gaussian mechanism $\Mv_G$ dominates $\Mv_{\Lap}(\cdot; b)$ if the relative efficiency is greater than 1. }
	\label{fig:L2_g_lap}
\end{figure}

\section{Applications to contingency table analysis} \label{sec:app_conting_tab}

A contingency table, or a multi-way frequency table, is a fundamental data summary widely used in various fields, such as official statistics, social sciences, and genome-wide association studies. Analysis of contingency tables routinely involves hypothesis tests based on the $ \chi^2 $ statistics such as the goodness-of-fit (GOF) and homogeneity tests. 

In this section, we demonstrate the applications of randomized mechanisms discussed in Sections \ref{section: mv-mech} and \ref{sec:Lap_mech} for private release of contingency tables, and numerically compare the statistical utilities. We also propose procedures for \emph{private GOF and homogeneity tests}, that can be conducted when only the differentially private (i.e., perturbed) contingency tables are available. For this application of hypothesis tests, the power of the test serves as another measure of statistical utility. %

Since a contingency table is simply the frequency table of Example~\ref{ex:contingency-table}, arranged in the form of cross-tabulation, we let $\thetav: \Xc^n \to \Real^p$ be the $r \times c$ contingency table, where  $p = r \times c$ is the number of cells.

\subsection{Private contingency table release}\label{sec:contingency_table_release}
The randomized multivariate mechanisms developed in Sections \ref{section: mv-mech} and \ref{sec:Lap_mech} are used to release randomly perturbed contingency tables, satisfying the $\mu$-GDP criterion. 

The ordinary Gaussian mechanism $ \Mv_G $ and the two James--Stein
mechanisms $ \Mv_{JS0} $ and $ \Mv_{JS} $ can be calibrated to satisfy $\mu$-GDP, by choosing $\sigma = \Delta_2(\thetav; \Xc^n)/\mu = \sqrt{2} / \mu$. The rank-deficient mechanisms $\Mv_r$ and $\Mv_{rJS}$ are beneficial since the sensitivity space of contingency tables is rank-deficient, with $d_{\thetav} = p-1$, and they are $\mu$-GDP if $\sigma = \sqrt{2} / \mu$. Since the output of $\Mv_{rJS}$ depends on the specific choice of the basis matrix $\Uv_{\thetav}$ (orthogonal to $\1v_p$), we choose to set $\Uv_{\thetav}$ to be the $p \times (p-1)$ Helmert sub-matrix \citep{harville1998matrix}. The Helmert matrix of order $p$ is $\Uv = [\uv_1, \Uv_{\thetav}]$, whose first column is $\mathbf{\uv}_1^\top=\1v_p / \sqrt{p}$, and the $k$th column (for $k=2,\ldots,p$) is $\mathbf{\uv}_k^\top=[k(k-1)]^{-1 / 2}(1,1, \ldots, 1,1-k, 0,0, \ldots, 0)$. 
For example, for $p = 4$, the basis matrix is 
$$
    \Uv_{\thetav} = \begin{pmatrix}
        1/\sqrt{2} & 1/\sqrt{6} & 1/\sqrt{12} \\
        -1/\sqrt{2} & 1/\sqrt{6} & 1/\sqrt{12} \\
        0 & -2/\sqrt{6} & 1/\sqrt{12} \\
        0 & 0 & -3/\sqrt{12}
    \end{pmatrix}.
$$

The multivariate Laplace mechanism with the optimal calibration $\Mv_{\Lap}(\cdot; \tilde b_{\mu})$ can also be used to release $\mu$-GDP contingency tables (see Section \ref{sec:lap_mech_table}).  

When the amount of perturbation in each of these mechanisms is large (or equivalently $\mu$ is small), it is possible that the output of the mechanism $\Mv(S) = \thetav(S) + \xiv$ contains negative values. Since it is desirable to have non-negative frequency counts, one can post-process each of mechanism outputs to have non-negative values. For any mechanism $\Mv$, we define the truncated mechanism, obtained by truncating the output of $\Mv(S)$ at zero, to be 
\begin{equation}\label{eq:contingency_truncated}
    (\Mv^c(S))_j = \max \{0, (\Mv(S))_j\}, \; j = 1, \cdots, p.
\end{equation}
By the post-processing property (Lemma~\ref{lem:post-processing}), if $\Mv$ is $\mu$-GDP, so is $\Mv^c$.
We write $\Mv_{rJS}^c$ for the truncated mechanism stemming from $\Mv_{rJS}$, and use a similar notation for other mechanisms. If $\thetav(S)_j \ge 0$ for any $j$ (which is the case for $\thetav$ being contingency tables), $\Mv^c$ dominates $\Mv$.

\subsection{Numerical comparison of \texorpdfstring{$L_2$}{L2} costs}\label{subsec:numerical_comparison}

In this section, we numerically compare the $L_2$-costs of the mechanisms discussed in Section~\ref{sec:contingency_table_release}. 
For this, we assume that the unperturbed contingency table $\thetav(S) \sim {\rm Multinomial}_p(n, \piv)$ is sampled from a multinomial distribution. Each $\thetav(S)$ is then ``sanitized'' by the application of the randomized mechanisms discussed in the previous section. 

We consider two models for the true probabilities $\piv = (\pi_1, \dots, \pi_p)$. In Model I, we assume a uniform probability for each cell, i.e., we set $\pi = 1/p$ for all $i$. In Model II, we set cell probabilities to increase linearly, i.e., $\pi_i = i / (\sum_{k=1}^p k)$. 
For each $p = 5, 10$, and $n = 100,200,\ldots, 3000$, a contingency table $\thetav(S)$ is sampled, and perturbed to obtain an output each from the mechanisms---$\Mv_G$, $\Mv_r$, $\Mv_{JS0}$, $\Mv_{JS}$, $\Mv_{rJS}$ and $\Mv_{\Lap}(\cdot; \tilde{b}_\mu)$ and their truncated versions---satisfying $\mu$-GDP. We set $\mu = 0.1$.
The squared $L_2$-loss $\|\Mv(S) - \thetav(S)\|^2_2$ for each output $\Mv(S)$ is averaged over $K = 10,000$ repetitions, for each mechanism. This is an estimate of the unconditional $L_2$-cost $L_2(\Mv)$.

\begin{figure}
\centering
\begin{subfigure}{0.45\textwidth}
    \includegraphics[width=\textwidth]{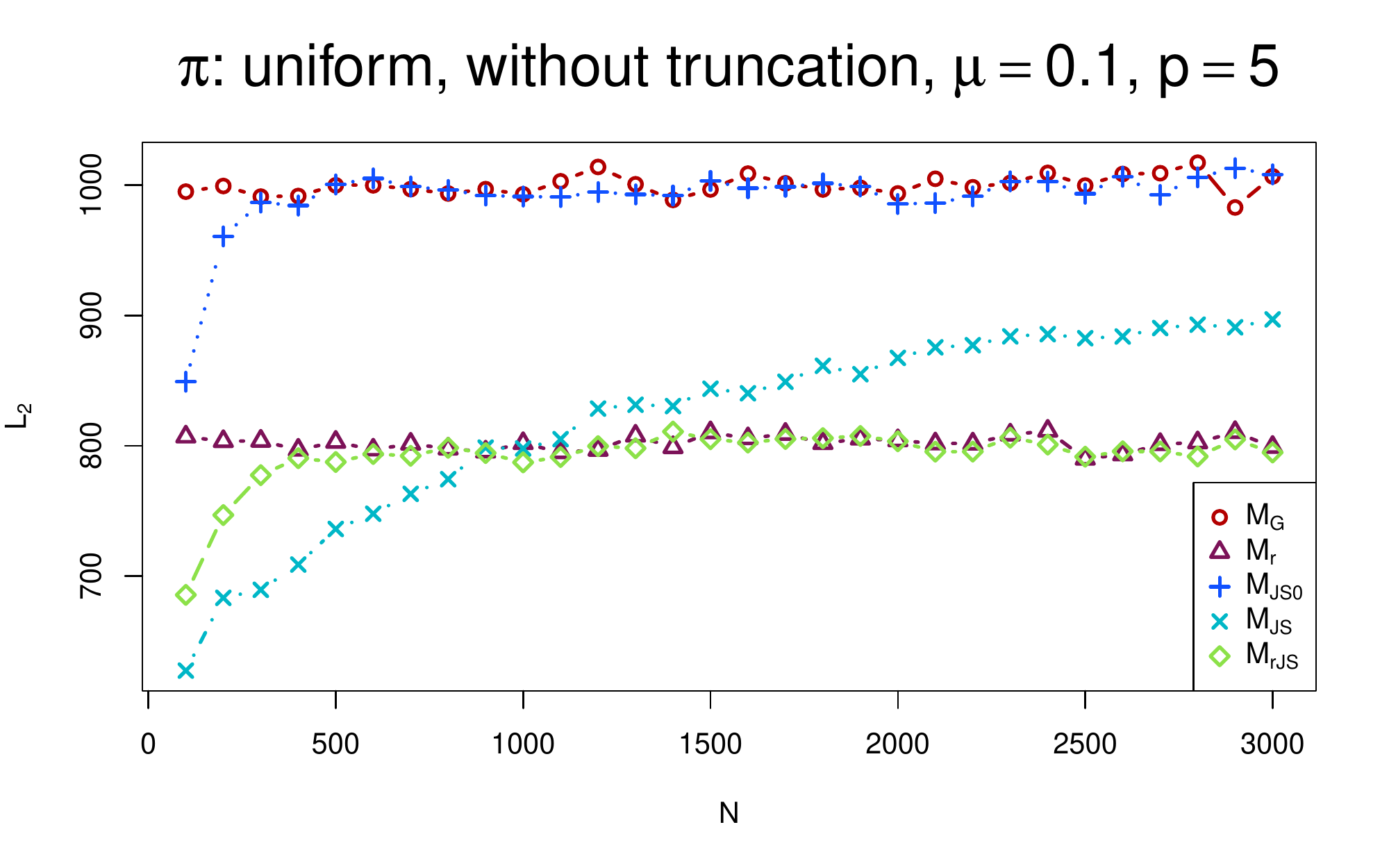}
\end{subfigure}
\begin{subfigure}{0.45\textwidth}
    \includegraphics[width=\textwidth]{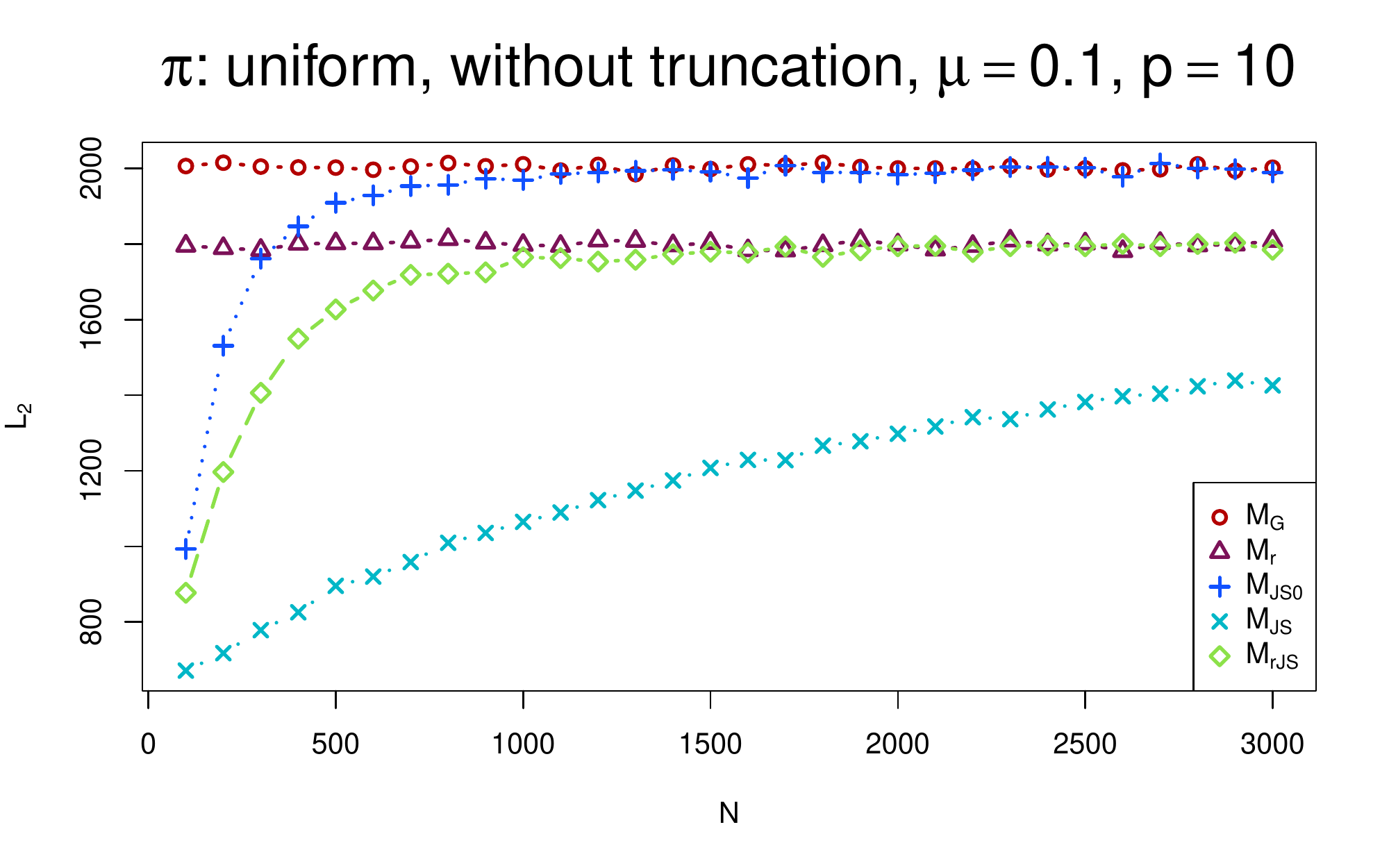}
\end{subfigure}

\begin{subfigure}{0.45\textwidth}
    \includegraphics[width=\textwidth]{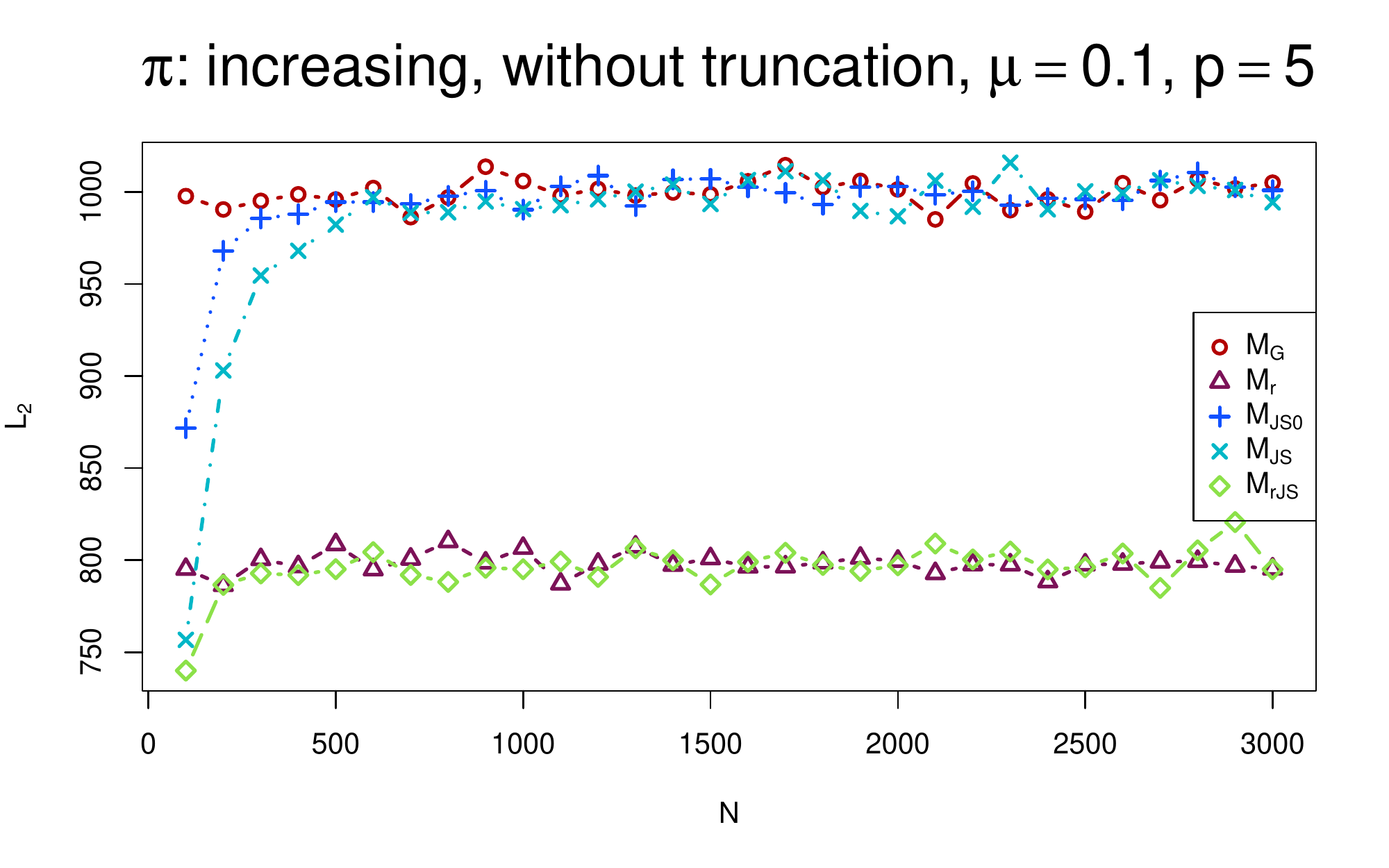}
\end{subfigure}
\begin{subfigure}{0.45\textwidth}
    \includegraphics[width=\textwidth]{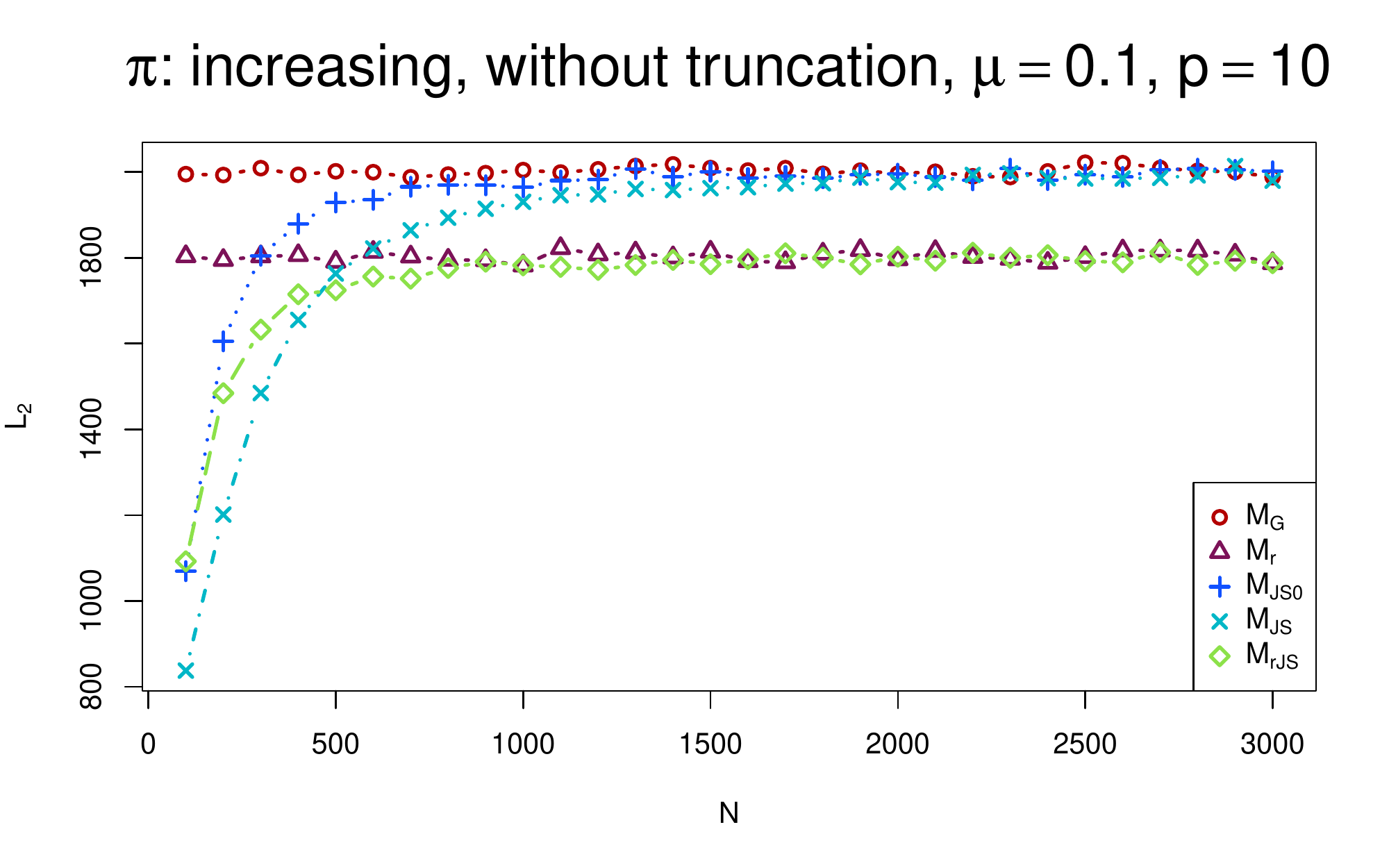}
\end{subfigure}
\caption{A comparison of $L_2$-costs of mechanisms without truncation, over increasing sample sizes. The first (or second) row shows results under Model I (or Model II, respectively).}
\label{fig:L2sim-nc}
\end{figure}

\begin{figure}
\centering
\begin{subfigure}{0.45\textwidth}
    \includegraphics[width=\textwidth]{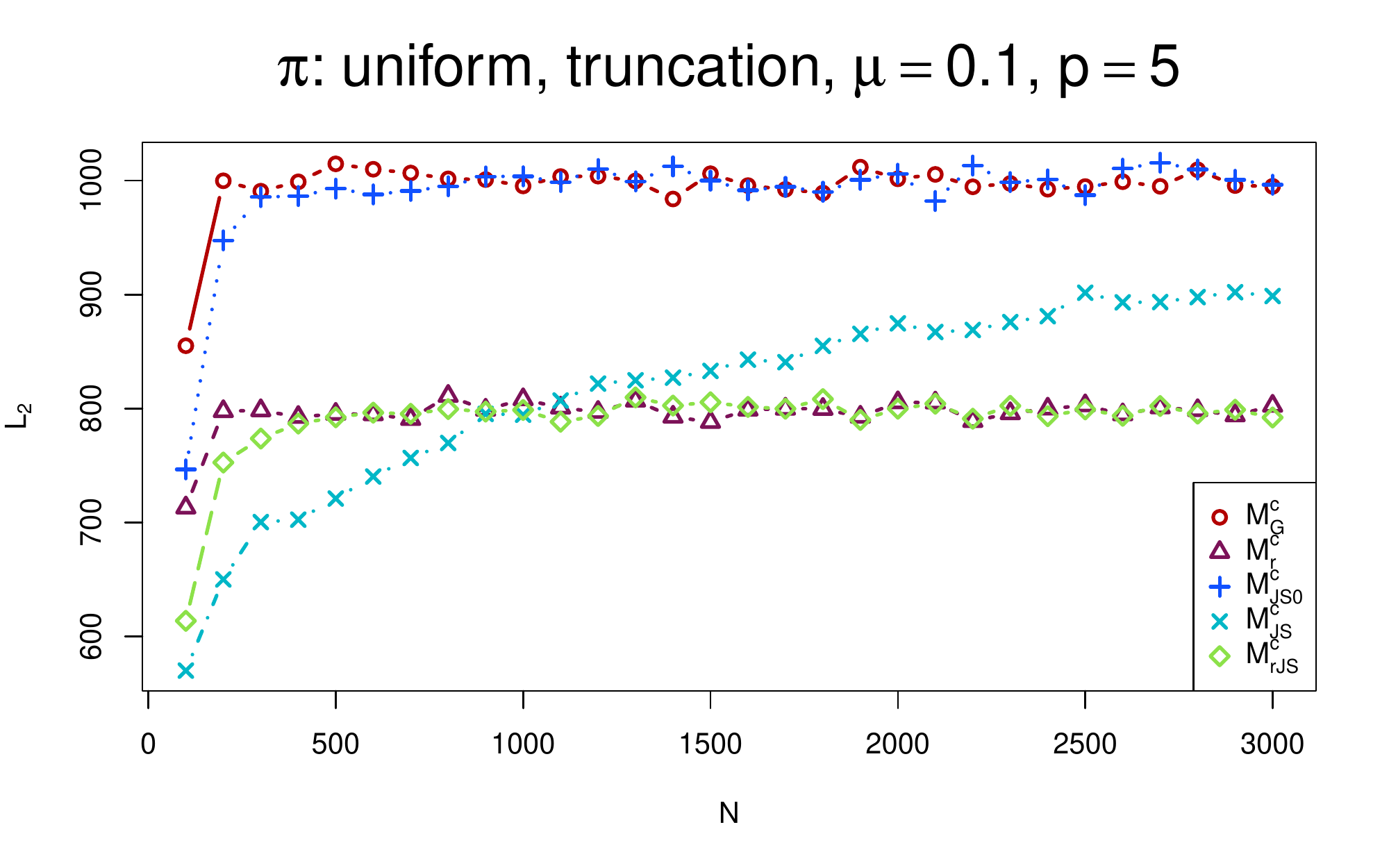}
\end{subfigure}
\begin{subfigure}{0.45\textwidth}
    \includegraphics[width=\textwidth]{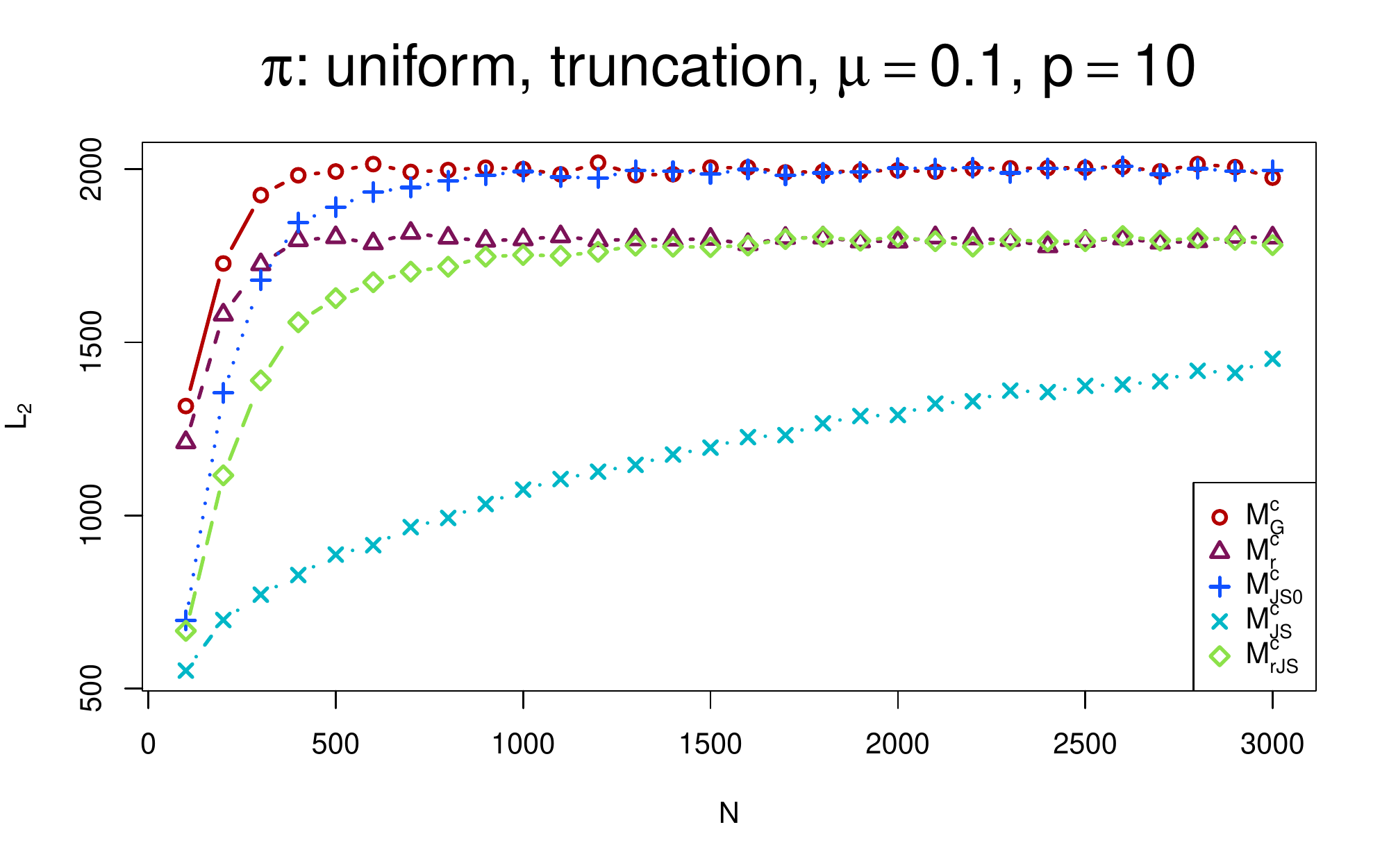}
\end{subfigure}
\begin{subfigure}{0.45\textwidth}
    \includegraphics[width=\textwidth]{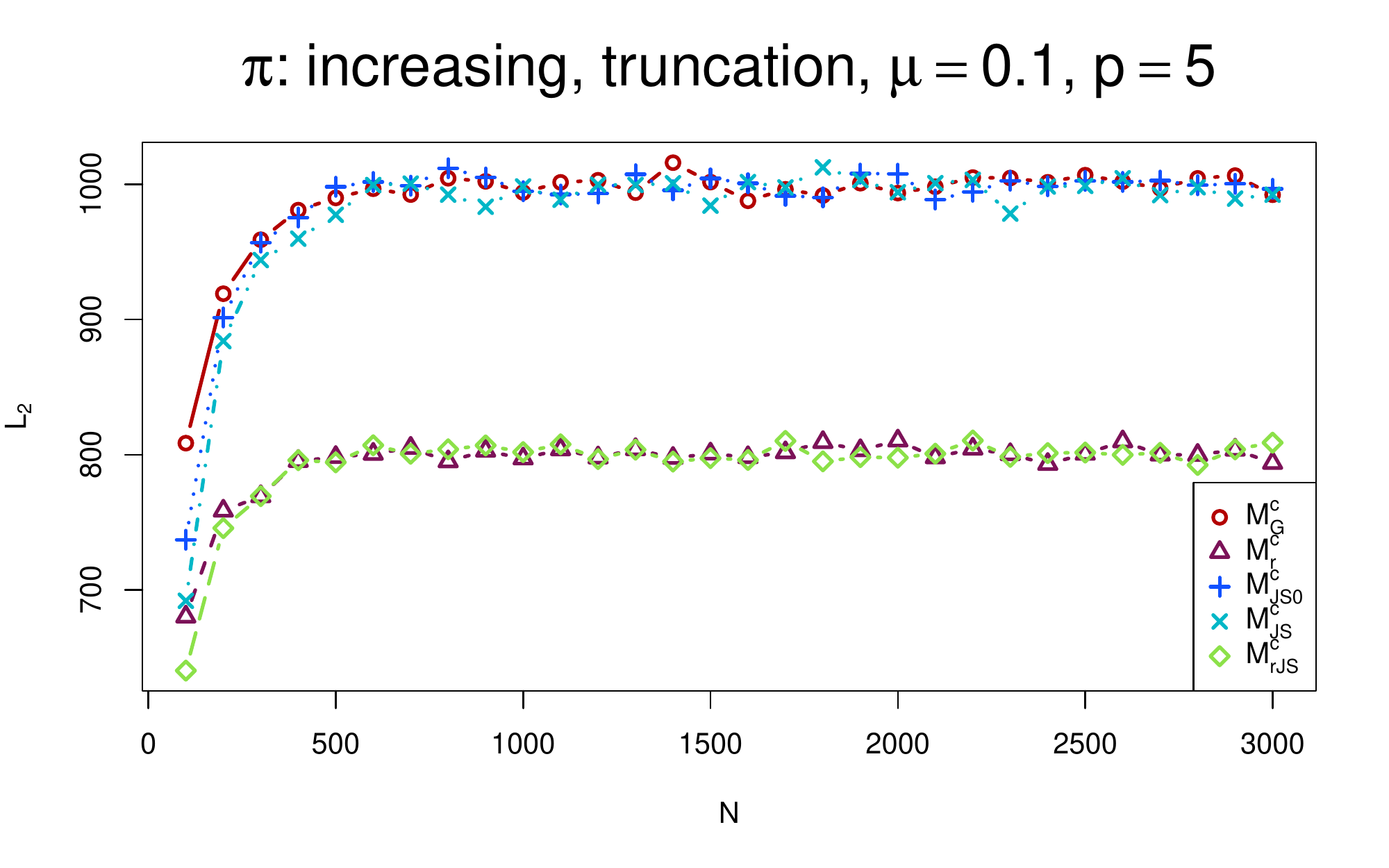}
\end{subfigure}
\begin{subfigure}{0.45\textwidth}
    \includegraphics[width=\textwidth]{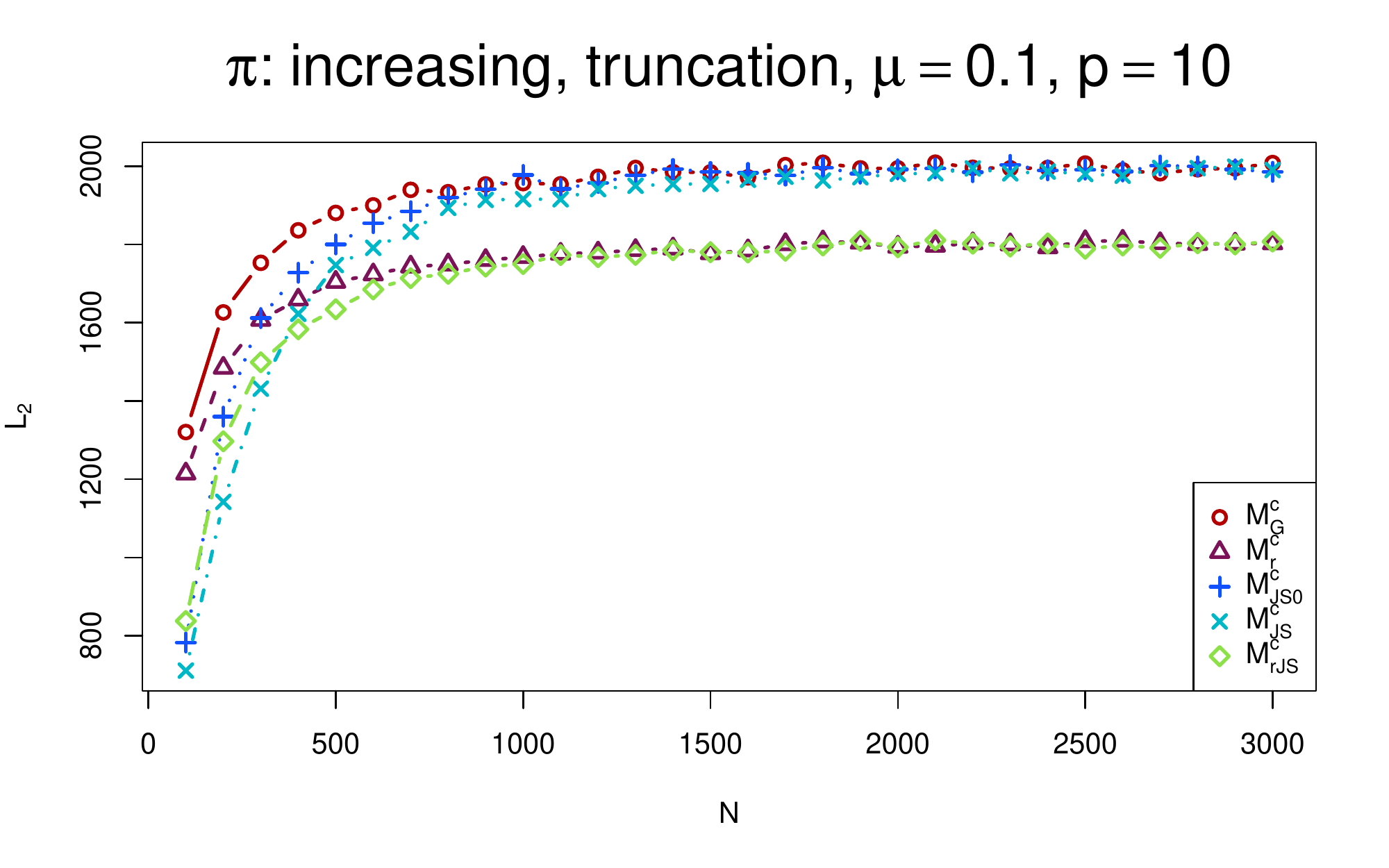}
\end{subfigure}
\caption{A comparison of $L_2$-costs of truncated mechanisms, over increasing sample sizes. The first (or second) row shows results under Model I (or Model II, respectively).}
\label{fig:L2sim-cut}
\end{figure}

The simulated results are plotted in Figures \ref{fig:L2sim-nc} and \ref{fig:L2sim-cut}. Results from $\Mv_{\Lap}(\cdot ; \tilde{b}_\mu)$ are omitted because their $L_2$-costs are much larger than those of Gaussian mechanisms, as seen in Section \ref{sec:comparison}. 
See Figure~\ref{fig:relationship} for a summary of the comparison between Gaussian mechanisms in this simulation study. In addition, we observe the following. 

The benefit of James--Stein shrinkage eventually wears off as the sample size increases, as seen in Theorem~\ref{thm:L2cost-conv}. However, depending on the model, the benefit of shrinkage can be persistent.  
For Model I, $\E(\thetav(S)) = n\piv = (n/p) \1v_p$ and 
the James--Stein mechanism $\Mv_{JS}$ performs best for smaller sample sizes as the model agrees with the target of the shrinkage. Nevertheless, the rank-deficient mechanisms $\Mv_{r}$ and $\Mv_{rJS}$ exhibit smallest $L_2$-costs for larger sample sizes. On the other hand, for Model II, while the rank-deficient James--Stein mechanism $\Mv_{rJS}$ performs the best across different sample sizes, $\Mv_r$'s $L_2$-cost is also the smallest for moderately large sample sizes. 

We also observe that the advantage of truncation is more pronounced for small sample sizes. The benefit of truncation is large only if there are many negative valued outputs in the untruncated mechanisms. As the sample size increases, $\E(\thetav(S)) = n \piv$  increases, thus the outputs of $\Mv(S)$ tend to be strictly positive.

\subsection{Differentially private hypothesis tests for contingency table}\label{sec:DP-tests}

In this section, we discuss private $\chi^2$ tests based on $\mu$-GDP contingency tables $\Mv(S)$, perturbing $\thetav(S)$, for testing  goodness-of-fit (GOF) and homogeneity. We consider two-way contingency tables  $\thetav(S)$  of size $r\times c$, in which  $\thetav(S)_{ij}$ collects the count of observations belonging to the $(i,j)$th category. 
  
Classical $ \chi^2 $ tests are based on the test statistic $ T $ given by 
\begin{equation}\label{eq:chi_sq}
    T = \sum_{i=1}^c \sum_{j=1}^r \frac{(\thetav(S)_{ij} - E_{ij})^2}{E_{ij}},    
\end{equation} 
where $ E_{ij} $ is the expected cell count under the null hypothesis. The large-sample null distribution of $T$ is $\chi^2_k$ for some degrees of freedom $k$, and this fact is used to construct a $\chi^2$ test. Suppose now that a \emph{sanitized} contingency table $\Mv(S)$ is only available, in place of the original data summary $\thetav(S)$. (The randomized mechanism $\Mv$ can be any mechanism we discussed in Sections \ref{section: mv-mech} and \ref{sec:Lap_mech}.)  Some researchers including \cite{kifer2016new, gaboardi2016DPchisq, wang2017revisiting,son2022parametric} proposed to simply replace $\thetav(S)$  by $\Mv(S)$ in (\ref{eq:chi_sq}), resulting in a private statistic 
$$
    \tilde T = \sum_{i=1}^c \sum_{j=1}^r \frac{(\Mv(S)_{ij} - \tilde E_{ij})^2}{\tilde E_{ij}},
$$
where $\tilde E_{ij}$ is obtained from $\Mv(S)$ under the null hypothesis. Critical values for the private test statistic $\tilde T$ are obtained either by large-sample asymptotic approximations or by a bootstrap procedure. In the following, we use a parametric bootstrap procedure to approximate the null distribution of $\tilde T$ for the GOF test and the homogeneity test. 

In the construction of private tests, we assume the following.
\begin{enumerate}

    \item A sanitized table $\Mv(S)$, satisfying $\mu$-GDP, is only available to us. The original table $\thetav(S)$ as well as the data set $S$ are not accessible. 

    \item The exact mechanism used to produce $\Mv(S)$ is known. 

    \item The privacy level $\mu$ (or, equivalently the scale parameters $\sigma$ and $b$, for Gaussian and Laplace mechanisms, respectively) is known.
\end{enumerate}

Since our test procedure depends on the data set $S$ only via  $\Mv(S)$, if $\Mv(S)$ is $\mu$-GDP, then any conclusion made from the test procedure is also $\mu$-GDP. 

We remark that most of private $\chi^2$ tests in the literature are based on Laplace mechanisms and are calibrated to satisfy $\epsilon$-DP. On the other hand, we calibrate our randomized mechanisms to satisfy $\mu$-GDP. Moreover, those previous works assume that the total sample size $n$ is known  \citep{gaboardi2016DPchisq, kifer2016new}. However, when $\Mv(S)$ is only available, the sample size $n$ is also hidden. Our test procedures do not assume that $n$ is known.

\subsubsection{Private goodness of fit test}\label{sec:GOF}
Consider a contingency table $ \thetav(S) \sim \mbox{Multinomial}(n, \piv)$, where $\piv \in \Real^{r}_+ \times \Real^{c}_+$ such that $\sum_{i,j} \pi_{ij} = 1$.
For a given probability parameter $ \piv_0 $,
the GOF test aims to test
\begin{equation} \label{eq: hypothesis-GOF}
    H_0: \piv = \piv_0 \quad \mbox{vs.} \quad H_1: \piv \neq \piv_0.
\end{equation}
In the non-private GOF test, the expected cell count under $H_0$ is simply $ E_{ij} = n (\piv_0)_{ij} $ and the $ \chi^2 $ statistic becomes
$
    T_{GOF} =\sum_{i=1}^r\sum_{j=1}^c {(\thetav(S)_{ij} - n (\piv_0)_{ij})^2}/ \{{n(\piv_0)_{ij}}\}.
$
Suppose now that an output $\tilde\thetav$ of $\Mv(S)$ is only available to us.  The private $\chi^2$ test statistic is 
$$
    \tilde T_{GOF} = \sum_{i=1}^r\sum_{j=1}^c \frac{(\tilde\thetav_{ij} - \tilde n (\piv_0)_{ij})^2}{\tilde n (\piv_0)_{ij}},
$$
where $\tilde n = \sum_{i,j} \tilde \thetav_{ij}$.
 
To approximate the null distribution of $\tilde T_{GOF}$, we use a parametric bootstrap resampling under the null hypothesis, which amounts to sampling $ \tilde \thetav^{(b)}  \sim {\rm Multinomial}(\tilde{n}, \piv_0)$ and calculate bootstrap-replicates $\tilde T_{GOF}^{(b)}$ $ b = 1, \dots, B $ for some large $B$. (We used $B = 5000$ in our numerical studies.)
The bootstrapped p-value of the observed statistic $\tilde T_{GOF}$ is then the proportion of bootstrap replicates greater than $\tilde T_{GOF}$. See Algorithm \ref{alg:BootGOF}. In the implementation of this private GOF test, we set $ 0/0 = 0 $, and if $\tilde n \le 0$, which can happen if the amount of perturbation is large, then we do not reject the null hypothesis.

\begin{algorithm}
	\SetAlgoLined\DontPrintSemicolon
	\SetKwFunction{proc}{BootGOF}
	\SetKwProg{myproc}{Procedure}{}{}
  	\myproc{\proc{$\tilde \thetav$, $ \mu $, $ \mathbf \Mv(\cdot) $, $\piv_0 $, $ \alpha $}}{ 
		$ \tilde n \gets \sum_{i=1}^r \sum_{j=1}^{c} \tilde \thetav_{ij}  $ \; 
		$ \tilde T_{GOF} =\sum_{i=1}^r \sum_{j=1}^{c} \frac{(\tilde \thetav_{i,j} - \tilde n (\piv_0)_{ij})^2}{\tilde n (\piv_0)_{ij}} $\;
	    \For(){$b = 1, \dots, B$}{
			$ \tilde \thetav^{(b)} \gets \mathbf{M}(\thetav^{(b)}) $ where
			$\thetav^{(b)} \sim \mbox{Multinomial}(\tilde n, \piv_0) $\;
			$ \tilde n^{(b)} \gets \sum_{i=1}^r \sum_{j=1}^{c} \tilde \thetav_{ij}^{(b)}  $\; 
			$\tilde T_{GOF}^{(b)} = \sum_{i=1}^r \sum_{j=1}^{c}
            \frac{(\tilde \thetav_{ij}^{(b)} - \tilde n^{(b)}  (\piv_{0})_{ij})^2}{\tilde n^{(b)}  (\piv_{0})_{ij}} $\;
		}
		$ \mbox{P-value} \gets \frac{1}{B}\sum_{b=1}^{B}I(\tilde T_{GOF}^{(b)} \ge \tilde T_{GOF}) $\;
		\uIf{$\mbox{P-value} < \alpha$}{
			\Return reject $ H_0 $\;
		}
		\Else(){
			\Return do not reject $H_0$\;
		}
 	}
	\caption{Private parametric bootstrap test of goodness of fit}\label{alg:BootGOF}
\end{algorithm}

\subsubsection{Private homogeneity test}
Consider the situation where there are $r$ populations, and the data set $S$ is the union of data sets $S_1, \dots, S_r $ from each population. For each $i = 1,\ldots, r$, the $i$th row of the $r\times c$ contingency table $\thetav(S)$ is $\theta(S_i)$. Assume that $\theta(S_i) \sim \mbox{Multinomial}(n_i, \piv_{i}) $ where $ \piv_i \in \mathbb{R}_+^c $ is a vector of probabilities for the $i$th population. Testing the homogeneity among these population amounts to testing the null hypothesis $H_0: \piv_1 = \dots = \piv_r$ against a general alternative.  

In the classical homogeneity test, the $\chi^2$ statistic is obtained from the data $\thetav(S)$ as follows: Since the null probabilities $\piv_0 := \piv_1 = \cdots = \piv_r$ are unknown, an estimator $\hat\piv_0 =  \sum_{i=1}^r \thetav(S_i) / ({\sum_{i=1}^r n_i}) =: (\hat\pi_{01}, \ldots, \hat\pi_{0c})^\top$ is used to set the expected counts $E_{ij} = n_i\hat\pi_{0j}$ under the null hypothesis. The non-private $\chi^2$ statistic is 
$
    T_{Hom} = \sum_{i=1}^r \sum_{j=1}^c {(\theta(S_i)_{j} - E_{ij} )^2}/\{E_{ij}\}.
$

If an output $\tilde{\thetav}$ of $\Mv(S)$ is only available, then the private $\chi^2$ statistic, say  $\tilde{T}_{Hom}$, is obtained by replacing $\thetav(S)$ with $\tilde{\thetav}$ in the formula for the non-private statistic $T_{Hom}$. A parametric bootstrap procedure for testing the homogeneity using the private statistic $\tilde{T}_{Hom}$ can be constructed by following the same arguments used for the private GOF test in Section \ref{sec:GOF}; see Algorithm \ref{alg:BootHomTest}.

\begin{algorithm}
	\SetAlgoLined\DontPrintSemicolon
	\SetKwFunction{proc}{BootHom}
	\SetKwProg{myproc}{Procedure}{}{}
  	\myproc{\proc{$\tilde \thetav$, $\mu $, $ \mathbf M(\cdot) $, $ \alpha $}}{
		$ \tilde n_i \gets \sum_{j=1}^{c}  (\tilde\thetav_{i})_j  $ for $ i = 1, \ldots, r $\;
		$ \tilde n \gets \sum_{i=1}^r \tilde n_i$	\;
		$ \tilde \piv_0 \gets (\tilde n)^{-1} \sum_{i=1}^r \tilde \thetav_{i}$	\;
		$ \tilde T_{Hom} \gets \sum_{i=1}^r \sum_{j=1}^c
		\frac{((\tilde \thetav_{i})_j - \tilde n_i \tilde \piv_{0j})^2}{\tilde n_i \tilde \piv_{0j}} $\;
	    \For(){$b = 1, \dots, B$}{
			$\tilde \thetav_{i}^{(b)} \gets
			\mathbf{M}(\thetav_i) $ where
			$\theta_i \sim \mbox{Multinomial}(\tilde n_i, \tilde \piv_0) $
			for $ i = 1, \ldots, r $\;
			$ \tilde n_i^{(b)} \gets  \sum_{j=1}^{c}  (\tilde\thetav_{i}^{(b)})_j  $
			for $ i = 1, \ldots, r $\;
		    $\tilde n^{(b)} \gets \sum_{i=1}^r \tilde n_i^{(b)}$	\;
		    $\tilde \piv_0^{(b)} \gets (\tilde n^{(b)})^{-1} \sum_{i=1}^r \tilde\thetav_{i}^{(b)} $ \;
			$\tilde T_{Hom}^{(b)} \gets \sum_{i=1}^r \sum_{j=1}^c \frac{((\tilde\thetav_{i}^{(b)})_j - \tilde n^{(b)} \tilde \piv^{(b)}_{0j})^2}{\tilde n^{(b)} \tilde \piv^{(b)}_{0j}} $\;
		}
		$ \mbox{P-value} \gets \frac{1}{B}\sum_{b=1}^{B}I(\tilde T_{Hom}^{(b)} \ge \tilde T_{Hom}) $\;
		\uIf{$\mbox{P-value} < \alpha$}{
			\Return reject $ H_0 $\;
		}
		\Else(){
			\Return do not reject $H_0$\;
		}
 	}
	\caption{Private parametric bootstrap test of homogeneity}\label{alg:BootHomTest}
\end{algorithm}

\subsection{Simulation studies for private test procedures}\label{sec:sim_gof}

In this section, we numerically investigate the private GOF tests in Section~\ref{sec:DP-tests}, based on the $\mu$-GDP private contingency tables discussed in Section~\ref{sec:contingency_table_release}. In particular, we examine whether the test procedures control the type I error rate at the given significance level $\alpha$, and compare the power of the tests as a measure of the statistical utility of the randomized mechanisms. We also have investigated the performance of private homogeneity tests, and the results are listed and discussed in Appendix~\ref{sec:app_sim}. The conclusions made for homogeneity tests are similar to the case of GOF tests. 

Probability models for the null and alternative hypotheses are generated using real data sets from 2020 resident population demographic data of  South Korea \citep{korAge2020} and the United States \citep{americaAge2020}. We set the null probabilities $\piv_0$ to be the $p \times 1$ vector $\piv_{\rm Kor}^{(p)}$ consisting of the age-wise population proportions, given by $p$ age groups, from the South Korean demographic data. We use $p = 9, 18$ corresponding to 5-year and 10-year age groups, respectively. The alternative probabilities $\piv_{\rm US}$ consist of the $p$ proportions from the U.S. demographic data. 

We simulate under two privacy levels given by $\mu=0.1$ (representing a high privacy regime) and $\mu = 0.3$ (moderate privacy). Type I error rates of the private GOF test (Algorithm \ref{alg:BootGOF}) are evaluated by randomly generating contingency tables $\thetav(S) \sim \mbox{Multinomial}(n, \piv_{Kor}^{(p)})$, applying a mechanism in Section~\ref{sec:contingency_table_release}, then conducting the private GOF test (Algorithm \ref{alg:BootGOF}) with the null hypothesis $\piv = \pi_{Kor}^{(p)} $. For the power analysis, we generate tables from
$\mbox{Multinomial}(n, \pi_{US}^{(p)})$. 
Over a range of sample sizes $n$, we record the proportions of rejection (with significance level $\alpha = 0.05$) over $K = 1000$ repetitions, as estimates of the type I error rate and the power of the test. 
 
The simulation results for the GOF test, based on the truncated Gaussian type mechanisms ($\Mv_{G}^c$, $\Mv_{r}^c$, $\Mv_{JS0}^c$, $\Mv_{JS}^c$, $\Mv_{rJS}^c$) and the truncated Laplace mechanism ($\Mv_{\Lap}(\cdot \: ; \tilde b_{\mu})$), are plotted in Figures \ref{fig:gof-type1} and \ref{fig:gof-power}

\begin{figure}
\centering
\begin{subfigure}{0.4\textwidth}
    \centering
    \includegraphics[width=.8\textwidth]{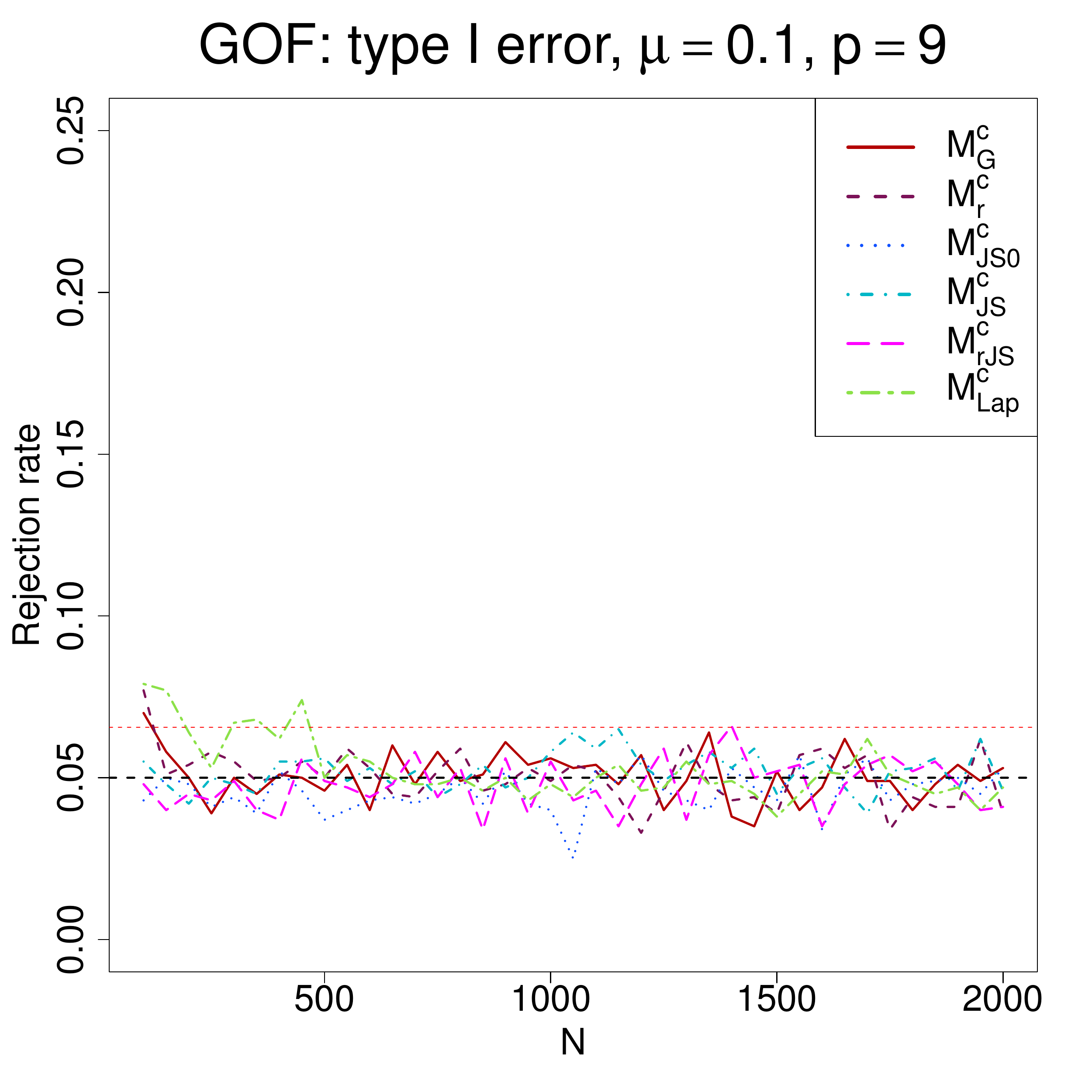}
\end{subfigure}
\begin{subfigure}{0.4\textwidth}
    \centering
    \includegraphics[width=.8\textwidth]{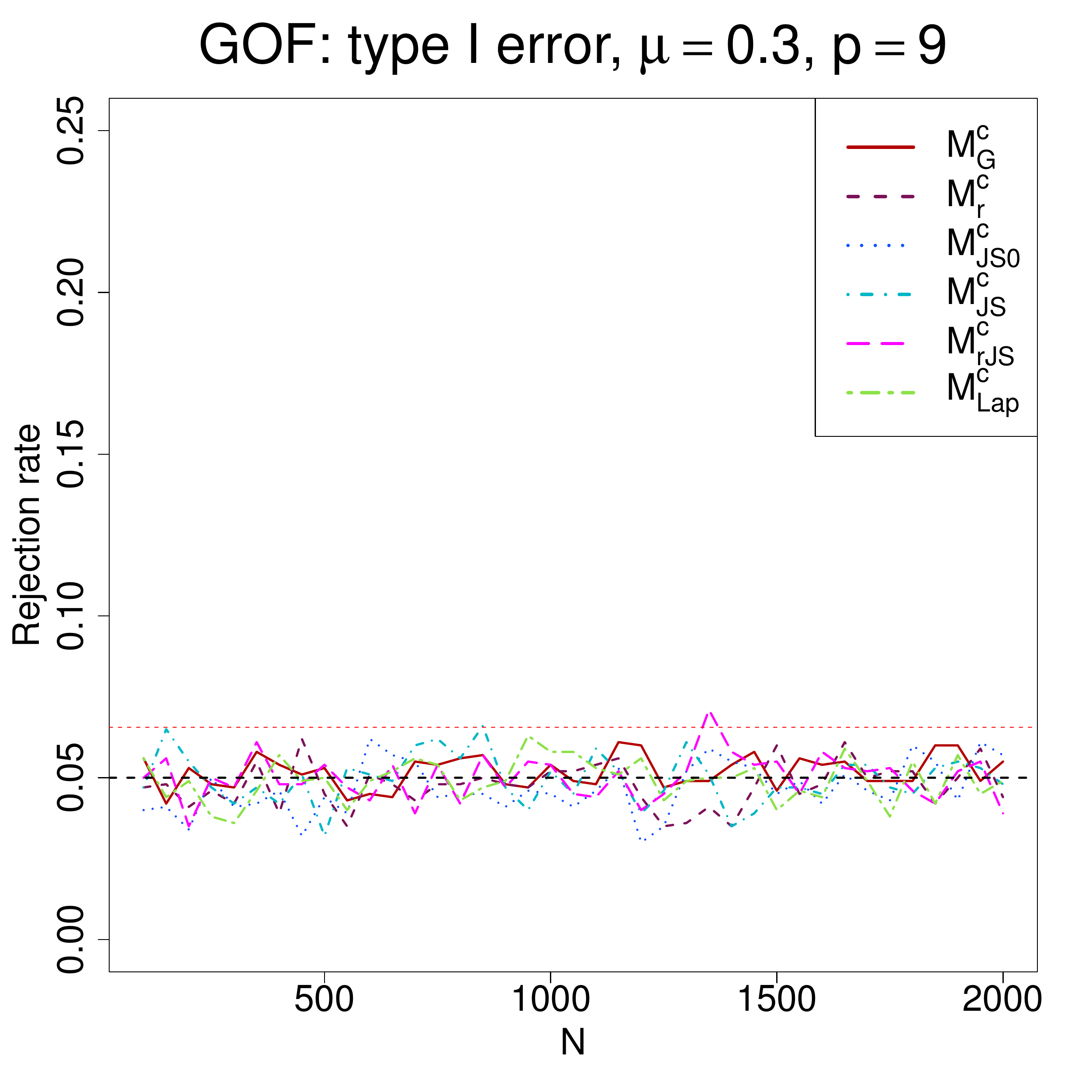}
\end{subfigure}
\begin{subfigure}{0.4\textwidth}
    \centering
    \includegraphics[width=.8\textwidth]{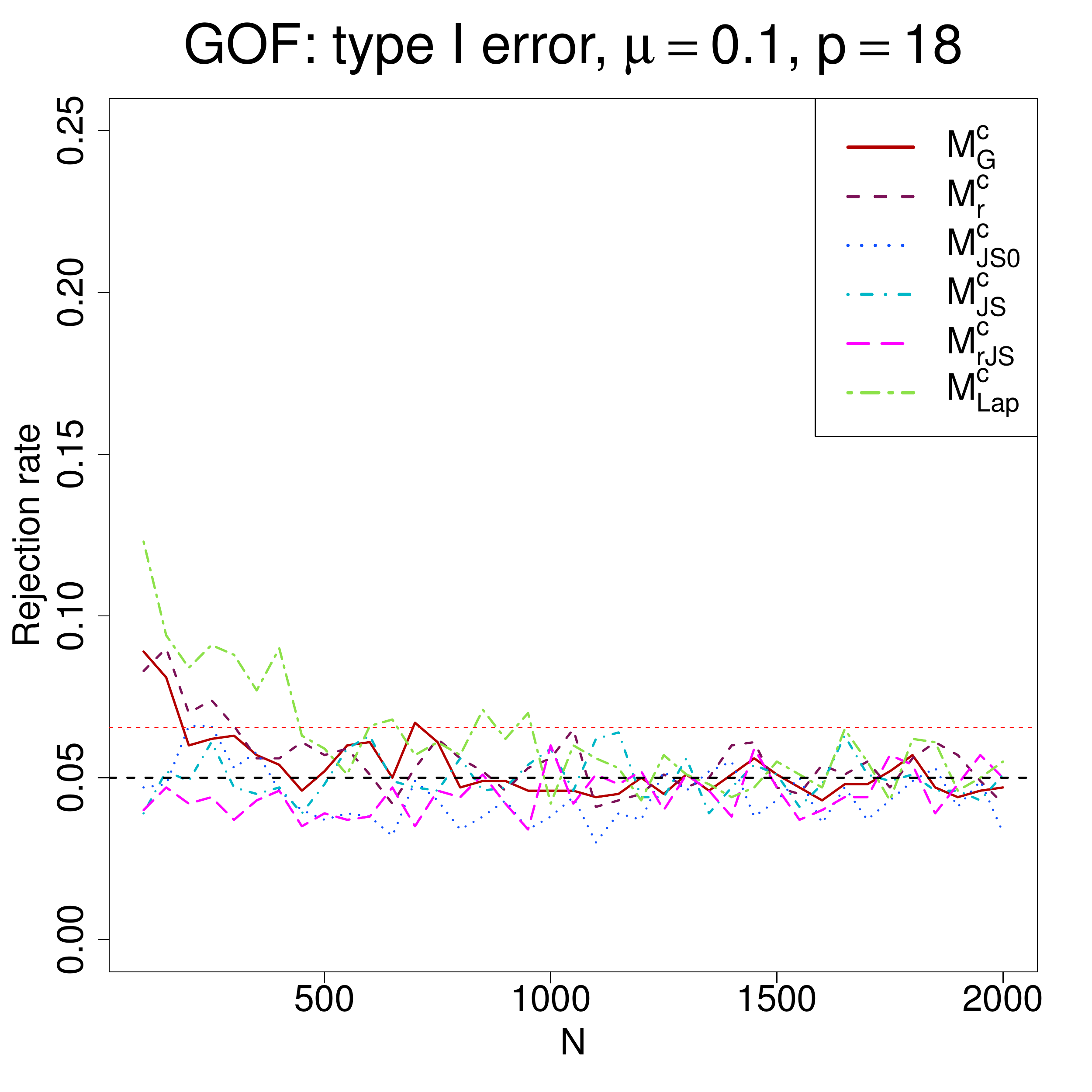}
\end{subfigure}
\begin{subfigure}{0.4\textwidth}
    \centering
    \includegraphics[width=.8\textwidth]{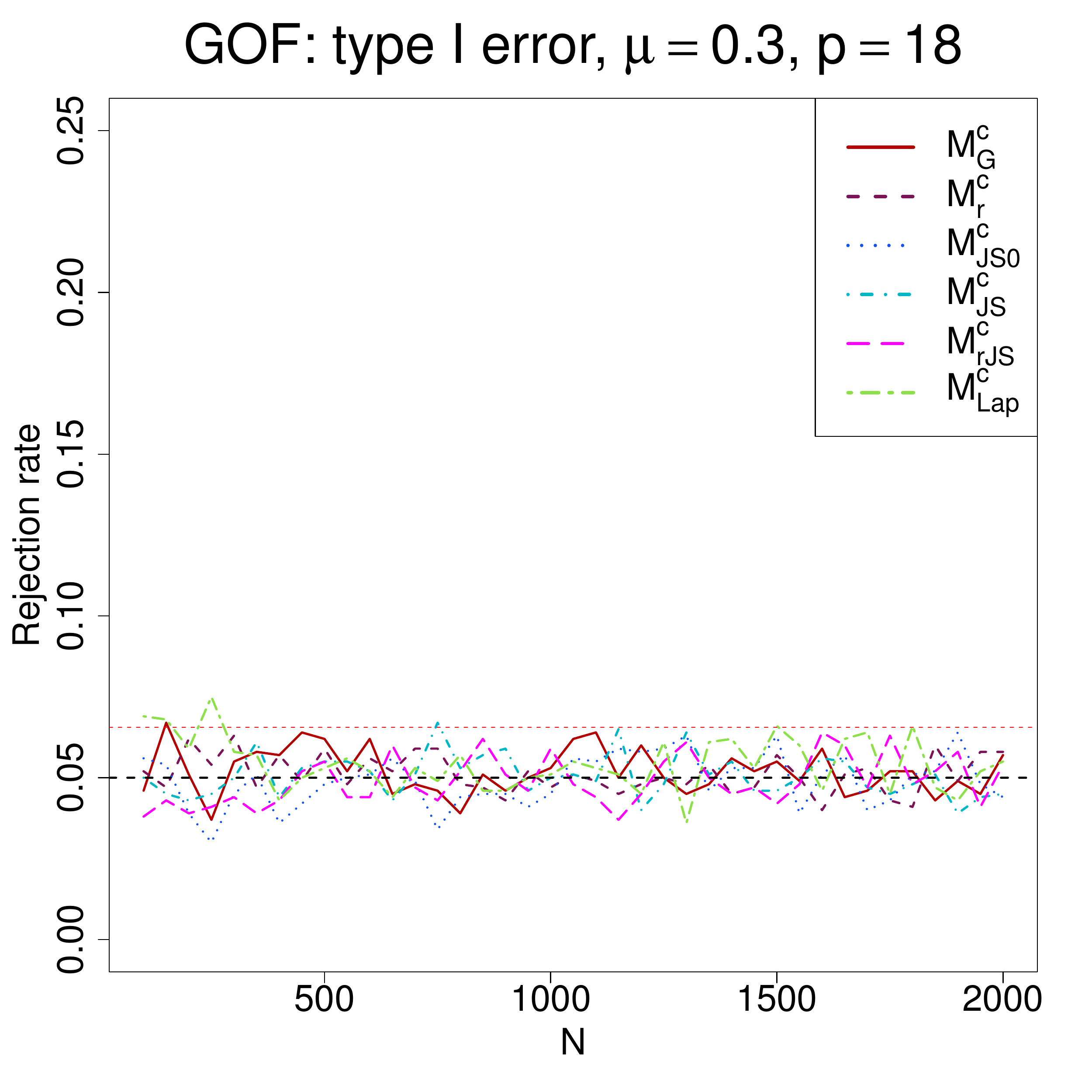}
\end{subfigure}
\caption{Graph of empirical type I error rates of GOF tests for each mechanism. The black dotted line is the level of the test, $\alpha = 0.05$. The red dotted line is the maximum empirical rejection rate value that contains $0.05$ within two times of standard error.}
\label{fig:gof-type1}
\end{figure}

\begin{figure}
\centering
\begin{subfigure}{0.4\textwidth}
    \centering
    \includegraphics[width=.8\textwidth]{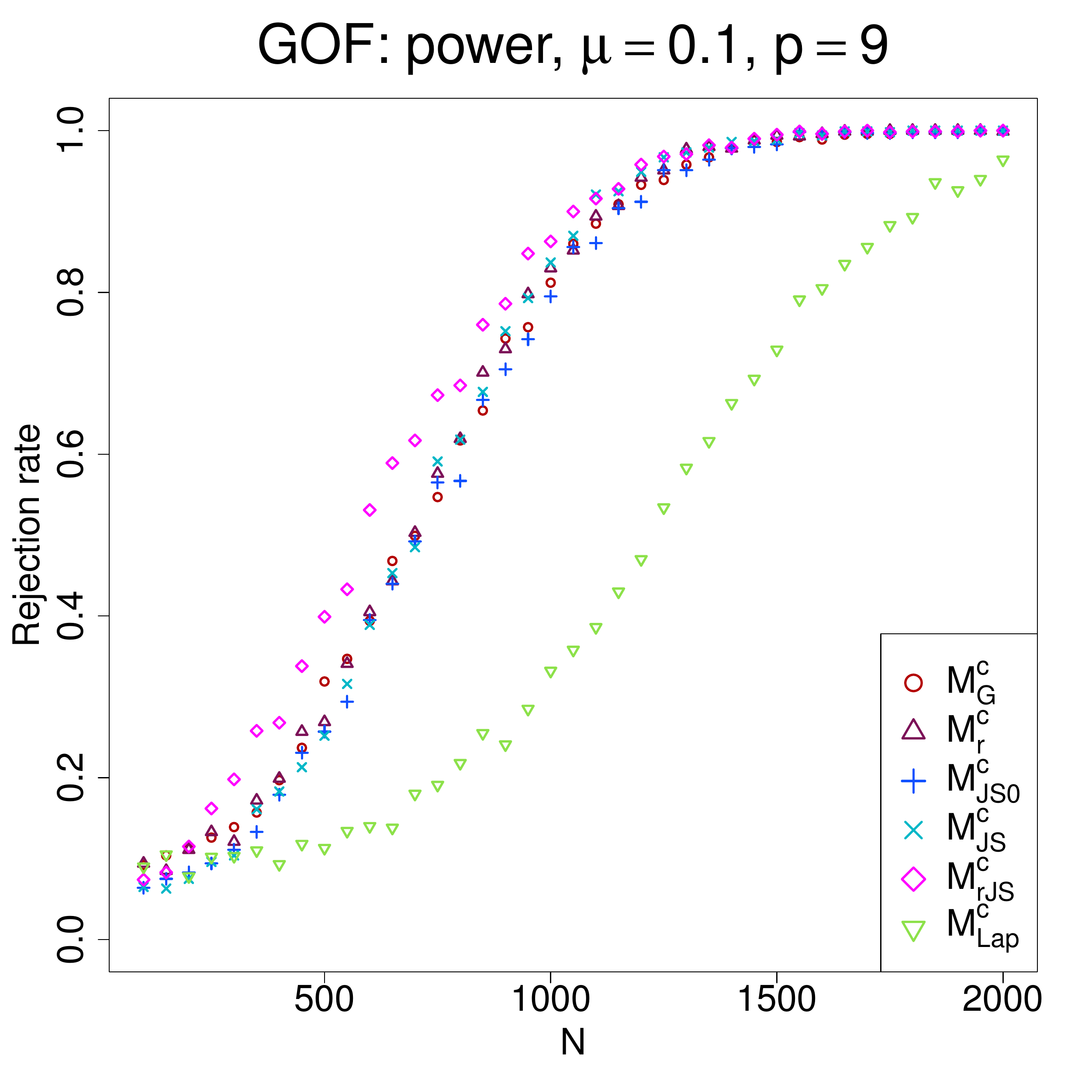}
\end{subfigure}
\begin{subfigure}{0.4\textwidth}
    \centering
    \includegraphics[width=.8\textwidth]{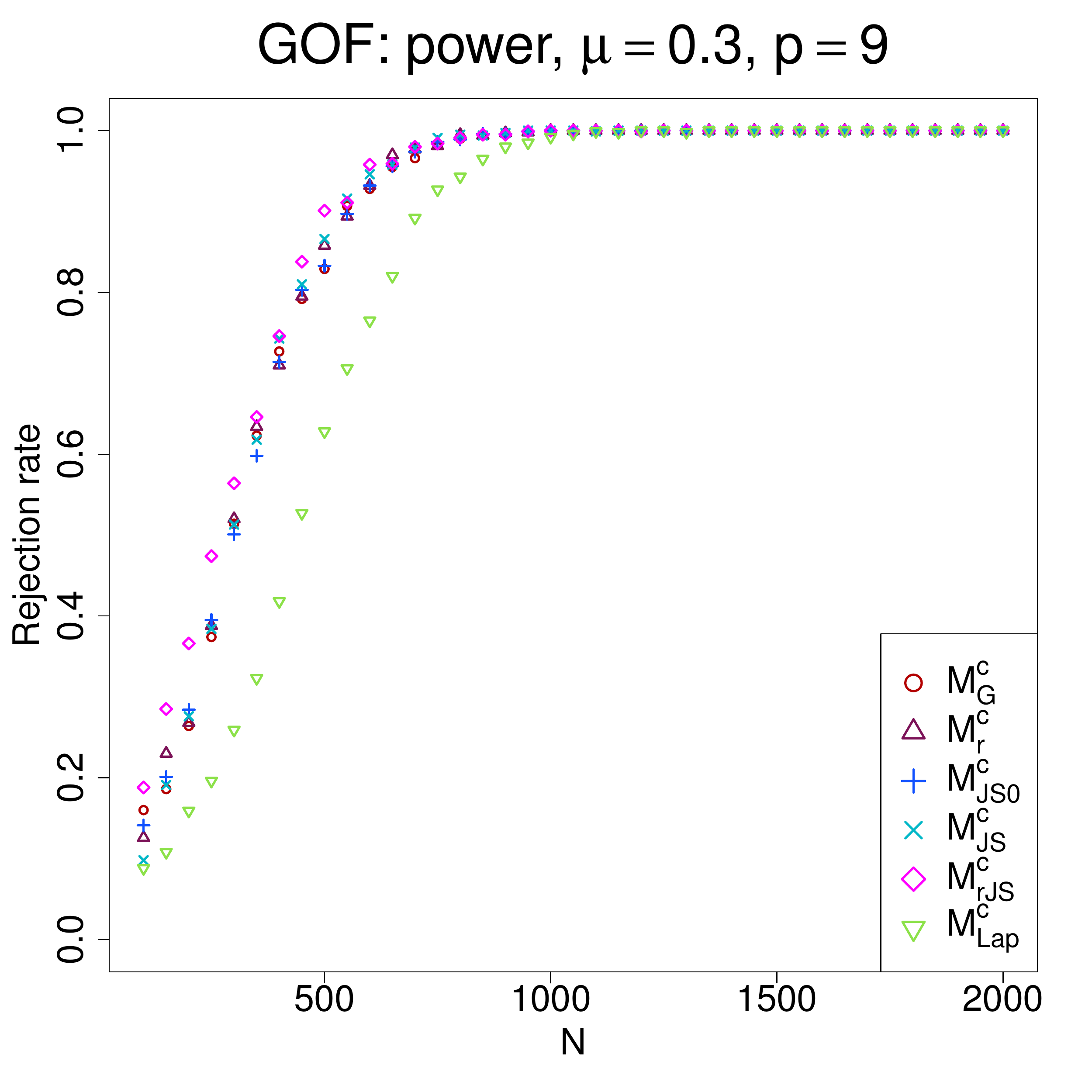}
\end{subfigure}
\begin{subfigure}{0.4\textwidth}
    \centering
    \includegraphics[width=.8\textwidth]{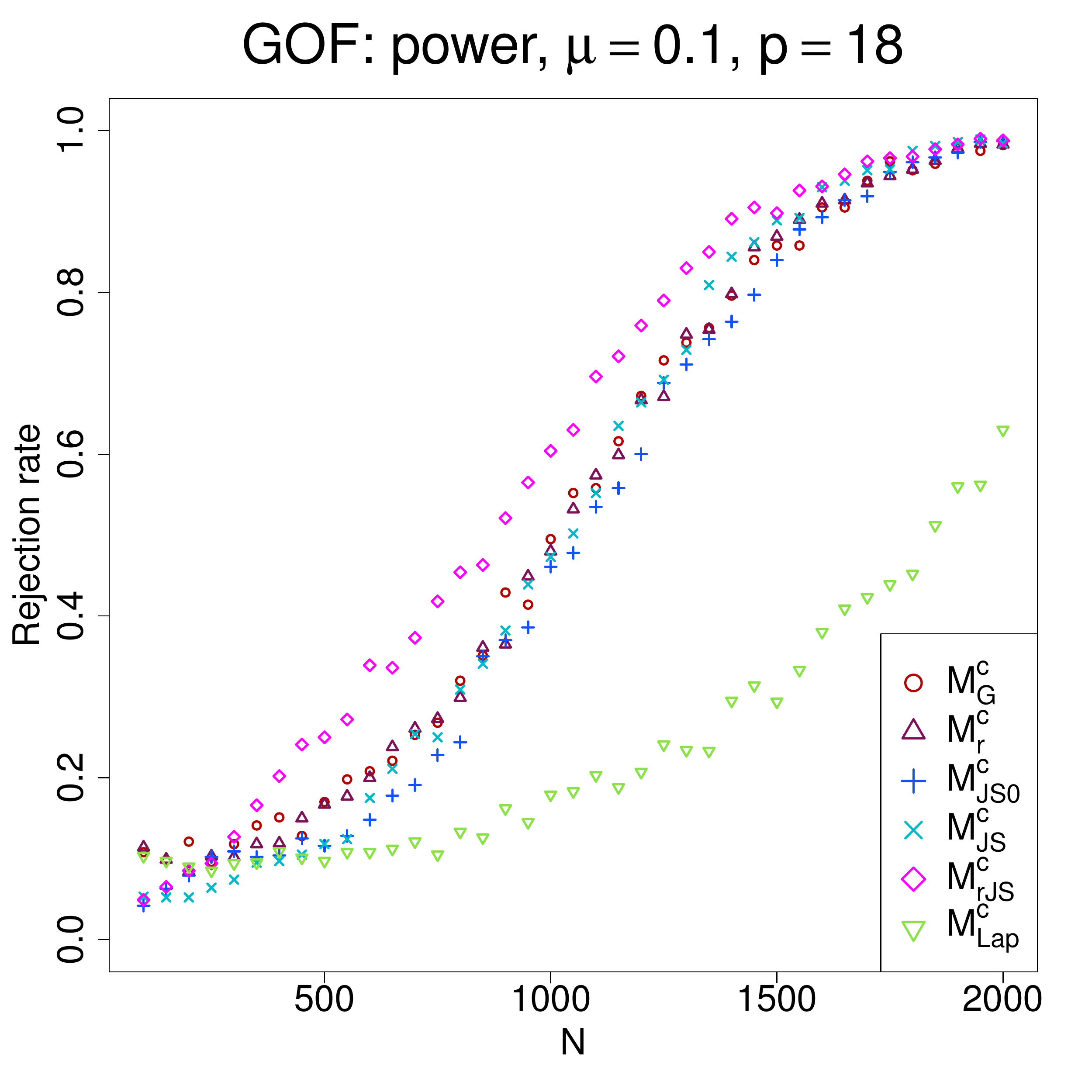}
\end{subfigure}
\begin{subfigure}{0.4\textwidth}
    \centering
    \includegraphics[width=.8\textwidth]{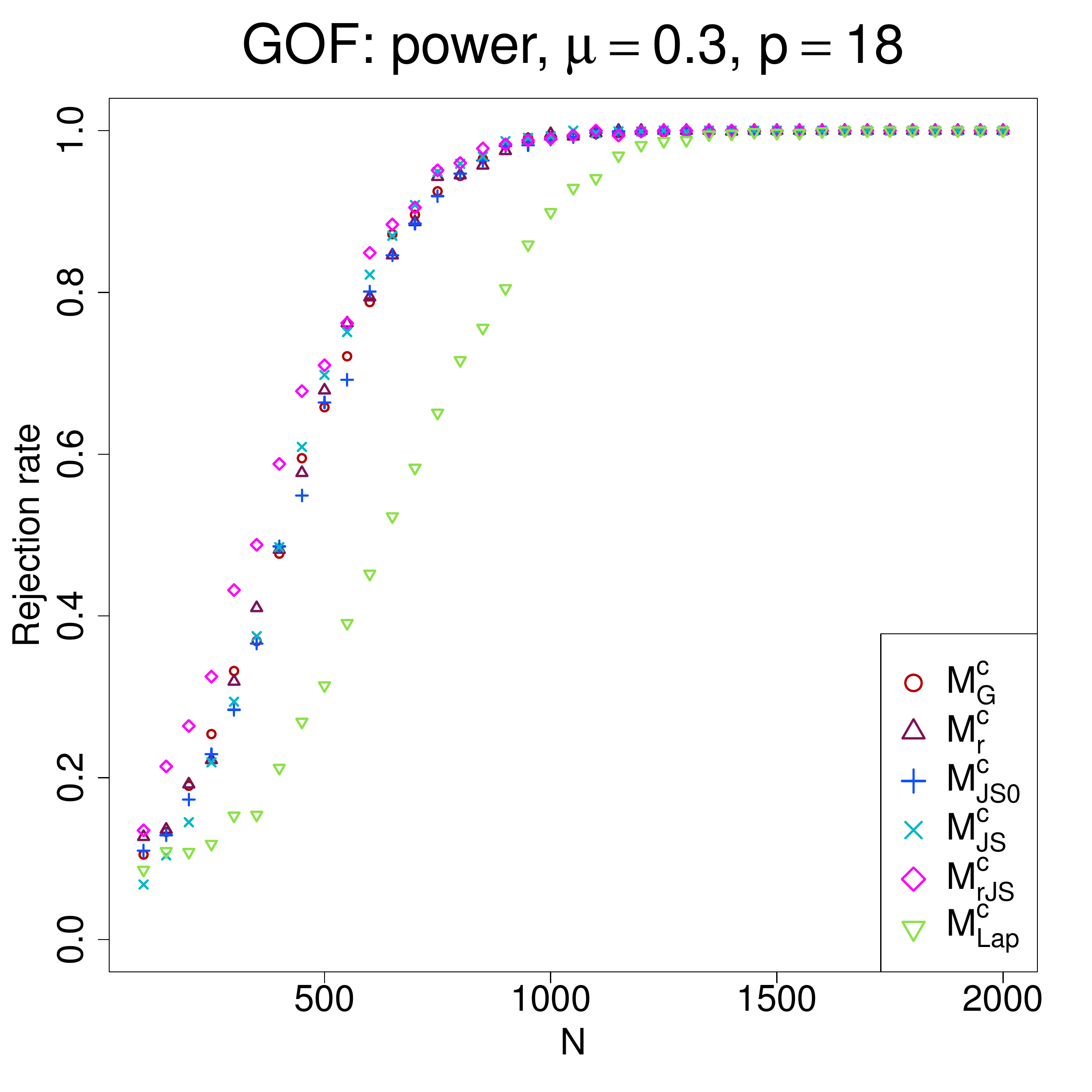}
\end{subfigure}
\caption{Graph of empirical power of GOF tests for each mechanism.}
\label{fig:gof-power}
\end{figure}

Private GOF tests from all mechanisms control the type I error rate at $\alpha = 0.05$ for most cases, as can be seen in Figure \ref{fig:gof-type1}. 
The only exception is the case of $N \le 500$ with $\mu=0.1$ and $p=18$, in which case the amount of perturbation is huge. Nevertheless, all the tests based on the James--Stein mechanisms control the type I error rates successfully. From Figure \ref{fig:gof-power}, we observe that the power is increasing as the sample size increases, for all cases. However, the power of the Laplace mechanism-based test is significantly lower than tests based on Gaussian-type mechanisms, albeit  the Laplace mechanism is tightly calibrated to satisfy $\mu$-GDP. Among Gaussian-type mechanisms, the test based on $\Mv_{rJS}^c$ stands out in terms of empirical power. The superior performance of using $\Mv_{rJS}^c$ over other choices is most pronounced for the high privacy regime ($\mu=0.1$) and for smaller sample sizes ($n \le 1000$).

\section{Discussion} \label{sec:discussion}

In this article, we demonstrated that the rank-deficiency of the sensitivity space $\Sc_{\thetav}$ and James--Stein shrinkage phenomenon can be used to improve the statistical utility of multivariate Gaussian mechanisms, which turns out to be more advantageous than using the optimally calibrated Laplace mechanism. In particular, except for the case where low levels of privacy protection are imposed, ordinary Laplace mechanisms are strictly dominated by ordinary Gaussian mechanisms. Nevertheless, one may naturally ask whether utilizing the rank-deficiency of $\Sc_{\thetav}$ and/or shrinkage phenomenon is beneficial for multivariate Laplace mechanisms. 

Similarly to the Gaussian case, utilizing the rank-deficiency of $\Sc_{\thetav}$ for Laplace mechanisms is also beneficial. To see this, let $\Delta_1(\thetav; \Xc^n) = \Delta_1$ and ${\rm dim}(\Sc_{\thetav}) = d_{\thetav} < p $. 
Let $\Uv = [\Uv_{{\Sc_{\thetav}}}, \Uv_1]$ be a $p \times p $ orthogonal matrix, where ${\rm span}(\Uv_{{\Sc_{\thetav}}}) = {\rm span}(\Sc_{\thetav})$. 
Then a rank-deficient Laplace mechanism $\Mv_{\Lap,r}(\cdot, b) $ may be defined as  
$$\Mv_{\Lap,r}(S; b) =    \Uv_{{\Sc_{\thetav}}}[ \Uv_{{\Sc_{\thetav}}}^\top \thetav(S) + \Lap_{d_{\thetav}}(\0v, b)] +  \Uv_1\Uv_1^\top \thetav(S).$$
Let $\Delta_1^{\Uv} :=  \sup_{\vv \in S_{\thetav}}\|\Uv_{{\Sc_{\thetav}}}^\top \vv\|_1$.
An application of  Theorem~\ref{thm:improved_lap_mech} implies that $\Mv_{\Lap,r}(\cdot, b)$ for $b \ge b_\mu^{\Delta_1^\Uv} :={\Delta_1^\Uv}/[{-2\log\{2\Phi(-\tfrac{\mu}{2})\}}]$ is $\mu$-GDP. 
Since the $L_r$-cost of $\Mv_{\Lap,r}(\cdot; b)$ is proportional to $d_{\thetav} b^r$, 
$\Mv_{\Lap,r}(\cdot; b)$ strictly dominates $\Mv_{\Lap}(\cdot; b)$ for $d_{\thetav} < p$. On the other hand, the smallest scale parameter $b$ of $\Mv_{\Lap,r}$, satisfying $\mu$-GDP, depends on the specific choice of $\Uv_{\Sc_{\thetav}}$. We suspect that $b_\mu^{\Delta_1^\Uv}$ is smaller than $b_\mu^{\Delta_1}$ for some choices of $\Uv_{\Sc_{\thetav}}$, but is not for other choices of $\Uv_{\Sc_{\thetav}}$. We have not pursued further investigation in this direction. This is because that the statistical utility of the ordinary Gaussian mechanism far exceeds that of Laplace mechanisms, and their rank-deficient siblings are deemed to have similar relations.

We believe that an application of shrinkage for multivariate Laplace mechanisms can be beneficial. However, since the $L_1$ shrinkage is more natural for Laplace perturbations, such an investigation falls outside the scope of this article. In addition, we note that the analytical composition bounds through  Edgeworth expansion \citep{zheng2020sharp, wang2022analytical} can be applied in handling Laplace trade-off functions when $p \ge 3$, and the central limit theorem approach seems plausible when $p$ is large \citep{dong2022gaussianDP, dong2021central, sommer2019privacy}. Applications of these approximations in calibration of randomized mechanisms of high-dimensional statistics may be an interesting direction for future research.


\begin{thebibliography}{35}
\newcommand{\enquote}[1]{``#1''}
\expandafter\ifx\csname natexlab\endcsname\relax\def\natexlab#1{#1}\fi

\bibitem[{Abadi et~al.(2016)Abadi, Chu, Goodfellow, McMahan, Mironov, Talwar,
  and Zhang}]{abadi2016deep}
Abadi, M., Chu, A., Goodfellow, I., McMahan, H.~B., Mironov, I., Talwar, K.,
  and Zhang, L. (2016), \enquote{Deep learning with differential privacy,} in
  \textit{Proceedings of the 2016 ACM SIGSAC Conference on Computer and
  Communications Security}, pp. 308--318.

\bibitem[{Anderson(2003)}]{anderson2009introduction}
Anderson, T.~W. (2003), \textit{An Introduction to Multivariate Statistical
  Analysis}, New York: Wiley, 3rd ed.

\bibitem[{Avella-Medina(2021)}]{avella2021privacy}
Avella-Medina, M. (2021), \enquote{Privacy-preserving parametric inference: a
  case for robust statistics,} \textit{Journal of the American Statistical
  Association}, 116, 969--983.

\bibitem[{Awan and Dong(2022)}]{awan2022log}
Awan, J. and Dong, J. (2022), \enquote{Log-concave and multivariate canonical
  noise distributions for differential privacy,} \textit{arXiv preprint
  arXiv:2206.04572}.

\bibitem[{Awan and Slavkovi{\'c}(2021)}]{awan2021structure}
Awan, J. and Slavkovi{\'c}, A. (2021), \enquote{Structure and sensitivity in
  differential privacy: Comparing $K$-norm mechanisms,} \textit{Journal of the
  American Statistical Association}, 116, 935--954.

\bibitem[{Balle and Wang(2018)}]{balle2018improving}
Balle, B. and Wang, Y.-X. (2018), \enquote{Improving the gaussian mechanism for
  differential privacy: Analytical calibration and optimal denoising,} in
  \textit{International Conference on Machine Learning}, PMLR, pp. 394--403.

\bibitem[{Cai et~al.(2021)Cai, Wang, and Zhang}]{cai2021cost}
Cai, T.~T., Wang, Y., and Zhang, L. (2021), \enquote{The cost of privacy:
  Optimal rates of convergence for parameter estimation with differential
  privacy,} \textit{The Annals of Statistics}, 49, 2825--2850.

\bibitem[{Dong et~al.(2022)Dong, Roth, and Su}]{dong2022gaussianDP}
Dong, J., Roth, A., and Su, W. J.~W. (2022), \enquote{Gaussian differential
  privacy,} \textit{Journal of the Royal Statistical Society: Series B
  (Statistical Methodology)}, 84, 3--37.

\bibitem[{Dong et~al.(2021)Dong, Su, and Zhang}]{dong2021central}
Dong, J., Su, W., and Zhang, L. (2021), \enquote{A central limit theorem for
  differentially private query answering,} \textit{Advances in Neural
  Information Processing Systems}, 34, 14759--14770.

\bibitem[{Duchi et~al.(2018)Duchi, Jordan, and Wainwright}]{duchi2018minimax}
Duchi, J.~C., Jordan, M.~I., and Wainwright, M.~J. (2018), \enquote{Minimax
  optimal procedures for locally private estimation,} \textit{Journal of the
  American Statistical Association}, 113, 182--201.

\bibitem[{Dwork et~al.(2006{\natexlab{a}})Dwork, Kenthapadi, McSherry, Mironov,
  and Naor}]{dwork2006our}
Dwork, C., Kenthapadi, K., McSherry, F., Mironov, I., and Naor, M.
  (2006{\natexlab{a}}), \enquote{Our data, ourselves: Privacy via distributed
  noise generation,} in \textit{Annual International Conference on the Theory
  and Applications of Cryptographic Techniques}, Springer, pp. 486--503.

\bibitem[{Dwork et~al.(2006{\natexlab{b}})Dwork, McSherry, Nissim, and
  Smith}]{dwork2006calibrating}
Dwork, C., McSherry, F., Nissim, K., and Smith, A. (2006{\natexlab{b}}),
  \enquote{Calibrating noise to sensitivity in private data analysis,} in
  \textit{Theory of Cryptography Conference}, Springer, pp. 265--284.

\bibitem[{Dwork et~al.(2014)Dwork, Roth, et~al.}]{dwork2014algorithmic}
Dwork, C., Roth, A., et~al. (2014), \enquote{The algorithmic foundations of
  differential privacy,} \textit{Foundations and Trends{\textregistered} in
  Theoretical Computer Science}, 9, 211--407.

\bibitem[{Efron(2012)}]{efron2012large}
Efron, B. (2012), \textit{Large-Scale Inference: Empirical Bayes Methods for
  Estimation, Testing, and Prediction}, Cambridge University Press.

\bibitem[{Gaboardi et~al.(2016)Gaboardi, Lim, Rogers, and
  Vadhan}]{gaboardi2016DPchisq}
Gaboardi, M., Lim, H., Rogers, R., and Vadhan, S. (2016),
  \enquote{Differentially private chi-squared hypothesis testing: Goodness of
  fit and independence testing,} in \textit{International Conference on Machine
  Learning}, PMLR, pp. 2111--2120.

\bibitem[{Geng et~al.(2015)Geng, Kairouz, Oh, and
  Viswanath}]{geng2015staircase}
Geng, Q., Kairouz, P., Oh, S., and Viswanath, P. (2015), \enquote{The staircase
  mechanism in differential privacy,} \textit{IEEE Journal of Selected Topics
  in Signal Processing}, 9, 1176--1184.

\bibitem[{Harville(1998)}]{harville1998matrix}
Harville, D.~A. (1998), \textit{Matrix Algebra from a Statistician's
  Perspective}, Taylor \& Francis.

\bibitem[{James and Stein(1961)}]{james1961estimation}
James, W. and Stein, C. (1961), \enquote{Estimation with quadratic Loss,} in
  \textit{Proceedings of the Fourth Berkeley Symposium on Mathematical
  Statistics and Probability, Volume 1: Contributions to the Theory of
  Statistics}, University of California Press, vol.~4, pp. 361--380.

\bibitem[{Kairouz et~al.(2014)Kairouz, Oh, and Viswanath}]{kairouz2014extremal}
Kairouz, P., Oh, S., and Viswanath, P. (2014), \enquote{Extremal mechanisms for
  local differential privacy,} \textit{Advances in Neural Information
  Processing Systems}, 27, 2879--2887.

\bibitem[{Kairouz et~al.(2015)Kairouz, Oh, and
  Viswanath}]{kairouz2015composition}
--- (2015), \enquote{The composition theorem for differential privacy,} in
  \textit{International Conference on Machine Learning}, PMLR, pp. 1376--1385.

\bibitem[{Karwa et~al.(2017)Karwa, Krivitsky, and
  Slavkovi{\'c}}]{karwa2017sharing}
Karwa, V., Krivitsky, P.~N., and Slavkovi{\'c}, A.~B. (2017), \enquote{Sharing
  social network data: differentially private estimation of exponential family
  random-graph models,} \textit{Journal of the Royal Statistical Society:
  Series C (Applied Statistics)}, 66, 481--500.

\bibitem[{Kifer and Rogers(2016)}]{kifer2016new}
Kifer, D. and Rogers, R. (2016), \enquote{A new class of private chi-square
  tests,} in \textit{Proceedings of the 20th International Conference on
  Artificial Intelligence and Statistics, AISTATS}, vol.~17, pp. 991--1000.

\bibitem[{Lambert(1758)}]{lambert1758observationes}
Lambert, J.~H. (1758), \enquote{Observationes variae in mathesin puram,}
  \textit{Acta Helvetica}, 3, 128--168.

\bibitem[{McSherry and Talwar(2007)}]{mcsherry2007mechanism}
McSherry, F. and Talwar, K. (2007), \enquote{Mechanism design via differential
  privacy,} in \textit{48th Annual IEEE Symposium on Foundations of Computer
  Science (FOCS'07)}, IEEE, pp. 94--103.

\bibitem[{Mironov(2017)}]{mironov2017renyi}
Mironov, I. (2017), \enquote{R{\'e}nyi differential privacy,} in \textit{2017
  IEEE 30th Computer Security Foundations Symposium (CSF)}, IEEE, pp. 263--275.

\bibitem[{Olver et~al.(2010)Olver, Lozier, Boisvert, and Clark}]{olver2010nist}
Olver, F.~W., Lozier, D.~W., Boisvert, R.~F., and Clark, C.~W. (2010),
  \textit{NIST Handbook of Mathematical Functions Hardback and CD-ROM},
  Cambridge University Press.

\bibitem[{Paolella(2007)}]{paolella2007intermediate}
Paolella, M.~S. (2007), \textit{Intermediate Probability: A Computational
  Approach}, John Wiley \& Sons.

\bibitem[{Sommer et~al.(2019)Sommer, Meiser, and Mohammadi}]{sommer2019privacy}
Sommer, D., Meiser, S., and Mohammadi, E. (2019), \enquote{Privacy loss
  classes: The central limit theorem in differential privacy,}
  \textit{Proceedings on Privacy Enhancing Technologies}, 2, 245--269.

\bibitem[{Son et~al.(2022)Son, Park, and Jung}]{son2022parametric}
Son, J., Park, M., and Jung, S. (2022), \enquote{A parametric bootstrap test
  for comparing differentially private histograms,} \textit{The Korean Journal
  of Applied Statistics}, 35, 1--17.

\bibitem[{{Statistics Korea}(2021)}]{korAge2020}
{Statistics Korea} (2021), \enquote{Census questionnaire,}
  \url{https://kosis.kr/statHtml/statHtml.do?orgId=101&tblId=DT_1B040M1&conn_path=I2},
  accessed: 2022-11-26.

\bibitem[{{United States Census Bureau}(2020)}]{americaAge2020}
{United States Census Bureau} (2020), \enquote{2020 Demographic Analysis
  estimates by age and sex,}
  \url{https://www.census.gov/data/tables/2020/demo/popest/2020-demographic-analysis-tables.html},
  {D}ecember 2020 release.

\bibitem[{Wang et~al.(2022)Wang, Gao, Zhang, Shen, and Su}]{wang2022analytical}
Wang, H., Gao, S., Zhang, H., Shen, M., and Su, W.~J. (2022),
  \enquote{Analytical composition of differential privacy via the Edgeworth
  accountant,} \textit{arXiv preprint arXiv:2206.04236}.

\bibitem[{Wang et~al.(2017)Wang, Lee, and Kifer}]{wang2017revisiting}
Wang, Y., Lee, J., and Kifer, D. (2017), \enquote{Revisiting differentially
  private hypothesis tests for categorical data,} \textit{arXiv preprint
  arXiv:1511.03376}.

\bibitem[{Wasserman and Zhou(2010)}]{wasserman2010statistical}
Wasserman, L. and Zhou, S. (2010), \enquote{A statistical framework for
  differential privacy,} \textit{Journal of the American Statistical
  Association}, 105, 375--389.

\bibitem[{Zheng et~al.(2020)Zheng, Dong, Long, and Su}]{zheng2020sharp}
Zheng, Q., Dong, J., Long, Q., and Su, W. (2020), \enquote{Sharp composition
  bounds for Gaussian differential privacy via Edgeworth expansion,} in
  \textit{International Conference on Machine Learning}, PMLR, pp.
  11420--11435.

\end{thebibliography}

\newpage

\appendix 
\section{Technical tools for working with trade-off functions}\label{sec:appendix_lemmas_trade-off}

Oftentimes it is convenient to write trade-off functions with respect to random variables.
For random variables $ X \sim P$ and $ Y \sim Q $, we write $T(X, Y) \coloneqq T(P, Q)$.
If the distributions $ P $ and $ Q $ have densities $ f $ and $ g $, we also denote
$ T(f, g) \coloneqq T(P, Q)$.

Trade-off functions are invariant under simultaneous one-to-one affine transformations:
\begin{lem}\label{lem:Tf-linearity} Suppose $f$ and $g$ are density functions. For $p$-variate random vectors $\xv \sim f$ and $\yv \sim g$, and for any full-rank $p\times p$ matrix $\Av$ and $\bv \in \Real^p$,
   $$T(\Av\xv + \bv, \Av\yv + \bv) = T(\xv,\yv) = T(f,g).$$
\end{lem}

\begin{proof}[Proof of Lemma~\ref{lem:Tf-linearity}]
By the Neyman-Pearson lemma, the most powerful rejection region corresponding to $T(\xv,\yv)$ is $R_\alpha = \{ \wv : \log (g(\wv) / f(\wv)) \ge c_\alpha\}$ where $c_\alpha$ satisfies $\mathbb P_{\Mv \sim f}(\Mv \in R_\alpha) = \alpha$. %
The density functions of $\Av\xv+\bv$ and $\Av\yv+\bv$ are $\frac{1}{|\Av|}f(\Av^{-1}(\cdot - \bv))$ and $\frac{1}{|\Av|}g(\Av^{-1}(\cdot - \bv))$, respectively. Then the rejection region corresponding to $T(\Av\xv+\bv,\Av\yv+\bv)$ at level $\alpha$ is
$$R_\alpha' = \{\zv: \log \frac{g(\Av^{-1}(\zv - \bv))}{f(\Av^{-1}(\zv - \bv))} \ge c_\alpha'\},$$
where $c_\alpha'$ satisfies $\mathbb P(\Zv \in R_\alpha' ) = \alpha$ for $\Zv \stackrel{d}= \Av\xv + \bv$. %
It can be checked that $c_\alpha' = c_\alpha$, and that $\mathbb P_{\Mv \sim g}(\Av\Mv+\bv \notin R_\alpha' ) = \mathbb P_{\Mv \sim g}(\Mv \notin R_\alpha )$ for all $\alpha \in (0,1)$. [This lemma was used implicitly in \cite{dong2022gaussianDP}.]
\end{proof}

The next few results are useful when dealing with Gaussian mechanisms.

\begin{lem}\label{lem:mv-Tf-normal} Let $\muv, \muv_1, \muv_2 \in \Real^p$ be arbitrary and $\Sigmav$ be a $p\times p$ symmetric positive-definite matrix. Then,
\begin{itemize}
  \item[(i)]
$T(N_p(\0v,\Id_p), N_p( \muv, \Id_p)) = T(N(0,1), N(\norm{\muv}_2, 1)) = G_{\norm{\muv}_2}$.

  \item[(ii)]  $T(N_p(\muv_1, \Sigmav), N_p(\muv_2, \Sigmav))  = G_{\|\Sigmav^{-1/2}(\bm{\mu}_2 - \bm{\mu}_1)\|_2}$.
\end{itemize}
\end{lem}
\begin{proof}[Proof of Lemma~\ref{lem:mv-Tf-normal}]
(i) By the Neyman--Pearson lemma, the most powerful rejection region corresponding to $T(N_p(\0v,\Id_p), N_p( \muv, \Id_p))$ is
$R_\alpha = \{ \xv \in \Real^p : \xv^\top \muv \ge  \norm{\muv}_2 \Phi^{-1}(1-\alpha)\}$. 
Now, one can easily check that the type II error rate is $\mathbb{P}_{\xv \sim N_p(\muv, \Id_p)}(\xv \notin R_\alpha) = G_{\norm{\muv}_2}(\alpha)$.

(ii) By Lemma \ref{lem:Tf-linearity}, we have
$T(N_p(\muv_1, \Sigmav), N_p(\muv_2, \Sigmav)) = T(N_p(\0v, \Id_p), N_p(\Sigmav^{-1/2}(\muv_2 - \muv_1), \Id_p))$. Applying part (i) gives the desired result.
\end{proof}

An immediate consequence of Lemmas~\ref{lem:Tf-linearity} and \ref{lem:mv-Tf-normal} is the following result on the privacy guarantee of ordinary Gaussian mechanisms.

\begin{thm}[Ordinary Gaussian mechanism] \label{thm:G-mech}
Let $\thetav$ be a multivariate statistic with $\Delta_2(\thetav; \Xc^n) = \Delta_2$.
The ordinary Gaussian mechanism $\Mv_G$, given by $\Mv_G(S)  = \thetav(S) + N_p(\0v, \sigma^2 \Id_p)$,
is $\mu$-GDP for any $\mu \ge \mu_0 := \Delta_2 / \sigma$.
\end{thm}

\begin{proof}[Proof of Theorem \ref{thm:G-mech}]
For any data sets $S$ and $S'$, $\Mv_G(S) \sim N_p(\thetav(S), \sigma^2\Id_p)$ and $\Mv_G(S') \sim N_p(\thetav(S'),\sigma^2\Id_p)$. By Lemmas~\ref{lem:Tf-linearity} and \ref{lem:mv-Tf-normal},
$$T(\Mv_G(S), \Mv_G(S')) = G_{\frac{\norm{ \thetav(S)-\thetav(S')}_2}{\sigma}}.$$
Since   $0 \le a \le b$ if and only if $G_a \ge G_{b}$,
$$
\inf_{S\sim S'}T(\Mv_G(S), \Mv_G(S'))  = \inf_{S\sim S'}G_{\frac{\norm{ \thetav(S)-\thetav(S')}_2}{\sigma}}
 = G_{\sup_{S\sim S'}\frac{\norm{ \thetav(S)-\thetav(S')}_2}{\sigma}} =
 G_{ \frac{\Delta_2}{\sigma}} =  G_{ \mu_0}.
$$
Thus, for any $\mu \ge\mu_0$,
$\inf_{S\sim S'}T(\Mv_G(S), \Mv_G(S'))  \ge G_{\mu}$.
\end{proof}

Care is needed when we compare two multivariate normal distributions with a \emph{rank-deficient} covariance matrix.
\begin{lem}\label{lem:rank-deficient}
Let $\xv \sim N_p(\muv_1, \Sigmav)$ and $\yv \sim N_p(\muv_2, \Sigmav)$, where $\Sigmav$ is a $p\times p$ symmetric semi-definite matrix with ${\rm rank}(\Sigmav) < p$. Let $\Cc(\Sigmav) \subset \Real^p$ be the column space of the matrix $\Sigmav$.
 \begin{itemize}
   \item [(a)] If $\muv_2 - \muv_1 \in \Cc(\Sigmav)$, then $T(\xv,\yv) = G_{\| (\Sigmav^\dagger)^{\frac{1}{2}} (\muv_2 -\muv_1) \|_2 }$, where $\Sigmav^\dagger$ is the Moore-Penrose pseudo-inverse of $\Sigmav$.
   \item [(b)] If $\muv_2 - \muv_1 \notin \Cc(\Sigmav)$, then $T(\xv,\yv)(\alpha) = 0$ for all $\alpha\in [0,1]$.
 \end{itemize}
\end{lem}

\begin{proof}
By Lemma \ref{lem:Tf-linearity}, $T(\xv, \yv) = T( N_p(\muv_1, \Sigmav), N_p(\muv_2, \Sigmav)) = T( N_p(\0v, \Sigmav), N_p(\muv, \Sigmav))$ for $\muv = \muv_2 - \muv_1$. Without loss of generality,  we let $\muv_1 = \0v$ and $\muv_2 = \muv$.

Let $\Sigmav = \Uv \Lambdav \Uv^\top$ be the eigendecomposition of $\Sigmav$ such that $\Lambdav = \mbox{diag}(\lambda_1,\ldots,\lambda_q)$ for $q = {\rm rank}(\Sigmav) < p$. Let $\Uv_2$ be such that $[\Uv, \Uv_2]$ is an orthogonal matrix. Then for
\begin{align*}
  &\zv_1 = \Uv^\top \xv,  \quad    \wv_1 = \Uv^\top \yv, \\
  &\zv_2 = \Uv_2^\top \xv,\quad   \wv_2 = \Uv_2^\top \yv,
\end{align*}
we have
\begin{align*}
  &\zv_1 \sim N_q(\0v, \Lambdav), \quad   \wv_1 \sim N_q(\Uv^\top \muv, \Lambdav), \\
  &\zv_2 =  \0v,  \quad  \quad \quad \quad \ \wv_2 = \Uv_2^\top \muv.
\end{align*}
(a) Suppose $\muv \in \Cc(\Sigmav)$, that is, $\Uv_2^\top \muv = 0$. Then, testing
$$H_0: N_p(\0v, \Sigmav) \quad vs. \quad H_1: N_p(\muv, \Sigmav)$$
based on one observation from either distribution is equivalent to testing
$$H_0: N_q(\0v, \Lambdav) \quad vs. \quad H_1: N_q(\Uv^\top \muv, \Lambdav)$$
based on one observation from either distribution. Thus,
\begin{align*}
  T(\xv,\yv) & = T(\zv_1, \wv_1)  \\
             & = T( N_q(\0v, \Lambdav), N_q(\Uv^\top \muv, \Lambdav) ) \\
             & = T( N(0,1), N( \| \Lambdav^{-1/2} \Uv^\top \muv \|_2, 1) ).
\end{align*}
Since $ \| \Lambdav^{-1/2} \Uv^\top \muv \|_2 =  \| \Uv \Lambdav^{-1/2} \Uv^\top \muv \|_2 = \| (\Sigmav^\dagger)^{\frac{1}{2}} \muv \|_2$, the result follows.

(b) Suppose $\muv \notin \Cc(\Sigmav)$.
Then $\Uv_2^\top \muv \neq 0$.
This means that $\zv_2 =  \0v$ and  $\wv_2 = \Uv_2^\top \muv$ are distinguishable with probability 1.
Consider a test function $\phi$ defined by
$$
\phi(\xv)
=
\begin{cases}
\mbox{reject}
& \mbox{if } (\Uv_2^\top\mu)^{\top} (\Uv_2^\top\xv - \Uv_2^\top\mu/2) > 0, \\
\mbox{do not reject} & \mbox{otherwise}.
\end{cases}
$$
Then the type I and II error rates are calculated as
$\alpha_\phi = \P( -\|\wv_2\|_2^2/2 > 0 ) = 0$ and
$\beta_\phi = \P( \|\wv_2\|_2^2/2 \le 0 ) = 0$.
Since the trade-off function is non-increasing, we have
$$
T(\xv, \yv) \le T(\xv, \yv)(0) = T(\xv, \yv)(\alpha_{\phi}) = \beta_{\phi} = 0.  %
$$
This completes the proof.
\end{proof}

Finally, we formally state the post-processing property for the $f$-DP criterion, discussed in Sections~\ref{sec:2.2} and \ref{sec:2.3}.

\begin{lem}[Post-processing property, \cite{dong2022gaussianDP}]\label{lem:post-processing}
Let $f$ be a trade-off function. If a mechanism $M : \Xc^n \rightarrow \Yc$ is $f$-DP, then $Proc \circ M$ is also $f$-DP for any post-processing function $Proc : \Yc \rightarrow \Zc$ where $\Zc$ is some abstract space.
\end{lem}

\section{Proofs for Section~\ref{section: mv-mech}}

\begin{proof}[Proof of Lemma~\ref{lem:iso_Gauss}]
Assume $\lambda_{\min}(\Sigmav) \ge \Delta_2^2 / \mu^2 > 0$. Then $\Sigmav$ is positive definite, and for any $\thetav$ with $\Delta_2(\thetav; \Xc^n) \le \Delta_2$, Lemma~\ref{lem:mv-Tf-normal} implies that 
\begin{align*} %
\begin{split}
    &\inf_{S \sim S'} T(\Mv_{\Sigmav}(S), \Mv_{\Sigmav}(S'))\\
    &= \inf_{S \sim S'} T(N(0,\Sigmav), N(\thetav(S) - \thetav(S'), \Sigmav))\\
    &= \inf_{S \sim S'} T(N(0, 1), N(\norm{\Sigmav^{-1/2} (\thetav(S) - \thetav(S'))}_2, 1))\\
    &= T(N(0, 1), N(\sup_{S \sim S'} \norm{\Sigmav^{-1/2} (\thetav(S) - \thetav(S'))}_2, 1))\\
    &\ge T(N(0, 1), N(\norm{\Sigmav^{-1/2}}_2 \sup_{S \sim S'} \norm{(\thetav(S) - \thetav(S'))}_2, 1))\\
    &= T(N(0,1), T(\frac{\Delta_2(\thetav; \Xc^n)}{\sqrt{\lambda_{\rm min}(\Sigmav)}},1). 
\end{split}
\end{align*}
Since $\frac{\Delta_2(\thetav, \Xc^n)}{\sqrt{\lambda_{\rm min}(\Sigmav)}} \le \frac{\Delta_2}{\sqrt{\lambda_{\rm min}(\Sigmav)}} \le \mu$, we have $\inf_{S \sim S'} T(\Mv_{\Sigmav}(S), \Mv_{\Sigmav}(S')) \ge G_\mu$, as required. 
 
The `only if` part can be verified by showing the contrapositive: If $\lambda_{\min}(\Sigmav) < \Delta_2^2/\mu^2$, then there exists a statistic $\thetav$ with $\Delta_2(\thetav; \Xc^n) \le \Delta_2$ such that $\Mv_{\Sigmav}$ does not satisfy $\mu$-GDP. 

Let $\bv \in \Real^p$ be an eigenvector of $\Sigmav$ satisfying $\Sigmav \bv = \lambda_{\min}(\Sigmav) \bv$ and $\|\bv\|_2 = 1$. Let $\thetav$ be such that for some $S,S' \in \Xc^n$, $S\sim S'$, $\thetav(S) - \thetav(S') = \Delta_2 \bv$. That is, $\Delta_2 \bv \in \Sc_{\thetav}$.
The trade-off function corresponding to $\Mv_{\Sigmav}$ is   
\begin{align*}
    \inf_{S \sim S'} T(\Mv_{\Sigmav}(S), \Mv_{\Sigmav}(S')) 
    = \inf_{\deltav \in \Sc_{\thetav} } T(N_p( \deltav , \Sigmav), N_p(\0v, \Sigmav)) 
    \le T(N_p( \Delta_2 \bv  , \Sigmav), N_p(\0v, \Sigmav)).
\end{align*} 

If $\lambda_{\min}(\Sigmav) = 0$, then by Lemma~\ref{lem:rank-deficient}, 
$T(N_p( \Delta_2 \bv  , \Sigmav), N_p(\0v, \Sigmav))(\alpha) = 0 < G_\mu (\alpha)$ for all $\alpha \in (0,1)$. Thus, $\Mv_{\Sigmav}$ does not satisfy $\mu$-GDP (for any $\mu >0$).

Suppose that $\lambda_{\min}(\Sigmav) > 0$. Then,  we have 
\begin{align*}
  \norm{\Sigmav^{-1/2}\Delta_2 \bv}_2 & =
     \norm{\Sigmav^{-1/2}}_2 \norm{\Delta_2 \bv}_2 \\ 
     & = \norm{\Sigmav^{-1/2}}_2  \Delta_2 \\
     & = \Delta_2 / \sqrt{\lambda_{\min}(\Sigmav)} > \mu,
\end{align*} 
which in turn implies that 
$T(N_p( \Delta_2 \bv  , \Sigmav), N_p(\0v, \Sigmav)) =  T(N(0,1), N(\norm{\Sigmav^{-1/2}\Delta_2 \bv}_2, 1)) < G_\mu$, as desired.  
\end{proof}

\begin{proof}[Proof of Lemma~\ref{lem:iso_Gauss_stronger}]
Lemma~\ref{lem:iso_Gauss} implies that if $\sigma^2  \ge \Delta_2^2/\mu^2$, 
$\Mv_\Sigmav$ is $\mu$-GDP.

For the `only if' part, suppose that $\sigma^2 < \frac{\Delta_2^2}{\mu^2}$. Then, for any $\thetav$ with $\Delta_2(\thetav; \Xc^n) = \Delta_2$,
\begin{align*} %
    &\inf_{S \sim S'} T(\Mv_{\Sigmav}(S), \Mv_{\Sigmav}(S'))\\
    &= T(N(0, 1), N(\sup_{S \sim S'} \norm{\sigma^{-1} (\thetav(S) - \thetav(S'))}_2, 1))\\
    &= G_{\Delta_2 / \sigma}\\
    &< G_\mu
\end{align*}
on $(0,1)$.
\end{proof}

\begin{proof}[Proof of Theorem~\ref{thm:rank-deficient-G-mech}]
For any $S\sim S'$, we have $\thetav(S) - \thetav(S') \in \Cc(\Pv_{\thetav})$.
Thus, by Lemma \ref{lem:rank-deficient} (a),
\begin{align*}
    T(\Mv_r(S), \Mv_r(S'))
    &= G_{\norm{((\sigma^2 \Pv_{\thetav})^\dagger)^{1/2} (\thetav(S) - \thetav(S'))}_2}\\
    &= G_{\frac{1}{\sigma} \norm{(\Pv_{\thetav})^{1/2} (\thetav(S) - \thetav(S'))}_2}\\
    &= G_{\frac{1}{\sigma} \norm{\thetav(S) - \thetav(S')}_2}.
\end{align*}

Therefore, we have
$$
    \inf_{S \sim S'} T(\Mv_r(S), \Mv_r(S'))
    = G_{\frac{1}{\sigma} \cdot \sup_{S \sim S'} \norm{\thetav(S) - \thetav(S')}_2}
    = G_{\frac{\Delta_2(\thetav; \Xc^n)}{\sigma}}.
$$
The result follows as $G_\mu$ strictly and uniformly (on $(0,1)$) decreases with respect to $\mu$.
\end{proof}

Next, we provide a proof of Proposition~\ref{prop:M45_costs}. We first list the related result on the conditional $L_2$-cost of $\Mv_{rJS}$ in Proposition \ref{lem:M6_cost}, and give proofs of Propositions~\ref{prop:M45_costs} and \ref{lem:M6_cost}.

\begin{prop}\label{lem:M6_cost} For a multivariate statistic $\thetav: \Xc^n \to \Real^p$ with $\Delta_2(\thetav; \Xc^n) = \Delta_2 > 0$, let $\Mv_{rJS}$  be the rank-deficient James--Stein Gaussian mechanisms with $\sigma>0$, as defined in Definition~\ref{def:rJS}. For any data set $S \in \Xc^n$,
    \begin{align} \label{eq:M6cost}
        L_2(\Mv_{rJS}(S)) = \sigma^2 d_{\thetav} - \sigma^2(d_{\thetav}-3)^2 \left\{ \frac{1}{2} e^{-\frac{1}{2}\widetilde{\tau}(S)} \sum_{\beta = 0}^{\infty} \left( \frac{\widetilde{\tau}(S)}{2} \right)^\beta \frac{1}{\beta! \{\frac{1}{2}(d_{\thetav} - 1) + \beta - 1\}} \right\} ,
    \end{align}
    where
$\widetilde{\tau}(S) =
\| (\Id_{d_{\thetav}} - \frac{1}{d_{\thetav}} \1v_{d_{\thetav}} \1v_{d_{\thetav}}^\top) \Uv_{\thetav}^\top \thetav(S) \|_2^2 / {\sigma^2}$.
\end{prop}
 
\begin{proof}[Proofs of Propositions~\ref{prop:M45_costs} and \ref{lem:M6_cost}]
For some deterministic $\etav \in \Real^p$, suppose that a random vector $\xv$ follows $N_p(\etav, \sigma^2 \Id_p)$.  
First, we verify \eqref{eq:M4cost} by showing that 
\begin{equation}\label{eq:JS0_MSE}
    \E \norm{m_0(\xv;\sigma) - \etav}_2^2 = \sigma^2  p - \sigma^2 (p-2)^2 \left\{ \frac{1}{2} e^{-\frac{1}{2}\tau} \sum_{\beta = 0}^{\infty} \left( \frac{\tau}{2} \right)^\beta \frac{1}{\beta! (\frac{1}{2}p + \beta - 1)} \right\},
\end{equation}
where $ \tau = {\norm{\etav}_2^2}/{\sigma^2}$.
Let $\gv(\tv) = -\frac{(p-2)}{\|\tv\|^2} \tv$ for $\tv \in \Real^p$. We have
\begin{align*}
    \frac{m_0(\xv;\sigma)}{\sigma} = \frac{\xv}{\sigma} + \gv\left(\frac{\xv}{\sigma}\right),
\end{align*}
with
$\xv / \sigma \sim N_p(\etav / \sigma, \Id_p)$. Using Stein's unbiased risk estimate lemma, we obtain
\begin{align*}
    \E\| m_0(\xv;\sigma) - \etav \|^2
    &= p\sigma^2 + \sigma^2 \E\left[ \left\|\gv \left(\frac{\xv}{\sigma}\right)\right\|^2 + 2 \nabla^\top \gv\left(\frac{\xv}{\sigma}\right) \right] \nonumber \\
    &= p\sigma^2 - (p-2)^2 \sigma^2 \E\left[ \frac{\sigma^2}{\|\xv\|^2} \right].
\end{align*}
Since $\frac{\| \xv \|^2}{\sigma^2} \sim \chi^2(\tau, p)$, we have \citep{anderson2009introduction}
\begin{align*}
    \E\left[ \frac{\sigma^2}{\|\xv\|^2} \right] = \frac{1}{2} e^{-\frac{1}{2}\tau} \sum_{\beta = 0}^{\infty} \left( \frac{\tau}{2} \right)^\beta \frac{1}{\beta! (\frac{1}{2}p + \beta - 1)}.
\end{align*}
This gives \eqref{eq:JS0_MSE}.

Now, we prove \eqref{eq:M5cost} by showing that
\begin{equation}\label{eq:JS_MSE}
    \E \norm{m_a(\xv;\sigma) - \etav}_2^2 = \sigma^2 p - \sigma^2(p-3)^2 \left\{ \frac{1}{2} e^{-\frac{1}{2}\tau_{\1v, \perp}} \sum_{\beta = 0}^{\infty} \left( \frac{\tau_{\1v, \perp}}{2} \right)^\beta \frac{1}{\beta! \{\frac{1}{2}(p - 1) + \beta - 1\}} \right\},    
\end{equation}
where $\tau_{\1v, \perp} =  {\|(\Id_p - \frac{1}{p} \1v_p \1v_p^\top)\etav\|_2^2}/{\sigma^2}$.

Take a $p \times p$ orthogonal matrix $\Uv = [\uv_1 \; \Uv_2]$ with $\uv_1 \coloneqq \frac{1}{\sqrt{p}} \1v_p$. Also, let $\yv = \begin{bmatrix}
    y_1\\
    \yv_2
\end{bmatrix} \coloneqq \begin{bmatrix}
    \uv_1^\top \xv\\
    \Uv_2^\top \xv
\end{bmatrix} = \Uv^\top \xv \in \Real^p$ with $y_1 \in \Real$, $\yv_2 \in \Real^{p-1}$ so that

\begin{align*}
    \yv =
    \begin{bmatrix}
    y_1\\
    \yv_2
    \end{bmatrix} =
    \begin{bmatrix}
    \sqrt{p} \overline{x}\\
    \Uv_2^\top \xv
    \end{bmatrix}
    \sim N_p \left(
    \begin{bmatrix}
    \eta_1\\
    \etav_2
    \end{bmatrix}
    , \sigma^2 \Id_p \right),
\end{align*}
where
\begin{align*}
    \eta_1 = \frac{1}{\sqrt{p}} \1v_p^\top \etav \in \Real, \; \etav_2 = \Uv_2^\top \etav \in \Real^{p-1}.
\end{align*}
Then,
\begin{align*}
    \Uv^\top m_a(\xv;\sigma) &=
    \begin{bmatrix}
    \uv_1^\top\\
    \Uv_2^\top
    \end{bmatrix}
    \left( \overline{x} \1v_p + \left( 1 - \frac{(p-3)\sigma^2}{\| \xv_c \|^2} \right) \xv_c \right)\\
    &=
    \begin{bmatrix}
    \frac{\overline{x}}{\sqrt{p}} \1v_p^\top \1v_p + \left( 1 - \frac{(p-3)\sigma^2}{\| \xv_c \|^2} \right) (\frac{1}{\sqrt{p}} \1v_p^\top \xv - \frac{\overline{x}}{\sqrt{p}} \1v_p^\top \1v_p)\\
    \overline{x} \Uv_2^\top \1v_p + \left( 1 - \frac{(p-3)\sigma^2}{\| \xv_c \|^2} \right) (\Uv_2^\top \xv - \overline{x} \Uv_2^\top \1v_p)
    \end{bmatrix}\\
    &=
    \begin{bmatrix}
    \sqrt{p} \overline{x}\\
    \left( 1 - \frac{(p-3)\sigma^2}{\| \xv_c \|^2} \right) \Uv_2^\top \xv
    \end{bmatrix}\\
    &=
    \begin{bmatrix}
    y_1\\
    \left( 1 - \frac{(p-3)\sigma^2}{\| \xv_c \|^2} \right) \yv_2
    \end{bmatrix}.
\end{align*}
However,
\begin{align*}
    \| \xv_c \|^2 &= \left\| \left(\Id_p - \frac{1}{p} \1v_p \1v_p^\top \right) \xv \right\|^2
    = \left\| (\Uv \Uv^\top - \uv_1 \uv_1^\top ) \xv \right\|^2
    = \left\| \Uv_2 \Uv_2^\top \xv \right\|^2
    = \left\| \Uv_2^\top \xv \right\|^2
    = \left\| \yv_2 \right\|^2.
\end{align*}
It implies that
$$
    \left( 1 - \frac{(p-3)\sigma^2}{\| \xv_c \|^2} \right) \yv_2
    = \left( 1 - \frac{\{(p-1) - 2\}\sigma^2}{\| \yv_2 \|^2} \right) \yv_2
    = m_0(\yv; \sigma)
$$
with $\yv_2 \sim N_{p-1}(\etav_2, \sigma^2 \Id_{p-1})$. Therefore, we obtain
\begin{align*}
    \E\| m_a(\xv;\sigma) - \etav \|^2
    &= \E\| \Uv^\top m_a(\xv;\sigma) - \Uv^\top \etav \|^2 \nonumber\\
    &= \E(y_1 - \eta_1)^2 + \E\left\| \left( 1 - \frac{(p-3) \sigma^2}{\| \yv_2 \|^2} \right) \yv_2 - \etav_2 \right\|^2 \nonumber\\
    &= \sigma^2 + \E \norm{m_0(\yv; \sigma) - \etav_2}_2^2.
\end{align*}
Here, \eqref{eq:JS0_MSE} with $\norm{\etav_2}_2^2 = \|(\Id_p - \frac{1}{p} \1v_p \1v_p^\top) \etav\|_2^2$ gives \eqref{eq:JS_MSE}.

For Proposition \ref{lem:M6_cost}, note that
\begin{align*}
    L_2(\Mv_{rJS}(S))
    &= \E \left[ \norm{[U_1 \; U_{\thetav}]^\top (\Mv_{rJS}(S) - \thetav(S))}_2^2 \; \middle| \; S \right]\\
    &= \E \left[ \norm{m_a(\Uv_{\thetav}^\top \Mv_r(S) ; \sigma) - \Uv_{\thetav}^\top \thetav(S)}_2^2 \; \middle| \; S \right]
\end{align*}
and $\Uv_{\thetav}^\top \Mv_r(S) \sim N_{d_{\thetav}} (\Uv_{\thetav}^\top \thetav(S), \sigma^2 \Id_{d_{\thetav}})$.
Therefore, by \eqref{eq:JS_MSE}, we have \eqref{eq:M6cost}.
\end{proof}
   
\begin{proof}[Proof of Theorem~\ref{thm:JS-mech}]

Let $p \ge 3$. Since for any $S \in \Xc^n$, $L_2(\Mv_G(S)) = p\sigma^2$, it can checked using \eqref{eq:M4cost} that $L_2(\Mv_{JS0}(S)) > L_2(\Mv_G(S))$ for any $S$. Thus $\Mv_{JS0}$ strictly dominates $\Mv_G$. Proof for (ii) is similarly obtained. 
\end{proof}

\begin{proof}[Proof of Theorem~\ref{thm:rJS-mech}]
The conclusion is drawn by observing that $$L_2(\Mv_{rJS}(S)) < L_2(\Mv_{r}(S)) < L_2(\Mv_{G}(S))$$ for any $S$, where
$L_2(\Mv_{rJS}(S))$ is given in Proposition~\ref{lem:M6_cost}. 
\end{proof}

\begin{proof}[Proof of Theorem~\ref{thm:L2cost-conv}]

We give proof of (ii) and (iii) only, since (i) can be shown in a similar manner. In this proof, we wrote $\thetav_n = \thetav(S_n)$.

(ii) Let $\Pi_{1}^{\perp} = \Id_p - \frac{1}{p} \1v_p \1v_p^\top$.
A simple but key observation is that
\begin{equation}\label{eq:diff-Mjs} %
    \Mv_{G} - \Mv_{JS}
    = \frac{(p-3)\sigma^2}{\|\Pi_{1}^{\perp}\Mv_G \|_2^2}\Pi_{1}^{\perp}\Mv_G.
\end{equation}
From this, we have
\begin{equation*}
\| \Mv_{G} - \Mv_{JS} \|_2^2 / \sigma^2
=  \frac{(p-3)^2}{\|\Pi_{1}^{\perp}\Mv_G / \sigma \|_2^2}.
\end{equation*}
Since $\Mv_G(S_n)$ follows $N_p(\thetav_n, \sigma^2 I_p)$,
$\|\Pi_{1}^{\perp}\Mv_G / \sigma \|_2^2$ follows the non-central $\chi^2$ distribution
$\chi^2(m, 2\lambda_N)$ where $m = p-1$ is a degree of freedom and
$\lambda_N = \frac{1}{2\sigma^2} \thetav_N^T \Pi_1^{\perp} \thetav_N$ is the non-centrality parameter.
The expectation of non-central $\chi^2$ distribution can be expressed by
confluent hypergeometric function ${ }_1 F_1$.
From (10.9) of \cite{paolella2007intermediate},
it holds that
\begin{equation*}
\mathbb E \left\{\frac{1}{\|\Pi_{1}^{\perp}\Mv_G / \sigma \|_2^2}\right\}
= \frac{1}{2e^{\lambda_n}}\frac{\Gamma(\frac{p-3}{2})}{\Gamma(\frac{p-1}{2})}
{ }_1 F_1(\frac{p-3}{2}, \frac{p-1}{2}; \lambda_n)
= \frac{\Gamma(\frac{p-3}{2})}{2\Gamma(\frac{p-1}{2})}
{ }_1 F_1(1, \frac{p-1}{2}; -\lambda_n).
\end{equation*}
For the second equality, we use Kummer's transformation:
${ }_1 F_1(a, b; z) = e^{z}{ }_1 F_1(b-a, b; -z)$.
Then the integral form of ${ }_1 F_1$, (5.27) of \cite{paolella2007intermediate},
implies that
\begin{equation}\label{eq:pf7-2_int_form}
\frac{\Gamma(\frac{p-3}{2})}{2\Gamma(\frac{p-1}{2})} { }_1 F_1(1, \frac{p-1}{2}; -\lambda_n)
= \frac{1}{2} \int_0^1 (1-y)^{\frac{p-5}{2}} e^{-\lambda_n y} dy.
\end{equation}
First, we consider the case of $p \ge 5$.
Since $(1-x)^r \le e^{-rx}$ for any $x \in [0, 1]$ and $r \ge 0$, \eqref{eq:pf7-2_int_form} has an upper bound as follows.
\begin{equation} \label{eq:pf7-2_ineq_p5}
\mathbb E \left\{\frac{1}{\|\Pi_{1}^{\perp}\Mv_G / \sigma \|_2^2}\right\}
= \frac{1}{2} \int_0^1 (1-y)^{\frac{p-5}{2}} e^{-\lambda_n y} dy
\le \frac{1}{2(\frac{p-5}{2} + \lambda_n)}(1 - e^{-(\frac{p-5}{2} + \lambda_n)}).
\end{equation}
If $p=4$, \eqref{eq:pf7-2_int_form} is expressed as
\begin{equation} \label{eq:pf7-2_ineq_p4}
\mathbb E \left\{\frac{1}{\|\Pi_{1}^{\perp}\Mv_G / \sigma \|_2^2}\right\}
= \frac{1}{2} \int_0^1 (1-y)^{-1/2} e^{-\lambda_n y} dy
= \frac{F(\sqrt{\lambda_n})}{\sqrt{\lambda_n}}
< \frac{1 - e^{-\lambda_n}}{\sqrt{\lambda_n}},
\end{equation}
where $F(x) = e^{-x^2}\int_0^x e^{y^2} dy$ is a Dawson's integral.
The last inequality of \eqref{eq:pf7-2_ineq_p4} comes from 7.8.7 of \cite{olver2010nist}.
Note that $p$ and $\mu$ are fixed in our setting.
So, the right-hand side of the inequalities \eqref{eq:pf7-2_ineq_p5} and \eqref{eq:pf7-2_ineq_p4} converge to $0$ as $n \to \infty$ under
the condition $\|\Pi_1^{\perp} \thetav_n / \Delta_2\|_2^2 \to \infty$.
This proves (ii) of the theorem.

(iii) Note that
$$
\Mv_r - \Mv_{rJS} = (\Uv_1 \Uv_1^\top + \Uv_{\thetav} \Uv_{\thetav}^\top) \Mv_r - \Mv_{rJS} = \Uv_{\thetav} (\Uv_{\thetav}^\top \Mv_r - m_a(\Uv_{\thetav}^\top \Mv_r; \sigma)).
$$
Let $\widetilde \Mv_r \coloneqq \Uv_{\thetav}^T \Mv_r \sim N_q(\Uv_{\thetav}^T \thetav_n, \sigma^2 I_{d_{\thetav}})$.
Since the multiplication of the orthogonal matrix preserves the $l_2$-norm, we have
$$
\|\Mv_r - \Mv_{rJS}\|_2^2 = \|\widetilde \Mv_r - m_a(\widetilde \Mv_r; \sigma)\|_2^2.
$$
From the definition of JS-process,
$$
\|\widetilde \Mv_r - m_a(\widetilde \Mv_r; \sigma)\|_2^2
= \widetilde \Mv_r^T \Pi_{1_{d_{\thetav}}}^{\perp} \widetilde \Mv_r
$$
and
\begin{equation*}
\|\Mv_r - \Mv_{rJS}\|_2^2/\sigma^2 \sim \chi^2(\tilde m, 2\tilde \lambda),
\end{equation*}
where $\tilde m = d_{\thetav}-1$ and
$\tilde \lambda =\frac{1}{2\sigma^2} \thetav_n^T\Pi_{\Uv_{\theta}} \thetav_n$.
Thus, by applying a similar argument from the proof of (ii), we get (iii).

\end{proof}

\begin{proof}[Proof of Corollary~\ref{cor:tothm:L2cost-conv}]
 We begin with writing $\av_n = \Mv_{rJS}(S_n) -\thetav(S_n)$, $\rv_n = \Mv_{r}(S_n) -\thetav(S_n)$, $\gv_n = \Mv_{G}(S_n) -\thetav(S_n)$ and $\jv_n = \Mv_{JS}(S_n) -\thetav(S_n)$. Then, by the Cauchy-Schwarz inequality, we have
\begin{align*}
  {L_2(\Mv_{rJS}(S_n))}/{\sigma_n^2} & = \E \| \av_n \|_2^2 /\sigma_n^2 \\
   & =  \E \| \av_n - \rv_n + \rv_n\|_2^2 /\sigma_n^2 \\
   & \le \frac{\E \| \av_n - \rv_n \|_2^2}{\sigma_n^2} + \frac{ \E\| \rv_n \|_2^2}{\sigma_n^2}  + 2\left(
         \frac{\E \| \av_n - \rv_n \|_2^2}{\sigma_n^2}\frac{ \E\| \rv_n \|_2^2}{\sigma_n^2} \right)^{\frac{1}{2}}.
\end{align*} Observe that
$\lim_{n\to\infty} {\E \| \av_n - \rv_n \|_2^2}/{\sigma_n^2} = 0$ by Theorem~\ref{thm:L2cost-conv}, and
${ \E\| \rv_n \|_2^2}{\sigma_n^2}  =  L_2( \Mv_r(S_n)) / \sigma_n^2 =d_{\thetav}$ for every $n$. Thus,
\begin{equation}\label{eq:proof1forcor:tothm:L2cost-conv}
  \limsup_{n\to\infty} {L_2(\Mv_{rJS}(S_n))}/{\sigma_n^2} \le d_{\thetav}
\end{equation}
 Similarly,
  \begin{align*}
  {L_2(\Mv_{JS}(S_n))}/{\sigma_n^2}  & =\E \| \jv_n \|_2^2 /\sigma_n^2\\
   & \ge \frac{ \E\| \gv_n \|_2^2}{\sigma_n^2} - \frac{\E \| \jv_n - \gv_n \|_2^2}{\sigma_n^2} - 2\left(
         \frac{\E \| \jv_n - \gv_n \|_2^2}{\sigma_n^2}\frac{ \E\| \gv_n \|_2^2}{\sigma_n^2} \right)^{\frac{1}{2}},
\end{align*}
and $\lim_{n\to\infty} {\E \| \jv_n - \gv_n \|_2^2}/{\sigma_n^2} = 0$ by Theorem~\ref{thm:L2cost-conv}, and
${ \E\| \gv_n \|_2^2}{\sigma_n^2}  =  L_2( \Mv_G(S_n)) / \sigma_n^2 = p$ for every $n$. Thus, $\liminf_{n\to\infty} {L_2(\Mv_{JS}(S_n))}/{\sigma_n^2} \ge p$. This, together with (\ref{eq:proof1forcor:tothm:L2cost-conv}), gives the conclusion since $p > d_{\thetav}$. The argument above works for $\Mv_{JS}$ replaced by $\Mv_{JS0}$ as well.
\end{proof}

\section{Proofs and technical details for Section~\ref{sec:Lap_mech}}\label{sec:appendix_proofs}

In Section~\ref{sec:appendix_proofs1}, we give exact expressions of trade-off functions between univariate and bivariate Laplace distributions, and the trade-off functions corresponding to Laplace mechanisms for the cases $p = 1$ (Lemma~\ref{lem:univariate_laplace_tradeoff})  and $p = 2$ (Lemma~\ref{lem:bivariate_laplace_tradeoff_worst}). Section~\ref{sec:appendix_proofs2} contains the proofs for Section~\ref{sec:Lap_mech}.
Throughout the section, we denote  $$T_{\mathrm{Lap}_p, \deltav} \coloneqq T(\Lap_p(\0v, 1), T(\Lap_p(\deltav, 1))$$ for any $\deltav \in \Real^p$.

\subsection{Trade-off functions between univariate and bivariate Laplace distributions}\label{sec:appendix_proofs1}
 
The exact expression of the trade-off function between univariate Laplace distributions, with a common scale parameter, is given in Proposition A.6 of \cite{dong2022gaussianDP}, 
which is derived from the log-concavity of the Laplace distribution. 
Below, we give an alternative proof of Proposition A.6 of \cite{dong2022gaussianDP}, by directly using the Neyman-Pearson lemma. It will be instructive in deriving the trade-off function between bivariate Laplace distributions.

First, we describe a general framework of using the Neyman-Pearson lemma in arbitrary dimension.
For $\deltav \in \Real^p$, the trade-off function $T(\mathrm{Lap}_p(\0v, 1), \mathrm{Lap}_p(\deltav, 1))$ is invariant with respect to the sign changes of $\delta_i$'s. 
Moreover, switching the order of $\delta_i$'s also does not affect the value of the trade-off function.
Thus, we only need to concern the case where $\delta_1 \ge \cdots \ge \delta_p \geq 0$.  

Recall that the trade-off function's value at $\alpha \in [0, 1]$ is the type II error rate of the most powerful test with significance level $\alpha$ on the hypotheses
$$
    H_0: \Xv \sim \mathrm{Lap}_p(\0v, 1) \quad vs. \quad H_1: \Xv \sim \mathrm{Lap}_p(\deltav, 1).
$$
The log-likelihood ratio function is
$$
l(\xv)
\coloneqq \log \frac{f(\xv|\deltav, b)}{f(\xv|\0v, b)}
= \sum_{i=1}^{p} (|x_i| - |x_i - \delta_i|),
$$
where $f(\xv|\deltav, b) = (2b)^{-p}\exp\left(-\frac{1}{b} \sum_{i=1}^p |x_i - \delta_i|\right)$ is the density function
of $\Lap_p(\deltav, b)$.
By the Neyman-Pearson lemma, the most powerful test at level $\alpha$ has the form of
$$
    \phi_{K,\gamma}(\Xv) = \mathds{1}(l(\Xv) > K) + \gamma \cdot \mathds{1}(l(\Xv) = K)
$$
for some $K$ and $\gamma \in [0, 1]$.
Here, $\phi_{K, \gamma}$ is a randomized test function, and $K$ and $\gamma$ are determined by $\alpha$ to satisfy $\alpha = \E_{\Xv \sim \mathrm{Lap}_p(\0v, 1)}
(\phi_{K,\gamma}(\Xv))$.
Randomization is necessary for any $\deltav$ as there exists $K$ such that $\P(l(\Xv) = K) > 0$.
For example, for $\delta = 1$, $l(X)$ has a nontrivial point mass at $K 
 = \pm 1$ as shown in Figure \ref{fig:l(x)}.
 
\begin{figure}[t]
\centering
\begin{subfigure}{0.4\textwidth}
    \centering
    \includegraphics[width=1.0\linewidth]{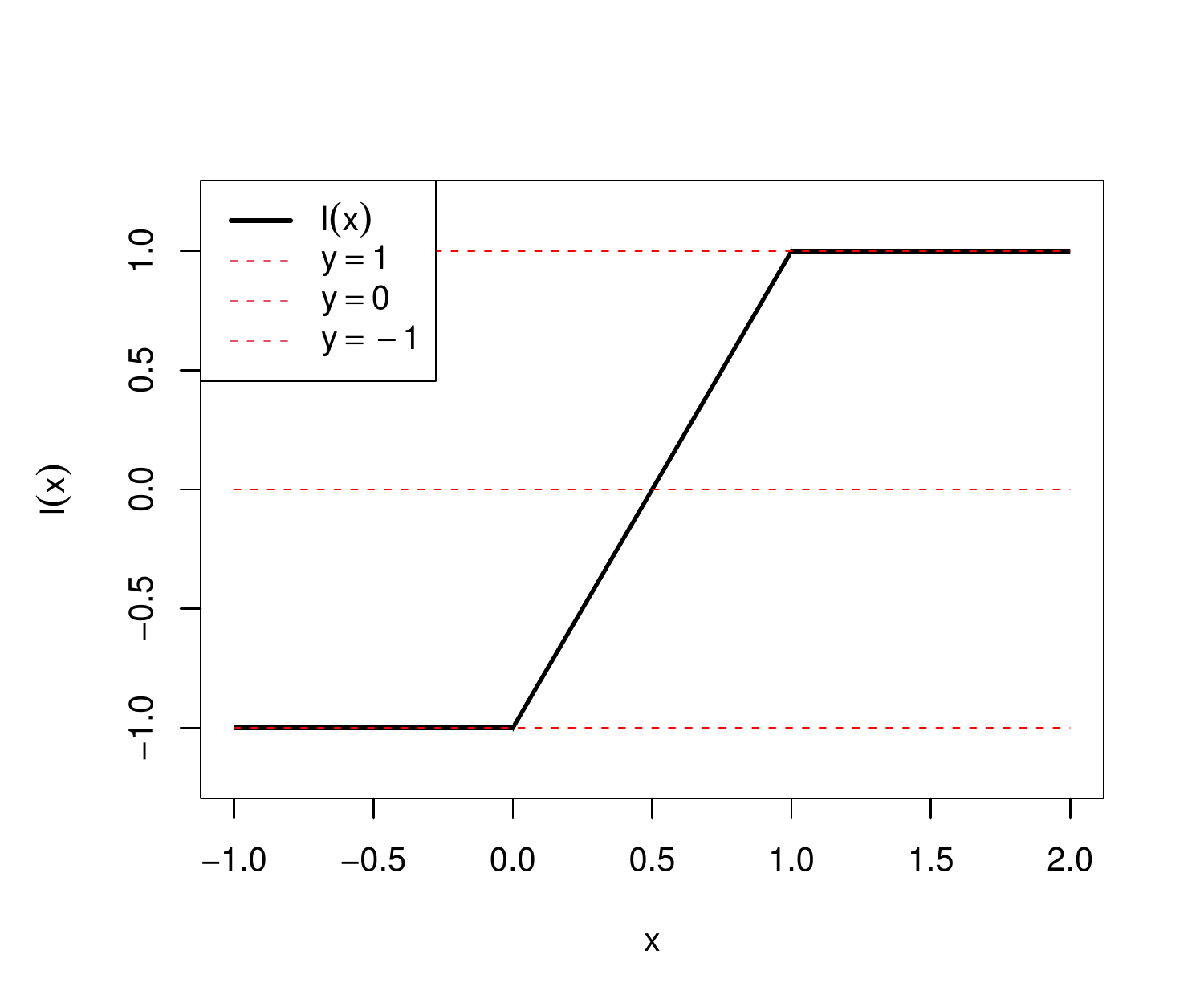}
    \caption{Graph of $l(x)$ for univariate case with $\delta = 1$.}
    \label{fig:l(x):univariate}
\end{subfigure}
\hspace{2em}
\begin{subfigure}{0.4\textwidth}
    \centering
    \includegraphics[width=1.0\linewidth]{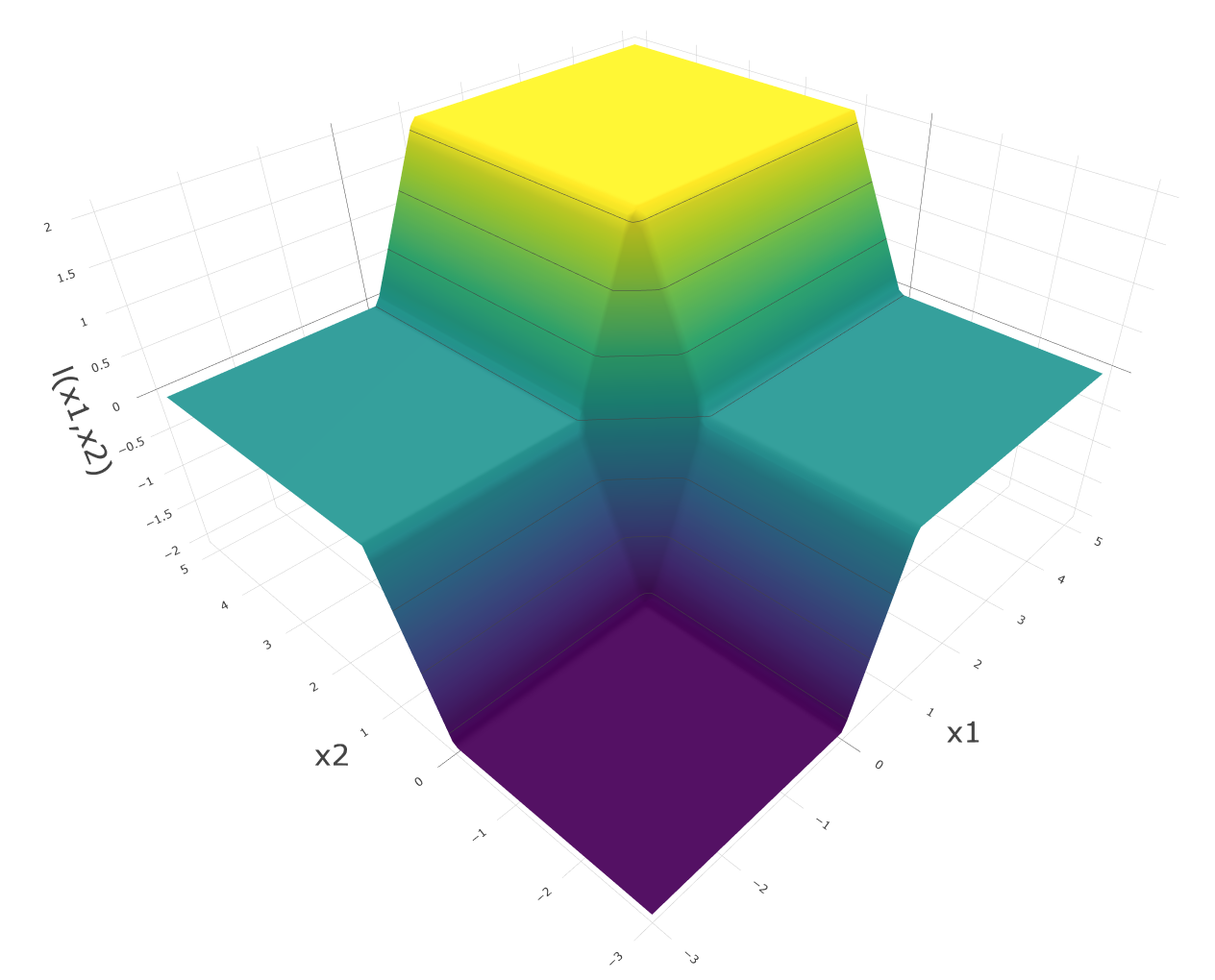}
    \caption{Graph of $l(x_1, x_2)$ for bivariate case with $\delta_1 = \delta_2 = 1$.}
    \label{fig:l(x):bivariate}
\end{subfigure}
\caption{Surface plots for $l(\xv).$}
\label{fig:l(x)}
\end{figure}

The trade-off function between the two $p$-dimensional Laplace distributions can be written
\begin{align}\label{eq:Lap_general_tradeoff}
\begin{split}
    T_{\mathrm{Lap}_p, \deltav}(\alpha)
    &= \E_{\Xv \sim \mathrm{Lap}_p(\deltav, 1)}(1 - \phi_{K,\gamma}(\Xv))\\
    &= \P_{\Xv \sim \mathrm{Lap}(\deltav, 1)}\left(l(\Xv) < K   \right) + (1-\gamma) \cdot \P_{\Xv \sim \mathrm{Lap}(\deltav, 1)} \left( l(\Xv) = K  \right).
\end{split}
\end{align}

For $p=1$ and $p=2$, we derive a closed form of \eqref{eq:Lap_general_tradeoff}. Evaluating \eqref{eq:Lap_general_tradeoff} for $p \ge 3$ can be done, but its expression turns out to be too complex and has little practical relevance.

\subsubsection{Univariate case}\label{subsubsec:1dim_Lap}

We first consider the case of $p=1$ to calibrate the univariate Laplace noise to $L_1$ sensitivity of $\theta$ by calculating the trade-off function between univariate Laplace distributions. By \eqref{eq:Lap_general_tradeoff}, we can derive the trade-off function between univariate Laplace distributions. First, the log-likelihood ratio is
\begin{align*}
    l(x) = \abs{x} - \abs{x - \delta}
    = \begin{cases}
        \delta & \text{if $x \ge \delta$},\\
        2x - \delta & \text{if $0 \le x < \delta$},\\
        -\delta & \text{otherwise}.
    \end{cases}
\end{align*}
Originally, $\phi_{K, \gamma}$ was determined as the most powerful level-$\alpha$ test given $\alpha$. Conversely, if $K$ and $\gamma \in [0,1]$ are given, $\alpha$ is uniquely determined. With this idea, we can derive the trade-off function as follows.

\begin{enumerate}[(i)]
    \item $K > \delta$:
    The most powerful test is given as $\phi \equiv 0$.
    In this case, the corresponding type I error is
    $0$.
    So, the trade-off function value at $\alpha = 0$ is
    $T_{\mathrm{Lap}_1, \delta}(0) = 1$.

    \item $K = \delta$:
    Since $\P(l(X) = K) > 0$, we need randomization.
    Fix a $\gamma \in [0, 1]$.
    The test $\phi_{K, \gamma}$ is the most powerful test at level $\alpha_{\gamma}$ where
    \begin{align*}
        \alpha_{\gamma} = \gamma \cdot \P_{X \sim \mathrm{Lap}_1(0, 1)}(X \ge \delta) = \gamma \cdot \frac{1}{2}e^{-\delta}
        \in [0, \frac{1}{2}e^{-\delta}].
    \end{align*}
    Conversely, for $0 < \alpha \le \frac{1}{2}e^{-\delta}$,
    let $\gamma_{\alpha} = 2e^{\delta}\alpha \in (0, 1]$.
    Then the most powerful test at level $\alpha \in (0, \frac{1}{2}e^{-\delta}]$ is $\phi_{K, \gamma_{\alpha}}$ and the trade-off function value at $\alpha$ is
    \begin{align*}
    T_{\mathrm{Lap}_1, \delta}(\alpha)
    &= \P_{X \sim \Lap_1(\delta, 1)}(X < \delta ) + (1 - \gamma_{\alpha}) \cdot \P_{X \sim \Lap_1(\delta, 1)}(X \ge \delta ) \\
    &= 1 - \frac{1}{2} \gamma_{\alpha} = 1 - e^\delta \alpha.
    \end{align*}

    \item $-\delta < K < \delta$:
    Since $\P(l(X) = K) = 0$, $\phi_{K, 0} = \mathds{1}(l(\Xv) > K) $ is the most powerful test at level $\alpha_K$ where
    \begin{align*}
        \alpha_K = \P_{X \sim \Lap_1(0, 1)} \left( 2X - \delta > K   \right) = \frac{1}{2} e^{- \frac{K + \delta}{2}} \in (\frac{1}{2}e^{-\delta}, \frac{1}{2}).
    \end{align*}
    Conversely, for $\frac{1}{2}e^{-\delta} < \alpha < \frac{1}{2}$, $\phi_{K_{\alpha}, 0}$ is the most powerful test at level $\alpha$ where $K_{\alpha} = -\delta - 2\log(2\alpha) \in (-\delta, \delta)$
    and the trade-off function value at $\alpha$ is
    $$
        T_{\mathrm{Lap}_1, \delta}(\alpha) = \P_{X \sim \Lap_1(\delta, 1)}(2X - \delta < K_{\alpha} ) = \frac{1}{2} e^{\frac{K_{\alpha} - \delta}{2}} = \frac{e^{-\delta}}{4 \alpha}.
    $$
\end{enumerate}
By (i)-(iii) and the symmetry of the trade-off function, $T_{\mathrm{Lap}_1, \delta}$ is given as

\begin{align}\label{eq:Lap_tradeoff_1dim_}
    T_{\mathrm{Lap}_1, \delta}(\alpha) =
    \begin{cases}
      1 - e^\delta \alpha & \text{if $\alpha < \frac{1}{2}e^{-\delta}$},\\
      \frac{e^{-\delta}}{4\alpha} & \text{if $\frac{1}{2}e^{-\delta} \leq \alpha < \frac{1}{2}$},\\
      e^{-\delta}(1-\alpha) & \text{otherwise.}
    \end{cases}
\end{align}

\begin{remark} \label{rmk:log-conc-tf}
The above expression \eqref{eq:Lap_tradeoff_1dim_} was firstly presented in Proposition A.6 of \cite{dong2022gaussianDP}.
In general, $T(\xi, t + \xi)(\alpha) = F(F^{-1}(1-\alpha) - t)$ holds for every $t > 0$ if and only if the density of a random variable $\xi$ is log-concave (Proposition A.3 of \cite{dong2022gaussianDP}). 
Here, $F$ is the distribution function of $\xi$.
Since the density of univariate Laplace distribution is log-concave,
the direct calculation of $F(F^{-1}(1-\alpha) - t)$ when $\xi \sim \Lap(0, 1)$ results in $T_{\Lap_1, t}$ of \eqref{eq:Lap_tradeoff_1dim_}.
Unlike the approach of \citeauthor{dong2022gaussianDP}, our derivation directly use the definition of trade-off function, and our approach can be applied to more higher dimensional case.
\end{remark}

\subsubsection{Bivariate Laplace mechanism}\label{subsubsec:2dim_Lap}

Let $p=2$. 
A primary importance of this bivariate case is that the trade-off function of mechanisms perturbing contingency tables (of any dimension) can be handled using the $p = 2$ results, as argued in Section \ref{sec:lap_mech_table}.

Following essentially the same steps used in obtaining \eqref{eq:Lap_tradeoff_1dim_}, but with messier computations, we derive the following trade-off function between bivariate Laplace distributions.
\begin{lem}\label{lem:Lap_tradeoff_p=2}
    For $\deltav = (\delta_1, \delta_2)$ with $\delta_1 \ge \delta_2 \ge 0$,
    \begin{equation}\label{eq:Laplace_tradeoff_p=2}
        T_{\Lap_2, \deltav}(\alpha) =
        \begin{cases}
          -e^{\delta_1+\delta_2}\alpha + 1
          & \text{if }~ 0 \leq \alpha < \frac{1}{4}e^{-\delta_1-\delta_2},\\
          \frac{1}{4}e^{-C_{\deltav}(\alpha)} \left( 3 + C_{\deltav}(\alpha) \right)
          & \text{if }~ \frac{1}{4}e^{-\delta_1-\delta_2} \leq \alpha < \frac{1}{4}e^{-\delta_1}(1+\delta_2),\\
          -e^{\delta_1-\delta_2}\alpha + \frac{1}{2}e^{-\delta_2}(2+\delta_2)
          & \text{if }~ \frac{1}{4}e^{-\delta_1}(1+\delta_2) \leq \alpha < \frac{1}{4}e^{-\delta_1}(2+\delta_2),\\
          \frac{1}{16\alpha}e^{-\delta_1 - \delta_2}(2+\delta_2)^2
          & \text{if }~ \frac{1}{4}e^{-\delta_1}(2+\delta_2) \leq \alpha \le \frac{1}{4}e^{\frac{-\delta_1 - \delta_2}{2}}(2+\delta_2),\\
        \end{cases}
    \end{equation}
    and $T_{\Lap_2, \deltav}(\alpha)$ for $\alpha > \frac{1}{4}e^{\frac{-\delta_1 - \delta_2}{2}}(2+\delta_2)$ is determined by the symmetry.
    Here, $C_{\deltav}(\alpha) \coloneqq \linebreak W(4 \alpha e^{1 + \delta_1 + \delta_2})  - 1$, and $W(\cdot)$ is the Lambert $W$ function.
\end{lem}

A proof of Lemma~\ref{lem:Lap_tradeoff_p=2} is given in Section~\ref{sec:misc}.

For a given $\deltav = (\delta_1, \delta_2)$ where $\delta_1 \ge \delta_2 > 0$, consider a reparameterization $\lambda_1 = \delta_1 + \delta_2 = \|\deltav\|_1$ and $\lambda_2 = \delta_1 - \delta_2$, so that 
$$ 
\deltav^\top = (\delta_1, \delta_2) = \left( \frac{\lambda_1 + \lambda_2}{2}, \frac{\lambda_1 - \lambda_2}{2}\right).
$$

Write $\widetilde{T}_{\lambdav} \coloneqq T_{\Lap_2, (\frac{\lambda_1 + \lambda_2}{2}, \frac{-\lambda_1 + \lambda_2}{2})^\top}$ for the trade-off function   $T_{\Lap_2, \deltav}$, parameterized with respect to  $\lambdav = (\lambda_1, \lambda_2)^\top.$ The following observation is useful:  

\begin{lem}\label{lemma:Lap_tradeoff_monotonicity}
Given $\lambda_1 > 0$ fixed, for any $\alpha \in (0,1)$, $\widetilde{T}_{\lambdav }(\alpha)$ is a monotone decreasing function of $\lambda_2$.
\end{lem}

A proof of Lemma~\ref{lemma:Lap_tradeoff_monotonicity} is given in Section~\ref{sec:misc}. 

We are interested in the evaluation of $\Lc_2(b) = \inf_{S\sim S'} T(\Lap_2(\0v, 1), \Lap_2( \deltav, 1))$, where $\deltav = (\thetav(S) - \thetav(S')) / b$.
Then 
$$
\Lc_2(b) = \inf_{\deltav \in \Sc_{\thetav}/b} T(\Lap_2(\0v, 1), \Lap_2( \deltav, 1)). 
$$
Without loss of generality, we can assume that any $\deltav \in \Sc_{\thetav}/b$ satisfies $\delta_1 \ge \delta_2 \ge 0$. To see this, let 
$$|\Sc_{\thetav}| := \left\{ 
\begin{pmatrix}
    \max(|\theta_1(S) - \theta_1(S')|, |\theta_2(S) - \theta_2(S')|)  \\
    \min(|\theta_1(S) - \theta_1(S')|, |\theta_2(S) - \theta_2(S')|)
\end{pmatrix} : S \sim S', S, S' \in \Xc^n
\right\},
$$
and observe that for any $\deltav \in \Sc_{\thetav}/b$, there exists $\deltav_0 \in |\Sc_{\thetav}| / b$ such that 
 $T(\Lap_2(\0v, 1), \Lap_2( \deltav, 1)) = T(\Lap_2(\0v, 1), \Lap_2( \deltav_0, 1))$. 
Define $$\Sc_\lambda = \{(\delta_1 + \delta_2, \delta_1 - \delta_2): (\delta_1, \delta_2) \in |\Sc_{\thetav}| / b\},$$ which may be called the sensitivity set of $\thetav$ with respect to the ``$\lambdav$-parameterization.'' 
If $\Delta_1 = \Delta_1(\thetav; \Xc^n)$, then for any $(\lambda_1, \lambda_2) \in \Sc_\lambda$, $\lambda_1 \in [0,\Delta_1/b]$. 

For each $\lambda_1 \ge 0$, if 
$(\lambda_1, \lambda_2) \in \Sc_\lambda$ for some 
$0 \le \lambda_2 \le \lambda_1$, define
$\nabla(\lambda_1) = \sup
\{\lambda_2: (\lambda_1,\lambda_2) \in \Sc_\lambda\}$.
Then 
\begin{align*}
    \Lc_2(b) & = \inf_{\deltav \in \Sc_{\thetav}/b} T(\Lap_2(\0v, 1), \Lap_2( \deltav, 1) \\
      & =  \inf_{\lambdav \in \Sc_\lambda} \widetilde{T}_{\lambdav } \\
      & = \inf_{\lambda_1 \in [0, \Delta_1/b]} \inf_{\lambda_2 \in [0, \lambda_1]} \widetilde{T}_{\lambdav } \\
      & = \inf_{\lambda_1 \in [0, \Delta_1/b]}  \widetilde{T}_{ (\lambda_1, \nabla(\lambda_1))}.  \quad \mbox{(by Lemma~\ref{lemma:Lap_tradeoff_monotonicity})}
\end{align*}

Summarizing the above discussion, we conclude that evaluating $\Lc_2(b)$ requires an exact specification of $\Sc_{\thetav}$. 

\begin{lem}\label{lem:bivariate_laplace_tradeoff_worst} Let $p = 2$, and $\thetav$ and $S_{\thetav}$ be given and $\Delta_1 = \Delta_1(\thetav; \Xc^n)$. Then the value of the trade-off function $\Lc_2(b):[0,1] \to [0,1]$ at $\alpha \in [0,1]$ is 
$\inf_{\lambda_1 \in [0, \Delta_1/b]}  \widetilde{T}_{ (\lambda_1, \nabla(\lambda_1))} (\alpha)$. 
\end{lem}

Note that since $\nabla(\lambda_1) \le \lambda_1$ for any $\Sc_\lambda$, we have 
\begin{align*}
    \Lc_2(b) & = \inf_{\lambda_1 \in [0, \Delta_1/b]}   \widetilde{T}_{  (\lambda_1, \nabla(\lambda_1)) } \\
     & \ge \inf_{\lambda_1 \in [0, \Delta_1/b]} \widetilde{T}_{ (\lambda_1, \lambda_1) } \\ 
     & = \inf_{\delta_1 \le \Delta_1/b} T(\Lap_2(\0v, 1), \Lap_2( (\delta_1, 0), 1) \\
     & = \inf_{\delta_1 \le \Delta_1/b} T(\Lap_1(0, 1), \Lap_1( \delta_1, 1) \\
     & =  T(\Lap_1(0, 1), \Lap_1(\Delta_1/b, 1).
\end{align*}
In the second equality above, we used the fact that $\lambda_2 = \lambda_1$ if and only if $\delta_1 = \lambda_1$ and $\delta_2 = 0$. This gives proof for Lemma~\ref{lem-high2one}, but only for the case $p = 2$.

\subsubsection{Proofs of Lemmas~\ref{lem:Lap_tradeoff_p=2} and \ref{lemma:Lap_tradeoff_monotonicity}}\label{sec:misc}

\begin{proof}[Proof of Lemma~\ref{lem:Lap_tradeoff_p=2}]

We consider the case where $\delta_1 > \delta_2 \geq 0$. In the case of $\delta_1 = \delta_2$, it is verified in the same way, and it also can be confirmed that results are compatible with the case of $\delta_1 > \delta_2$. The log-likelihood for $p=2$ is
\begin{align*}
    l(\xv) &= |x_1| - |x_1 - \delta_1| + |x_2| - |x_2 - \delta_2|\\
    &= \begin{cases}
      \delta_1 + \delta_2, & \text{if $x_1 \geq \delta_1, \; x_2 \geq \delta_2$}\\
      \delta_1 - \delta_2 + 2x_2, & \text{if $x_1 \geq \delta_1, \; 0 < x_2 < \delta_2$}\\
      \delta_1 - \delta_2, & \text{if $x_1 \ge \delta_1, \; x_2 \leq 0$}\\
      -\delta_1 + \delta_2 + 2x_1, & \text{if $0 < x_1 < \delta_1, \; x_2 \geq \delta_2$}\\
      -\delta_1 - \delta_2 + 2x_1 + 2x_2, & \text{if $0 < x_1 < \delta_1, \; 0 < x_2 < \delta_2$}\\
      -\delta_1 - \delta_2 + 2x_1, & \text{if $0 < x_1 < \delta_1, \; x_2 \leq 0$}\\
      -\delta_1 + \delta_2, & \text{if $x_1 \leq 0, \; x_2 \geq \delta_2$}\\
      -\delta_1 - \delta_2 + 2x_2, & \text{if $x_1 \leq 0, \; 0 < x_2 < \delta_2$}\\
      -\delta_1 - \delta_2, & \text{if $x_1 \leq 0, \; x_2 \leq 0$}\\
    \end{cases},
\end{align*}
and examples of its values by $\xv = (x_1, x_2)^\top$ are denoted in Figure \ref{fig:3d_tmp} with contour lines.

\begin{figure}[]
\centering
\begin{subfigure}{0.4\textwidth}
  \centering
  \includegraphics[width=\linewidth]{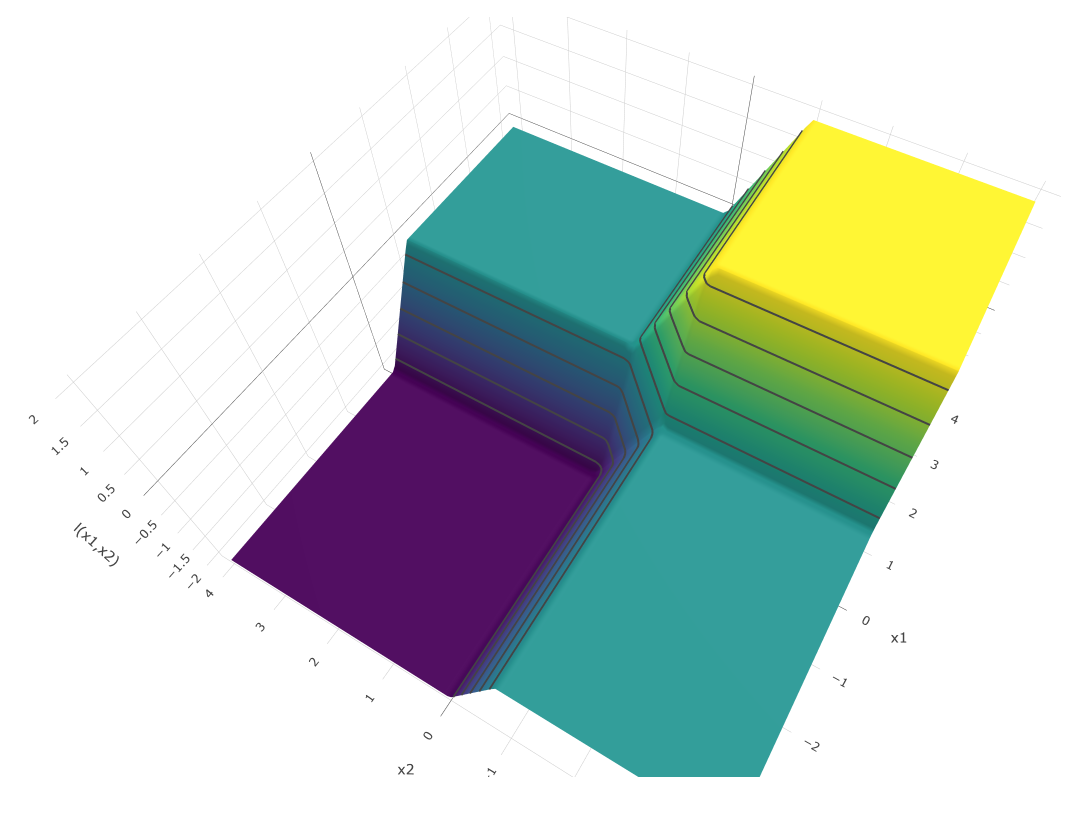}
  \caption{$\delta_1 = -\delta_2 = 1$}
  \label{fig:3d_tmp:samedelta}
\end{subfigure}
\begin{subfigure}{0.4\textwidth}
  \centering
  \includegraphics[width=\linewidth]{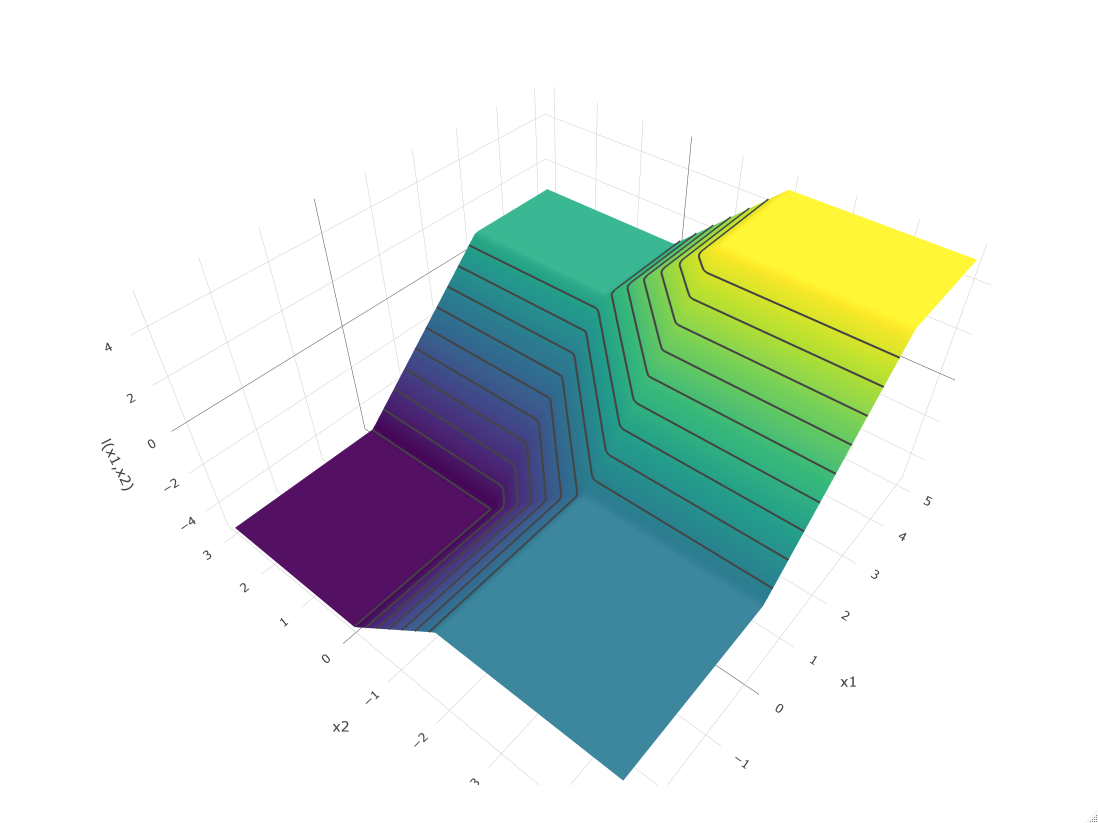}
  \caption{$\delta_1 = 3, \delta_2 = -2$}
  \label{fig:3d_tmp:diffdelta}
\end{subfigure}
\caption{Surface plots for $l(\xv).$}
\label{fig:3d_tmp}
\end{figure}

Let $\alpha(K, \gamma) := \P_{\Xv \sim \mathrm{Lap}_2(\0v, 1)}(l(\Xv) > K  ) + \gamma \cdot \P_{\Xv \sim \mathrm{Lap}_2(\0v, 1)}(l(\Xv) = K )$. 
First, we obtain $\alpha(K, \gamma)$ by the range of $K$.
Next, we calculate the corresponding trade-off function value at $\alpha(K, \gamma)$ by
$
    T_{\Lap_2, \deltav}(\alpha(K,\gamma))
    = \P_{\Xv \sim \mathrm{Lap}_2(\deltav, 1)}(l(\Xv) < K ) + (1-\gamma) \cdot \P_{\Xv \sim \mathrm{Lap}_2(\deltav, 1)}(l(\Xv) = K ).
$

\begin{enumerate}[(i)]
    \item $K > \delta_1 + \delta_2$ :
    $\alpha(K, \gamma) = 0$, and $ T_{\Lap_2, \deltav}(0) = 1. $

    \item $K = \delta_1 + \delta_2$ :
    For any $\gamma \in [0, 1]$, 
    $$
    \alpha(K, \gamma) = \gamma \cdot \P_{\Xv \sim \mathrm{Lap}_2(\0v, 1)}(X_1 \geq \delta_1, \; X_2 \geq \delta_2   ) = \frac{1}{4}e^{-\delta_1-\delta_2} \gamma. 
    $$
    Then, $T_{\Lap_2, \deltav}(\alpha(K, \gamma)) = \frac{3}{4} + \frac{1}{4}(1-\gamma).$
    Conversely, for $0 \leq \alpha \leq \frac{1}{4}e^{-\delta_1-\delta_2}$,
    by choosing $\gamma_\alpha = 4e^{\delta_1+\delta_2}$, we have 
    \[
    T_{\Lap_2, \deltav}(\alpha) 
    = \frac{3}{4} + \frac{1}{4}(1 - \gamma_{\alpha})
    = - e^{\delta_1 + \delta_2}\alpha + 1.
    \]

    \item $\delta_1 - \delta_2 < K < \delta_1 + \delta_2$ :
    \begin{align*}
        \alpha(K, \gamma)
        &= \P_{\Xv \sim \mathrm{Lap}_2(\0v, 1)} \left( X_1 + X_2 \geq \frac{\delta_1 + \delta_2 + K}{2}, \; X_1 \geq \frac{\delta_1 - \delta_2 + K}{2}, \; X_2 \geq \frac{-\delta_1 + \delta_2 + K}{2}   \right)\\
        &= \frac{1}{4}e^{-\frac{\delta_1 + \delta_2 + K}{2}} \left( 1 + \frac{\delta_1 + \delta_2 - K}{2} \right),
    \end{align*}
    and
    $$T_{\Lap_2, \deltav}(\alpha(K,\gamma)) = \frac{1}{4}e^{-\frac{\delta_1 + \delta_2 - K}{2}} \left( 3 + \frac{\delta_1 + \delta_2 - K}{2} \right).$$
    Conversely, for $\frac{1}{4}e^{-\delta_1 - \delta_2} < \alpha < \frac{1}{4}e^{-\delta_1}(1+\delta_2)$,
    $$
    T_{\Lap_2, \deltav}(\alpha) = \frac{1}{4}e^{-\frac{\delta_1 + \delta_2 - K_{\alpha}}{2}} \left( 3 + \frac{\delta_1 + \delta_2 - K_{\alpha}}{2} \right),
    $$
    where $K_{\alpha}$ is a solution of $\alpha = \frac{1}{4}e^{-\frac{\delta_1 + \delta_2 + K}{2}} \left( 1 + \frac{\delta_1 + \delta_2 - K}{2} \right).$
    By using Lambert $W$ function, one can check that 
    $K_\alpha = \delta_1 + \delta_2 - 2W(4\alpha^{1+\delta_1+\delta_2})$.
    Then, 
    $
    T_{\Lap_2, \deltav}(\alpha) = \frac{1}{4}e^{-C_{\deltav}(\alpha)} \left( 3 + C_{\deltav}(\alpha) \right),
    $
    where $C_{\deltav}(\alpha) = W(4\alpha e^{1+\delta_1 + \delta_2}) - 1 $.

    \item $K = \delta_1 - \delta_2$ :
    $
        \alpha(K, \gamma) = \frac{1}{4}e^{-\delta_1}(1+\delta_2) + \frac{1}{4}e^{-\delta_1} \gamma,
    $
    and 
    $$
    T_{\Lap_2, \deltav}(\alpha(K, \gamma)) = \frac{1}{4}e^{-\delta_2}(2+\delta_2) + \frac{1}{4}e^{-\delta_2} (1-\gamma).
    $$
    By using the same argument used in (ii), we have  
    $$
    T_{\Lap_2, \deltav}(\alpha) 
    = -e^{\delta_1 - \delta_2}\alpha + \frac{1}{2}e^{-\delta_2}(2 + \delta_2)
    $$
    for  $\frac{1}{4}e^{-\delta_1}(1+\delta_2) \leq \alpha \leq \frac{1}{4}e^{-\delta_1}(2+\delta_2)$.

    \item $-\delta_1 + \delta_2 < K < \delta_1 - \delta_2$ :
    $
        \alpha(K, \gamma) = \frac{1}{4}e^{-\frac{\delta_1 + \delta_2 + K}{2}} \left( 2 + \delta_2 \right),
    $
    and
    $$
    T_{\Lap_2, \deltav}(\alpha(K, \gamma)) = \frac{1}{4}e^{-\frac{\delta_1+\delta_2-K}{2}}(2+\delta_2).
    $$
    Conversely, for $\frac{1}{4}e^{-\delta_1}(2+\delta_2) < \alpha < \frac{1}{4}e^{-\delta_2}(2+\delta_2)$, it holds that 
    $$
    T_{\Lap_2, \deltav}(\alpha) = \frac{1}{16\alpha}e^{-\delta_1-\delta_2} (2+\delta_2)^2.
    $$

\end{enumerate}
Note that $T_{\Lap_2, \deltav}$ is symmetric on $y=x$.
Since 
$$
\frac{1}{4}e^{-(\delta_1+\delta_2)/2}(2+\delta_2) \in \left[\frac{1}{4}e^{-\delta_1}(2+\delta_2), \frac{1}{4}e^{-\delta_2}(2+\delta_2)\right],
$$
(v) leads to 
$$ T_{\Lap_2, \deltav}(\frac{1}{4}e^{-(\delta_1+\delta_2)/2}(2+\delta_2))
= \frac{1}{4}e^{-(\delta_1+\delta_2)/2}(2+\delta_2).$$
Thus, $T_{\Lap_2, \deltav}(\alpha)$ for $\alpha >\frac{1}{4}e^{-(\delta_1+\delta_2)/2}(2+\delta_2) $ is determined by symmetry and (i)-(v).

\end{proof}

\begin{proof}[Proof of Lemma \ref{lemma:Lap_tradeoff_monotonicity}]
Since the trade-off function is symmetric in the sense that it coincides with its inverse function, we need only to consider the left part where $0 \leq \alpha < \frac{1}{4}e^{-\frac{\delta_1 + \delta_2}{2}}(2+\delta_2)$.

Let $\lambda_1 \in (0, \Delta_1)$ be fixed. Rephrasing $T_{\Lap_2, \deltav}(\alpha)$ by $\lambda_1$ and $\lambda_2$, we get
\begin{align*}
    &T(\lambda_2)(\alpha) \coloneqq \widetilde{T}_{\lambdav}(\alpha) =\\
    &\begin{cases}
      -e^{\lambda_1}\alpha + 1, & \text{if $\alpha \in \left[ 0, \frac{1}{4}e^{-\lambda_1} \right) \eqqcolon I_1(\lambda_2) \equiv I_1$}\\
      \frac{1}{4}e^{-C_{\lambda_1}(\alpha)} \left( 3 + C_{\lambda_1}(\alpha) \right), & \text{if $\alpha \in \left[ \frac{1}{4}e^{-\lambda_1}, P_2(\lambda_2) \right) \eqqcolon I_2(\lambda_2)$}\\
      -e^{\lambda_2}\alpha + \frac{1}{2}e^{\frac{\lambda_2 - \lambda_1}{2}}(2-\frac{\lambda_2 - \lambda_1}{2}), & \text{if $\alpha \in \left[ P_2(\lambda_2), P_3(\lambda_2) \right) \eqqcolon I_3(\lambda_2)$}\\
      \frac{1}{16\alpha}e^{-\lambda_1}(2-\frac{\lambda_2 - \lambda_1}{2})^2, & \text{if $\alpha \in \left[ P_3(\lambda_2), P_4(\lambda_2) \right] \eqqcolon I_4(\lambda_2)$}\\
      \text{(symmetric part omitted),}
    \end{cases}
\end{align*}
where
\begin{align*}
    &C_{\lambda_1}(\alpha) := C_{\deltav}(\alpha) = W(4 \alpha e^{1 + \lambda_1}) - 1 \\
    &P_2(\lambda_2) \coloneqq \frac{1}{4}e^{-\frac{\lambda_1 + \lambda_2}{2}}(1-\frac{\lambda_2 - \lambda_1}{2})\\
    &P_3(\lambda_2) \coloneqq \frac{1}{4}e^{-\frac{\lambda_1 + \lambda_2}{2}}(2-\frac{\lambda_2 - \lambda_1}{2})\\
    &P_4(\lambda_2) \coloneqq \frac{1}{4}e^{-\frac{\lambda_1}{2}}(2-\frac{\lambda_2 - \lambda_1}{2}).
\end{align*}

Now let $\lambda_2^{(1)} < \lambda_2^{(2)}$. Then, as $P_i, \; i = 2, 3, 4$ are decreasing with respect to $\lambda_2$, in particular we have that either $P_3(\lambda_2^{(2)}) \in I_3(\lambda_2^{(1)})$ or $P_3(\lambda_2^{(2)}) \in I_2(\lambda_2^{(1)})$. In both cases, we can see that
\begin{equation}\label{eq:2dim_tradeoff_I1_eq}
    T(\lambda_2^{(1)})(\alpha) = T(\lambda_2^{(2)})(\alpha) \mbox{ for } \alpha \in [0, P_2(\lambda_2^{(2)}))
\end{equation}
and
\begin{equation}\label{eq:2dim_tradeoff_I4_ineq}
    T(\lambda_2^{(1)})(\alpha) > T(\lambda_2^{(2)})(\alpha) \mbox{ for } \alpha \in I_4(\lambda_2^{(1)}) \cap I_4(\lambda_2^{(2)}) = [P_3(\lambda_2^{(1)}), P_4(\lambda_2^{(2)})].
\end{equation}
Thus, we only need to check inequality for $\alpha \in [P_2(\lambda_2^{(2)}), P_3(\lambda_2^{(1)})]$.

First we regard the case where $P_3(\lambda_2^{(2)}) \in I_3(\lambda_2^{(1)})$. The trade-off functions are convex, and $T(\lambda_2^{(2)})(\alpha)$ is linear in $\alpha$ on $I_3(\lambda_2^{(2)})$. Due to \eqref{eq:2dim_tradeoff_I1_eq}, we obtain
$$
    T(\lambda_2^{(1)})(\alpha) \ge T(\lambda_2^{(2)})(\alpha) \mbox{ for } \alpha \in [P_2(\lambda_2^{(2)}), P_3(\lambda_2^{(2)})].
$$
Therefore, we have $T(\lambda_2^{(1)})(P_3(\lambda_2^{(2)})) > T(\lambda_2^{(2)})(P_3(\lambda_2^{(2)}))$ and $T(\lambda_2^{(1)})(P_3(\lambda_2^{(1)})) > T(\lambda_2^{(2)})(P_3(\lambda_2^{(1)}))$. As $T(\lambda_2^{(1)})(\alpha)$ is linear on $[P_3(\lambda_2^{(2)}), P_3(\lambda_2^{(1)})]$, we get that
$$
    T(\lambda_2^{(1)})(\alpha) \ge T(\lambda_2^{(2)})(\alpha) \mbox{ for } \alpha \in [P_3(\lambda_2^{(2)}), P_3(\lambda_2^{(1)})].
$$

For the case of $P_3(\lambda_2^{(2)}) \in I_2(\lambda_2^{(1)})$, let us show that there exists a finite sequence $\{\lambda_{2,i}\}_{i=0}^n, \; n \ge 2$ of
$$
\lambda_2^{(1)} = \lambda_{2,0} \le \lambda_{2,1} \le \cdots \le \lambda_{2,n-1} \le \lambda_{2,n} = \lambda_2^{(2)}
$$
such that
$$
P_3(\lambda_{2,i}) \in I_3(\lambda_{2,i-1}), \; i = 1, \cdots, n.
$$
Then we have $T(\lambda_{2,i-1})(\alpha) > T(\lambda_{2,i})(\alpha), \; i = 1, \cdots, n$ for $\alpha \in [0,1]$, so that $T(\lambda_2^{(1)})(\alpha) = T(\lambda_{2,0})(\alpha) > T(\lambda_{2,n})(\alpha) = T(\lambda_2^{(2)})(\alpha), \alpha \in [0,1]$.

Consider $\lambda_2' > \lambda_2^{(1)}$. We seek the range of $\lambda_2'$ to satisfy $P_3(\lambda_2') \in I_3(\lambda_2^{(1)})$. Let $c \coloneqq \frac{\lambda_1 - \lambda_2^{(1)}}{2} > 0$ and $\lambda_2' = \lambda_2^{(1)} + 2\epsilon, \; \epsilon \in (0,c)$. Then, we obtain
\begin{align}\label{eq:I3_bound}
\begin{split}
    &P_3(\lambda_2') \in I_3(\lambda_2^{(1)})\\
    &\Leftrightarrow \frac{1}{4} e^{-\frac{\lambda_2' + \lambda_1}{2}} (2 + \frac{\lambda_1 - \lambda_2'}{2}) \ge \frac{1}{4} e^{-\frac{\lambda_2^{(1)} + \lambda_1}{2}} (1 + \frac{\lambda_1 - \lambda_2^{(1)}}{2})\\
    &\Leftrightarrow e^{-\epsilon} (2 + \frac{\lambda_1 - \lambda_2^{(1)}}{2} - \epsilon) \ge 1 + \frac{\lambda_1 - \lambda_2^{(1)}}{2}\\
    &\Leftrightarrow 2 + c - \epsilon \ge (1+c)e^{\epsilon}\\
    &\Leftrightarrow \epsilon \le 2 + c - W((1+c)e^{2+c})
\end{split}
\end{align}
Let $h(x) := 2 + x - W((1+x)e^{2+x})$.
As $h(c) > 0$, any $\lambda_2' \in (\lambda_2^{(1)}, \lambda_2^{(1)} + 2h(c))$ satisfies $P_3(\lambda_2') \in I_3(\lambda_2^{(1)})$. Let $\lambda_{2,1} \coloneqq \lambda_2^{(1)} + h(c)$.

We also have that
$$
\frac{\partial h(x)}{\partial x} = \frac{1 - e^{h(x)}}{1 + (1+x)e^{h(x)}} < 0
$$
for all $x > 0$.
This implies, as $\frac{\lambda_1 - \lambda_{2,1}}{2} < \frac{\lambda_1 - \lambda_2^{(1)}}{2} = c$, that either $\lambda_{2,1} + h(c) \ge \lambda_2^{(2)}$ or $P_3(\lambda_{2,1} + h(c)) \in I_3(\lambda_{2,1})$ by applying the logic of \eqref{eq:I3_bound} again. In the former case, $n=2$ and $P_3(\lambda_2^{(2)}) \in I_3(\lambda_{2,1})$ so the proof is done. In the latter case, let $\lambda_{2,2} = \lambda_{2,1} + h(c)$ and construct $\{\lambda_{2,i}\}_{i=0}^n$ inductively as $\lambda_{2, i+1} = \lambda_{2,i} + h(c)$ until $\lambda_{2,n-1} + h(c) \ge \lambda_2^{(2)}$.
\end{proof}

\subsection{Proofs  for Section~\ref{sec:Lap_mech}}\label{sec:appendix_proofs2}

\subsubsection{Proof of Lemma~\ref{lem:eps-DP_to_mu-GDP_conversion}}

\begin{proof}[Proof of Lemma~\ref{lem:eps-DP_to_mu-GDP_conversion}]
(i)
Suppose that a randomized mechanism $\Mv$ is $\varepsilon$-DP for some $\varepsilon > 0$. 
Then, we have $\inf_{S \sim S'} T(\Mv(S), \Mv(S'))(\alpha) \ge T_{\varepsilon \mbox{-} \mathrm{DP}}$.
Let $\mu_{\varepsilon} = 2\Phi^{-1}\{e^{\varepsilon} / ( e^{\varepsilon} + 1) \}$.
To show $\Mv$ is $\mu$-GDP for any $\mu \ge \mu_{\varepsilon}$ it is enough to show that $ T_{\varepsilon \mbox{-} \mathrm{DP}} \ge G_{\mu_{\varepsilon}} $ since $G_{\mu}$ is decreasing in $\mu$.
Observe that 
\begin{align*}
    G_{\mu_{\varepsilon}}\left(\frac{1}{1+e^{\varepsilon}}\right)
    &= \Phi\left(\Phi^{-1}\left(\frac{e^{\varepsilon}}{1+e^{\varepsilon}}\right) - 2\Phi^{-1}\left(\frac{e^{\varepsilon}}{1+e^{\varepsilon}} \right) \right) \\
    &= \Phi\left(- \Phi^{-1}\left(\frac{e^{\varepsilon}}{1+e^{\varepsilon}} \right) \right) \\
    &= 1 - \Phi\left(\Phi^{-1}\left(\frac{e^{\varepsilon}}{1+e^{\varepsilon}} \right) \right) \\ 
    &= \frac{1}{1+e^{\varepsilon}}.
\end{align*}
From the convexity of any trade-off function, for any $\alpha \in [0, (1+e^{\varepsilon})^{-1}]$, we have 
$
    G_{\mu_{\varepsilon}} (\alpha) 
    \le l(\alpha),
$
where $l(\alpha)$ is function of line segment passes through $(0, 1)$ to $( (1+e^{\varepsilon})^{-1},(1+e^{\varepsilon})^{-1} ) $.
In fact, simple calculation shows that $l(\alpha) = T_{\varepsilon \mbox{-} \mathrm{DP}}(\alpha)$ for $\alpha \in [0, (1+e^{\varepsilon})^{-1}]$.
Finally, trade-off functions are symmetric on $y=x$, we have, for any $\alpha \in [0, 1]$
$$
    G_{\mu_{\varepsilon}} (\alpha) 
    \le T_{\varepsilon \mbox{-} \mathrm{DP}}(\alpha).
$$

(ii) 
Suppose that $\Mv$ is $\varepsilon$-DP where $\varepsilon \le \varepsilon_{\mu}$.
Note that $e^{\varepsilon} / (e^{\varepsilon}+1)$ is increasing in $\varepsilon > 0$, and 
\begin{align*}
2\Phi^{-1} \left( \frac{e^{\varepsilon}}{1+e^{\varepsilon}} \right) 
&\le 
2\Phi^{-1} \left( \frac{e^{\varepsilon_{\mu}}}{1+e^{\varepsilon_{\mu}}} \right) \\
&= 2\Phi^{-1} \left( \frac{\Phi(\mu/2)/\Phi(-\mu/2)}{1 + \Phi(\mu/2)/\Phi(-\mu/2)} \right)  \\
&= 2\Phi^{-1} \left( \Phi(\mu/2) \right)  \\
&= \mu.
\end{align*}
From the part (i) of the Lemma, $\Mv$ is $\mu$-GDP.
\end{proof}

\subsubsection{Proof of Lemma~\ref{lem:univariate_laplace_tradeoff}}
\begin{proof}[Proof of Lemma~\ref{lem:univariate_laplace_tradeoff}]
 
From the previous remark, for any $t > 0$ we have
$T_{\Lap_1, t}(\alpha) = F(F^{-1}(1-\alpha)-t)$, where $F$ is the distribution of $\Lap(0, 1)$.
For a fixed $\alpha \in (0, 1)$, $F(F^{-1}(1-\alpha)-t)$ is decreasing in $t$.  
Thus, for any fixed $\alpha \in (0, 1)$, if $t \ge t' > 0$ 
\[
T_{\Lap_1, t}(\alpha) \le  T_{\Lap_1, t'}(\alpha).
\]
Since the trade-off function $\Lc_1(b)$ is defined over the univariate Laplace mechanism, we have
\begin{align*}
\Lc_1(b)
&= \inf_{S \sim S'} T(M_{\Lap}(S; b), M_{\Lap}(S'; b)) \\
&= \inf_{S \sim S'} T(\Lap(0, 1), \Lap(|\theta(S) - \theta(S')|/b, 1)) \\
&= T(\Lap(0, 1), \Lap(\sup_{S \sim S'}|\theta(S) - \theta(S')|/b, 1)) \\
&= T(\Lap(0, 1), \Lap(\Delta_1/b, 1)) \\
&= T(\Lap(0, 1), \Lap(\delta, 1)).
\end{align*}
The proof is completed by \eqref{eq:Lap_tradeoff_1dim_}.

\end{proof}

 \subsubsection{Proof of Proposition~\ref{prop:1dim_lap_stronger} }
 
Proof of Proposition~\ref{prop:1dim_lap_stronger} relies on the following lemma, which answers the question: For which values  of $\delta$, $T_{\Lap_1,\delta}$ is greater than $G_\mu$? 

\begin{lem} \label{lem:1dim-tight-lap-tf}
For a $\mu > 0$, let $\delta = -2\log\{2\Phi(-\mu/2)\}$.
Then 
\[
T_{\Lap_1, \delta}(\alpha)
> G_{\mu}(\alpha)
\]
for any $\alpha \in (0, 1)$, except for $\alpha_0 = \Phi(-\mu/2)$, where $T_{\Lap_1, \delta}(\alpha_0) = G_{\mu}(\alpha_0). $
Moreover, For any $\delta' > \delta$, $T_{\Lap_1, \delta'}(\alpha_0) < G_{\mu}(\alpha_0)$.
\end{lem}

\begin{proof}
From \eqref{eq:Lap_tradeoff_1dim_}, we have
\[
T_{\Lap_1, \delta}(\alpha)
=    
\begin{cases}
    1 - \frac{1}{4\Phi^2(-\mu/2)}  \alpha 
    & \text{ if } \alpha < 2\Phi^2(-\mu/2) ,\\
    \Phi^2(-\mu/2) / \alpha 
    & \text{ if } 2\Phi^2(-\mu/2) \le \alpha < \frac{1}{2} ,\\
    4\Phi^2(-\mu/2)(1-\alpha) & \text{ otherwise.}
\end{cases}
\]
(Here, $\Phi^2(t) \coloneqq \{\Phi(t)\}^2$.)
Since $\alpha_0 \in [2\Phi^2(-\mu/2), 1/2)$,
it can be easily checked that
\[
T_{\Lap_1, \delta}(\alpha_0)
= \alpha_0 = G_{\mu}(\alpha_0).
\]
Next, we will show that $T_{\Lap_1, \delta}(\alpha) > G_{\mu}(\alpha)$ for any other choice of $\alpha$.

First, consider $\alpha \in (\alpha_0, 1/2]$.
The derivatives become 
$$ 
\frac{\partial T_{\Lap_1, \delta}}{\partial \alpha}(\alpha) 
= - \frac{\Phi^2(-\mu/2)}{\alpha^2},
$$
and 
$$ 
\frac{\partial G_{\mu}(\alpha)}{\partial \alpha}(\alpha) 
= - e^{\mu\Phi^{-1}(1-\alpha) - \tfrac{1}{2}\mu^2}.
$$
From the definition of $\alpha_0$, it holds that 
$$
\frac{\partial T_{\Lap_1, \delta}}{\partial \alpha}(\alpha_0)
=\frac{\partial G_{\mu}(\alpha)}{\partial \alpha}(\alpha_0). 
$$
The difference between two derivatives is
$$
\left(
\frac{\partial T_{\Lap_1, \delta}}{\partial \alpha}
-\frac{\partial G_{\mu}(\alpha)}{\partial \alpha}
\right)(\alpha) 
= - \frac{\Phi^2(-\mu/2)}{\alpha^2} + 
e^{\mu\Phi^{-1}(1-\alpha) - \tfrac{1}{2}\mu^2}.
$$
Observe that
$$
\left(
\frac{\partial T_{\Lap_1, \delta}}{\partial \alpha}
-\frac{\partial G_{\mu}(\alpha)}{\partial \alpha}
\right)(\alpha)
> 0 
$$
if and only if 
\begin{equation} \label{eq:pf-prop11-g}
g(\alpha) \coloneqq \log(\alpha) - \frac{1}{2}\mu\Phi^{-1}(\alpha) - \frac{1}{4}\mu^2  -\log{\Phi(-\mu/2)}
> 0.
\end{equation}
After some calculation, one can see that $g(\alpha_0) = 0$.
The derivative of $g$ becomes
\[
g'(\alpha) = \frac{2}{\alpha} - \frac{\mu}{\phi(\Phi^{-1}(\alpha))}.
\]
Note that 
$
g'(\alpha_0) 
= \frac{2}{\alpha_0\phi(\mu/2)}\left(\phi(- \tfrac{\mu}{2}) + (-\tfrac{\mu}{2})\Phi(-\tfrac{\mu}{2})\right) > 0
$.
The last inequality holds from the well-known fact that $\phi(x) > x(1-\Phi(x))$ for any $x > 0$.
For $\alpha \in (\alpha_0, 1/2)$,  reparameterize $\alpha$ by $\eta = -2\Phi^{-1}(\alpha)$ where $\eta \in (0, \mu)$, then $g'(\alpha)$ is reparameterized by $h(\eta)$ which is defined as 
\[
h(\eta) \coloneqq \frac{2}{\Phi(-\eta/2)} - \frac{\mu}{\phi(-\eta/2)}.
\]
Note that $h(\eta) > 0$ if and only if $\bar h(\eta) \coloneqq 2\phi(- \tfrac{\eta}{2}) - \mu \Phi(-\tfrac{\eta}{2}) > 0$.
Observe that the derivative of $\bar h$ is positive:
for any $0 < \eta < \mu$,
\[
\frac{\partial \bar h}{\partial \eta} (\eta)
= \frac{1}{2} \phi(\frac{\eta}{2})(\mu - \eta) > 0.
\]
So, $\bar h$ is an increasing function on $(0, \mu)$.
Also, $\bar h(\mu) = 2\phi(-\mu/2) - \mu\Phi(-\mu/2) > 0$, and we have $h(\mu) > 0$.

If $\bar h(0) > 0$ then $\bar h(\eta) > 0$ for any $\eta \in (0, \mu)$, and  $h(\eta) > 0$. 
As a consequence, $g'(\alpha) > 0$ on $[\alpha_0, 1/2)$.
Since $g(\alpha_0) = 0$, we can conclude that 
$g(\alpha) > g(\alpha_0) = 0$ for $\alpha \in (\alpha_0, 1/2)$.
By \eqref{eq:pf-prop11-g}, we can conclude that 
$\left(
\frac{\partial T_{\Lap_1, \delta}}{\partial \alpha}
-\frac{\partial G_{\mu}}{\partial \alpha}
\right)(\alpha) $ strictly increasing on $(\alpha_0, 1/2]$, and then 
\[
T_{\Lap_1, \delta}(\alpha) -
G_{\mu}(\alpha)
> 
T_{\Lap_1, \delta}(\alpha_0) -
G_{\mu}(\alpha_0) = 0.
\]

Otherwise, suppose that $\bar h(0) \le 0$.
Then there exists a $\eta_0 \in [0, \mu)$ such that 
$\bar h(\eta) \le 0$ if $\eta \le \eta_0$, and $\bar h(\eta) > 0$ if $\eta > \eta_0$.
Due to the reparametrization, we have
$g'(\alpha) \le 0$ if $\alpha \ge \Phi(-\eta_0/2)$, and $g'(\alpha) > 0$ if $\alpha < \Phi(-\eta_0/2)$. 
From \eqref{eq:pf-prop11-g}, 
\begin{align*}
\left(
\frac{\partial T_{\Lap_1, \delta}}{\partial \alpha}
-\frac{\partial G_{\mu}(\alpha)}{\partial \alpha}
\right)(\alpha)
> 0 \quad&\text{ if }\: \alpha < \Phi(-\eta_0/2), \\
\left(
\frac{\partial T_{\Lap_1, \delta}}{\partial \alpha}
-\frac{\partial G_{\mu}}{\partial \alpha}
\right)(\alpha)
\le 0 \quad&\text{ if }\: \alpha \ge \Phi(-\eta_0/2).
\end{align*}
Thus, $T_{\Lap_1, \delta} - G_{\mu}$ is strictly increasing in $\alpha < \Phi(-\eta_0/2)$, and is decreasing in  $\alpha \ge \Phi(-\eta_0/2)$.
However, we know that $(T_{\Lap_1, \delta} - G_{\mu})(\alpha_0) = 0$, and 
\[
(T_{\Lap_1, \delta} - G_{\mu})(1/2)
= 2\Phi^2(-\mu/2) - \Phi(-\mu) > 0,
\]
due to the log-concavity of $\Phi$.
This shows that $T_{\Lap_1, \delta}(\alpha) - G_{\mu}(\alpha) > 0 $ if $\alpha \in (\alpha_0, 1/2]$.

Next, suppose that $\alpha \in (1/2, 1)$.
It holds that $T_{\Lap_1, \delta}(1) = G_{\mu}(1) = 0$. 
Note that $T_{\Lap_1, \delta}(\alpha) = e^{-\delta}(1-\alpha)$ is the function of line segment connecting $(1/2, T_{\Lap_1, \delta}(1/2))$ to $(1, 0)$.
By the convexity, $G_{\mu}(\alpha) < T_{\Lap_1, \delta}(\alpha)$ for any $\alpha \in (1/2, 1)$.

Lastly, if $\alpha \in [0, \alpha_0)$, then $G_{\mu}(\alpha) \le T_{\Lap_1, \delta}(\alpha)$ always holds since the trade-off functions are symmetric on $y=x$.
\end{proof}

\begin{proof}[Proof of Proposition~\ref{prop:1dim_lap_stronger}]
(i)
By Lemma \ref{lem:univariate_laplace_tradeoff}, 
\[
\Lc_1(b) = T\left(\Lap(0, 1), \Lap\left(\frac{\Delta_1}{b}, 1\right)\right).
\]
If $b \ge b_{\mu}^{\Delta_1}$, then by the monotonicity of $\Lc_1(b)$ and Lemma \ref{lem:1dim-tight-lap-tf}, 
\begin{align*}
\Lc_1(b)
&\ge T\left(\Lap(0, 1), \Lap\left(\frac{\Delta_1}{b_{\mu}^{\Delta_1}}, 1)\right)\right) \\
&= T\left(\Lap(0, 1), \Lap\left(-2\log\{2\Phi(-\tfrac{\mu}{2})\}, 1)\right)\right) \\
&\ge G_{\mu}.
\end{align*}

For the `\textit{only if}' part, suppose that $b < b_{\mu}^{\Delta_1}$.
It follows that
\[
\Lc_1(b) < T\left(\Lap(0, 1), \Lap\left(-2\log\{2\Phi(-\tfrac{\mu}{2})\}, 1)\right)\right).
\]
However, from Lemma \ref{lem:1dim-tight-lap-tf}, 
$\Lc_1(b)(\alpha_0) < G_{\mu}(\alpha_0)$ where $\alpha_0 = \Phi(-\mu/2)$.
Thus, $\Mv_{\Lap}(\cdot, b)$ is not $\mu$-GDP.
These prove the first part of (i).

Next, we claim that $b_{\mu}^{\Delta_1} < b_{\mu}^{\varepsilon}$.
To prove the claim, it is enough to show that
\begin{equation} \label{eq:pf-prop11-logPhi}
   \log(\Phi(\mu/2)) - \log(\Phi(-\mu/2)) <  -2\log( 2\Phi(-\mu/2) ). 
\end{equation}
The above inequality \eqref{eq:pf-prop11-logPhi} holds if 
\begin{align*}
\log\left( 4\Phi\left(-\frac{\mu}{2}\right) \Phi\left(\frac{\mu}{2}\right) \right) < 0.
\end{align*}
So, if we show that $\Phi(\mu/2)\Phi(\mu/2) < \frac{1}{4}$ for any $\mu > 0$ the claim is proved.
Let $x = \Phi(\mu/2)$. Then $x \in (1/2, 1)$, and
$$\Phi\left(-\frac{\mu}{2}\right) \Phi\left(\frac{\mu}{2}\right) = (1-x)x = -(x-\frac{1}{2})^2 + \frac{1}{4} < \frac{1}{4}.$$

(ii)
The proof can be easily derived from the first part of the proposition.  
\end{proof} 

\subsubsection*{Proofs for Section \ref{sec:Lap_multivariate_sub}}

First, we introduce a notion of tensor product of trade-off functions. 
For two trade-off functions $ f = T(P, Q) $ and $ g = T(P', Q') $, the tensor product is defined as
$
    f \otimes g \coloneqq T(P \times P', \: Q \times Q').
$
The tensor product makes it easy to write trade-off function between multivariate distributions.
The following proposition and lemma are due to \cite{dong2022gaussianDP} and used in the proof of Lemma~\ref{lem-high2one}.

\begin{prop}[Proposition D.1 of \cite{dong2022gaussianDP}]
    \label{prop-tensor}
    $  $
\begin{enumerate}
    \item The tensor product $ \otimes $ is commutative and associative.
    \item If $ g_1 \ge g_2 $ then $ f \otimes g_1 \ge f \otimes g_2 $.
    \item $f \otimes {\rm{Id}} = {\rm{Id}} \otimes f = f$.
\end{enumerate}
\end{prop}

The following lemma is useful in handling a lower bound of composition of trade-off functions. 
\begin{lem}[Lemma A.5 of \cite{dong2022gaussianDP}]
    \label{lem-comp}
    Suppose $ T(P, Q) \ge f $, $ T(Q, R) \ge g $, then $ T(P, R)(\alpha) \ge g(1-f(\alpha)) $ for any $\alpha \in [0, 1]$.
\end{lem}

In the proof of Lemma~\ref{lem-high2one}, we importantly use the following Lemma, slightly modified from Lemma A.2 of \cite{awan2022log}, and apply it to the Laplace mechanism.

\begin{lem}
    \label{lem-tens}
    For $ \delta_1, \delta_2 \ge 0 $, consider two trade-off functions
    $ f = T(\rm{Lap}(0, 1), \rm{Lap}(\delta_1, 1) ) $ and
    $ g = T(\rm{Lap}(0, 1), \rm{Lap}(\delta_2, 1) ). $
    Then $ f \otimes g \ge T(\rm{Lap}(0, 1), \rm{Lap}(\delta_1+\delta_2, 1) ) $. 
\end{lem}
\begin{proof}
    Let $P = \rm{Lap}(0, 1) \times \rm{Lap}(0, 1)$, 
    $Q = \rm{Lap}(\delta_1, 1) \times \rm{Lap}(0, 1)$,
    and $R = \rm{Lap}(\delta_1, 1) \times \rm{Lap}(\delta_2, 1)$.
    Observe that
    \begin{align*}
        f 
        &= f \otimes \mbox{Id} \quad (\text{by Proposition \ref{prop-tensor}.3}) \\
        &= T(\Lap(0, 1), \Lap(\delta_1, 1)) \otimes T(\Lap(0, 1), \Lap(0, 1)) \\
        &= T(P, Q) \quad (\text{by the definition of tensor product}),
    \end{align*}
    and similarly, $g = T(Q, R)$.
    However, $(f \otimes g)(\alpha) = T(P, R)$, and then Lemma \ref{lem-comp} implies that 
    $$ 
    (f \otimes g)(\alpha) = T(P, R) \ge g(1-f(\alpha)) 
    $$ 
    for any $\alpha \in (0, 1)$.
    Let $F$ be the distribution of $\Lap(0, 1)$.
    From Remark \ref{rmk:log-conc-tf}, we have
    \begin{align*}
    g(1-f(\alpha))
    &= F(F^{-1}(1-(1-f(\alpha))) - \delta_2) \\ 
    &= F(F^{-1}( F(F^{-1}(1-\alpha) - \delta_1) ) - \delta_2) \\ 
    &= F(F^{-1}(1-\alpha) - \delta_1 - \delta_2) \\ 
    &= T(\Lap(0, 1), \Lap(\delta_1+\delta_2, 1)).
    \end{align*}
    This completes the proof.

\end{proof}

Now, we are ready to prove Lemma~\ref{lem-high2one}. 

\begin{proof}[Proof of Lemma~\ref{lem-high2one}]
We first prove \eqref{eq:trade-off_multiLap}.
Fix any $\deltav = (\delta_1, \dots, \delta_p) \in \Real^p$.
By symmetricity of Laplace distribution and trade-off function,  we may assume that $ \delta_1 \ge \delta_2 \ge \dots \ge \delta_p \ge 0 $ without loss of generality.
Consider $ f_i = T(\Lap(0, 1), \Lap(\delta_i, 1)) $ and $ f = T(\Lap_p(\0v_p, 1), \Lap_p(\deltav, 1)) $.
To prove the lemma, we use the mathematical induction on $ p $.
When $ p = 1 $, it holds trivially, and Lemma \ref{lem-tens} implies that \eqref{eq:trade-off_multiLap} also holds when $ p = 2 $.
Now, suppose that the lemma holds for $ p = k $ ($ k \ge 2 $).
Then
\begin{align*}
    f
    &= f_1 \otimes f_2 \otimes \dots \otimes f_{k+1} \\
    &= f_1 \otimes (f_2 \otimes \dots \otimes f_{k+1})
    \quad (\text{by Proposition \ref{prop-tensor}.1})
    \\
    &\ge f_1 \otimes T_{\Lap_{1}, \delta_2 + \dots + \delta_{k+1}}
    \quad (\text{by the induction assumption and Proposition \ref{prop-tensor}.2})
    \\
    &\ge T_{\Lap_{1}, \delta_1 + \delta_2 + \dots + \delta_{k+1}}
    \quad (\text{by Lemma \ref{lem-tens}})
    \\
    &= T(\Lap(0, 1), \Lap(\|\deltav\|_1, 1)).
\end{align*}

Suppose that $\deltav = (\delta_1, 0, \dots, 0)$
so that $\delta_1 \neq 0$ and other elements in $\deltav$ are all zero.
Then, the trade-off function becomes
\[
T(\Lap_p(\0v_p, 1), \Lap_p(\deltav, 1))
= T(\Lap_1(0, 1), \Lap_1(\delta_1, 1))
= T(\Lap_1(0, 1), \Lap_1(\|\deltav\|_1, 1)).
\]
This proves (i).

To prove (ii), suppose that at least two elements in $\deltav = (\delta_1, \dots, \delta_p)$ are nonzero.
We may assume that $\delta_i > 0$ for all $i = 1, \dots, p$ and $p \ge 2$.
First, consider the case of $p = 2$. 
More precisely, let $\deltav = (\delta_1, \delta_2)$ where  $\delta_1 \ge \delta_2 > 0$.
Note that $\alpha_0^{(1)} = \frac{1}{2}e^{-\frac{\|\deltav\|_1}{2}}$, and $\alpha_0^{(2)} = \frac{1}{4} e^{-\frac{\|\deltav\|_1}{2}}(2+\delta_2)$ are fixed points of $T_{\Lap_1, \|\deltav\|_1}$, and $T_{\Lap_2, \deltav}$, respectively, i.e., $T_{\Lap_1, \|\deltav\|_1}(\alpha_0^{(1)}) = \alpha_0^{(1)} $, and $T_{\Lap_2, \deltav}(\alpha_0^{(2)}) = \alpha_0^{(2)} $.
Since $\delta_2 > 0$, it follows that 
\[
\alpha_0^{(1)} = 
\frac{1}{4}e^{-\frac{\|\deltav\|_1}{2}}(2+\delta_2)
> 
\frac{1}{2}e^{-\frac{\|\deltav\|_1}{2}}
= \alpha_0^{(2)}.
\]
Then 
\[
T_{\Lap_1, \|\deltav\|_1}(\alpha_0^{(1)})
< T_{\Lap_2, \deltav}(\alpha_0^{(2)})
\le T_{\Lap_2, \deltav}(\alpha_0^{(1)}).
\]
By the continuity of $T_{\Lap_2, \deltav} - T_{\Lap_1, \|\deltav\|_1}$, there exists an interval $(a, b)$ such that contains $\alpha_0^{(1)}$, and for any $\alpha \in (a, b)$, 
\[
(T_{\Lap_2, \deltav} - T_{\Lap_1, \|\deltav\|_1})(\alpha) > 0.
\]

Next, assume $p \ge 3$, and let $\deltav = (\delta_1, \dots, \delta_p) $ where $\delta_i > 0$ for all $i$.
Denote $\deltav = (\delta_1, \deltav_{-})$, where $\deltav_{-} = (\delta_2, \dots, \delta_p)$.
Then
\begin{align*}
    T_{\Lap_p, \deltav}
    &= T_{\Lap_1, \delta_1} \otimes T_{\Lap_{p-1}, \deltav_{-}} \\
    &\ge T_{\Lap_1, \delta_1} \otimes T_{\Lap_{1}, \|\deltav_{-}\|_1} \quad (\text{by \eqref{eq:trade-off_multiLap}  and Proposition \ref{prop-tensor}.2}) \\
    &= T_{\Lap_2, (\delta_1, \|\deltav_{-}\|_1)^{\top}}.
\end{align*}
Since $\delta_1 \neq 0$ and $\|\deltav_{-}\|_1 \neq 0$, from the case of $p=2$, it follows that there exists an interval $(a, b)$, $0 < a < b< 1$ such that, for any $\alpha \in (a, b)$, 
\[
T_{\Lap_2, (\delta_1, \|\deltav_{-}\|_1)^{\top}}(\alpha)
> T_{\Lap_1, \delta_1+\|\deltav_{-}\|_1}(\alpha) = T_{\Lap_1, \|\deltav\|_1}(\alpha).
\]
Thus, $T_{\Lap_p, \deltav}(\alpha) > T_{\Lap_1, \|\deltav\|_1}(\alpha).$ 
Also, $(a, b)$ contains $\|\deltav\|_1$.
This completes the proof of (ii) of the lemma.
\end{proof}

\begin{proof}[Proof of Theorem~\ref{thm:improved_lap_mech}]

(i)
We first show the \textit{`if'} part. 
We have already observed that for any given $\Mv_{\Lap}(\cdot; b)$ of $\thetav$,
\begin{align*}
\Lc_p(b; \thetav, \mathcal{X}^n)
\ge 
T\left(
    \Lap(0, 1),
    \Lap\left(\Delta_1(\thetav, \mathcal{X}^n)/b, 1 \right) 
\right).
\end{align*}
Since $b \ge b_{\mu}^{\Delta_1} \ge b_{\mu}^{\Delta_1(\thetav, \Xc^n)}$,   Proposition \ref{prop:1dim_lap_stronger}(i) implies  
$ T\left(
    \Lap(0, 1),
    \Lap\left(\Delta_1(\thetav, \mathcal{X}^n)/b, 1 \right) 
\right)
\ge G_{\mu}$, as required. 

The  \textit{`only if'} part can be verified by showing the contrapositive: If $b < b_\mu^{\Delta_1}$, then there exists a statistic $\thetav$ with $\Delta_1(\thetav, \Xc^n) = \Delta_1$ such that $\Mv_{\Lap}(\cdot; b)$ does not satisfy $\mu$-GDP. 

Let $\thetav$ and $\Xc^n$ be such that for some $S,S' \in \Xc^n$ satisfying $S\sim S'$, $\thetav(S) - \thetav(S') = (\Delta_1,0,\ldots,0)^\top$. Then 
$$\Lc_p(b ; \thetav, \Xc^n) = \inf_{S\sim S'} 
T(\Lap_p( (\thetav(S) - \thetav(S'))/b, 1), \Lap_p(\0v, 1)) \le T( \Lap_p( (\Delta_1/b,0,\ldots,0),1), \Lap_p(\0v,1)),$$ and by Lemma\ref{lem-high2one}(i), we have $\Lc_p(b ; \thetav, \Xc^n) \le T(\Lap(\Delta_1/b), \Lap(0,1)).$ 
By  Proposition \ref{prop:1dim_lap_stronger}, for any $b < b_\mu^{\Delta_1}$, 
there exists $\alpha \in (0,1)$ such that 
$$\Lc_p(b ; \thetav, \Xc^n)(\alpha) \le T(\Lap(\Delta_1/b), \Lap(0,1))(\alpha) < G_\mu(\alpha).$$ This shows that 
$\Mv_{\Lap}(\cdot; b)$, perturbing $\thetav$, is not $\mu$-GDP for any $b < b_\mu^{\Delta_1}$. 
 
(ii)
This can be easily derived from the first part of the theorem.

(iii)
Consider a statistic $\thetav$ such that the sensitivity space $\mathcal{S}_{\thetav}$ is 
\[
\mathcal{S}_{\thetav} = \left\{
\frac{\Delta_1}{2}(\ev_i + \ev_j) : i, j = 1, \dots, p
\right\}.
\]
Then, by the process similar to \eqref{eq:freq_1}, we obtain
\[
\Lc_p(b) = T\left(\Lap_2(\0v, 1), \Lap_2\left(\frac{\Delta_1}{2b}\1v, 1\right)\right)
\]
From the above inequality and Lemma \ref{lem:1dim-tight-lap-tf}, 
we have 
\begin{equation} \label{eq:pf-thm13-iii-eq1}
\Lc_p(b_{\mu}^{\Delta_1})(\alpha)
> G_{\mu}(\alpha),
\end{equation}
for any $\alpha \in (0, 1)$.
Define a function $\Fc(b, \alpha): (0, \infty)\times [0,1] \to \Real$ by 
\[
\Fc(b, \alpha) = \Lc_p(b)(\alpha) - G_\mu(\alpha).
\]
The inequality \eqref{eq:pf-thm13-iii-eq1} is equivalent to 
$\Fc(b_{\mu}^{\Delta_1}, \alpha) > 0$ for any $\alpha \in (0, 1)$.
Since 
$$
\lim_{\alpha \to 0+} \frac{\partial G_{\mu}}{\partial \alpha}(\alpha) = -\infty,
$$
and $\Lc_p(b)(\alpha) = -e^{\Delta_1/b}\alpha + 1$ when $\alpha < \frac{1}{4}e^{-\Delta_1/b}$,
there exists $c_1, c_2 > 0$ such that 
if $b \in (b_{\mu}^{\Delta_1}-c_1, b_{\mu}^{\Delta_1})$ and $\alpha \in (0, c_2)$ then $\Fc(b, \alpha) > 0$.
Due to the trade-off functions are symmetric on $y=x$, we have
$\Fc(b, \alpha) > 0$ if $\alpha \in (0, c_2) \cup (1-c_2, 1)$ whenever $b \in (b_{\mu}^{\Delta_1}-c_1, b_{\mu}^{\Delta_1})$.

On the other hand, since $[c_2, 1-c_2]$ is compact, there exists a $c_3 > 0$ such that if $b \in (b_{\mu}^{\Delta_1} - c_3, b_{\mu}^{\Delta_1})$, and $\alpha \in [c_2, 1-c_2]$, then $\Fc(b, \alpha) > 0.$

Take $C = \min(c_1, c_3)$, and choose $b_0 \in (b_{\mu}^{\Delta_1} - C, b_{\mu}^{\Delta_1})$.
Then, by the above arguments, we have $\Fc(b_0, \alpha) > 0$ for any $\alpha \in (0, 1)$.
This shows that, for any $\alpha \in (0, 1)$,
\[
\Lc_p(b_0)(\alpha) > G_{\mu}(\alpha).
\]
This completes the proof of (iii).
\end{proof}

\section{Technical details omitted from Section~\ref{sec:comparison}} \label{sec:detailsforComparionsSection}

\begin{lem}\label{lem:re}
Let $\mu>0$ Let $\thetav: \Xc^n \to \Real^p$, and let $\Mv_G$ and $\Mv_{\Lap} = \Mv_{\Lap}(\cdot, b_{\mu}^{\Delta_1})$  be the ordinary Gaussian and Laplace mechanisms, respectively, both calibrated to satisfy $\mu$-GDP. Let $\Delta_r = \Delta_r(\thetav; \Xc^n)$ for $r = 1,2$. For any $r\ge 1$,
\begin{itemize}
  \item[(i)] $L_r(\Mv_{\Lap}) = p (b_\mu^{\Delta_1})^r \Gamma(r+1)$ and $L_r(\Mv_G) = p (\Delta_2/\mu)^r \frac{2^{r/2}}{\sqrt{\pi}}\Gamma\left( \frac{p+1}{2} \right)$
  \item[(ii)] ${\rm RE}_r(\Mv_G , \Mv_{\Lap} ) = \left(\frac{\Delta_1}{\Delta_2}\right)^r  2^{\frac{r}{2}} \Gamma\left(\frac{r}{2} + 1\right) \left(\frac{\mu}{-2 \log (2 \Phi (-\frac{\mu}{2}))}\right)^r$
  \item[(iii)] For each given $r$, ${\rm RE}_r(\Mv_G , \Mv_{\Lap} )$ increases  as  $\Delta_1/ \Delta_2$ increases or $\mu$ decreases.
\end{itemize}
\end{lem}

\begin{proof}
(i)
First, we calculate $L_r$ cost of the improved Laplace mechanism.
For $Y \sim \Lap(0; 1)$, the $r$th moment is
\begin{align*}
    \mathbb E |Y|^r
    = \int_{-\infty}^{\infty} \frac{1}{2} e^{-|y|} \cdot |y|^r dy
    = \int_0^{\infty} y^re^{-y} dy
    = \Gamma(r+1).
\end{align*}
Let $\xiv = (Y_1, \dots, Y_p)^\top \sim \Lap_p(\0v; b_\mu^{\Delta_1} )$.
Then, the $L_r$ cost becomes
\begin{align*}
L_r(\Mv_{\Lap}(S; b_\mu^{\Delta_1}))
= p \mathbb E  |Y_1|^r
= p (b_\mu^{\Delta_1})^r \mathbb E  |Y|^r
= p (b_\mu^{\Delta_1})^r \Gamma(r+1).
\end{align*}
In the normal case, it is well known that
$$
\mathbb E |X|^r = \frac{2^{r/2}}{\sqrt{\pi}}\Gamma\left( \frac{p+1}{2} \right),
$$
where $X \sim N(0, 1)$.
Thus, we have
$$
L_r(\Mv_{G}(S))
= p (\Delta_2/\mu)^r \mathbb E |X|^r
= p (\Delta_2/\mu)^r \frac{2^{r/2}}{\sqrt{\pi}}\Gamma\left( \frac{p+1}{2} \right).
$$

(ii)
From (i), the relative efficiency becomes
$$
{\rm RE}_r(\Mv_G, \Mv_{\rm Lap})
= \frac{\sqrt{\pi}}{2^{r/2}} \left( \frac{b_\mu^{\Delta_1}}{\Delta_2/\mu} \right)^r \frac{\Gamma(r+1)}{\Gamma(\frac{r+1}{2})}
= \left(\frac{\Delta_1}{\Delta_2}\right)^r  2^{\frac{r}{2}} \Gamma\left(\frac{r}{2} + 1\right) \left(\frac{\mu}{-2 \log (2 \Phi (-\frac{\mu}{2}))}\right)^r.
$$
The last equality holds from the definition of $b_{\mu}^{\Delta_1}$ and by Legendre duplication formula for gamma function (5.5.5 of \cite{olver2010nist}), which says that
$$
\Gamma(r+1) = \frac{2^{r}}{\sqrt{\pi}} \Gamma \left( \frac{r+1}{2} \right)
\Gamma \left( \frac{r}{2} + 1 \right).
$$

(iii)
From (ii), it is obvious that ${\rm RE}_r(\Mv_G, \Mv_{\rm Lap})$ is increasing in $\Delta_1/\Delta_2$ for each given $r$. 
To check whether ${\rm RE}_r(\Mv_G, \Mv_{\rm Lap})$ is an increasing function of $\mu$, it is enough to show that 
\begin{equation} \label{eq:L_r-cost-cond}
g(\mu) \coloneqq \frac{-\log (2 \Phi (-\frac{\mu}{2}))}{\mu}
\end{equation}
increases as $\mu$ increases. 
Define a function $h$ on $(0, \infty)$ as
$
h(\mu) = -\log\left( \Phi \left(- \frac \mu 2 \right) \right).
$
Then $g$ is expressed by using $h$ as
$$
g(\mu) 
= \frac{ -\log(\Phi(-\mu/2)) - (-\log(1/2)) }{\mu}
= \frac{ h(\mu) - h(0) }{\mu - 0}.
$$
Let $\phi$ be a density of $N(0, 1)$.
Note that
$$
h'(\mu)
= \frac{\phi(\mu/2)}{2\Phi(-\mu/2)},
$$
and
$$
h''(\mu)
= \frac{\phi(\mu/2)}{4(\Phi(-\mu/2))^2}
( \phi(\frac{\mu}{2}) - \frac{\mu}{2}\Phi(-\frac{\mu}{2}) ).
$$
However, it is well known that $\Phi(-x) < \phi(x)/x$ for any $x > 0$, and thus $h''(\mu) > 0$ for any $\mu >0$.
So, $h$ is a convex function, and as a consequence, $g(\mu)$ is increasing in $\mu > 0$.
\end{proof}

\section{Additional simulation results} \label{sec:app_sim}

In this section, we provide simulation results for the private homogeneity tests, as referred to in Section~\ref{sec:sim_gof}. 
We use the same model settings used for comparing GOF tests, with 5-year age groups ($p=18$ categories) and 10-year age groups ($p=9$ categories).
In 5-year age grouped table ($p=18$), categories are divided into \texttt{age:0-4}, \texttt{age:5-10}, \dots, \texttt{age:80-84} and \texttt{age:85+}.
Similarly, we set \texttt{age:0-9}, \texttt{age:10-19}, \dots, \texttt{age:80+} as the categories in the case of 10-year age grouped table ($p=9$).

For the homogeneity test (Algorithm \ref{alg:BootHomTest}),  $2 \times p$ contingency table is generated, in which the rows of the table are sampled from $\mbox{Multinomial}(n, \piv_1)$ and $\mbox{Multinomial}(n, \piv_2)$ respectively.
We set $\piv_1 = \piv_2 = \piv_{\rm Kor}^{(p)} $ to simulate the type I error rate, and set $(\piv_1, \piv_2) = (\piv_{\rm Kor}^{(p)}, \piv_{\rm US}^{(p)})$ for the power analysis.

The simulation results for homogeneity tests are illustrated in  Figures \ref{fig:hom-type1} and \ref{fig:hom-power}.

\begin{figure}
\centering
\begin{subfigure}{0.4\textwidth}
    \centering
    \includegraphics[width=.8\textwidth]{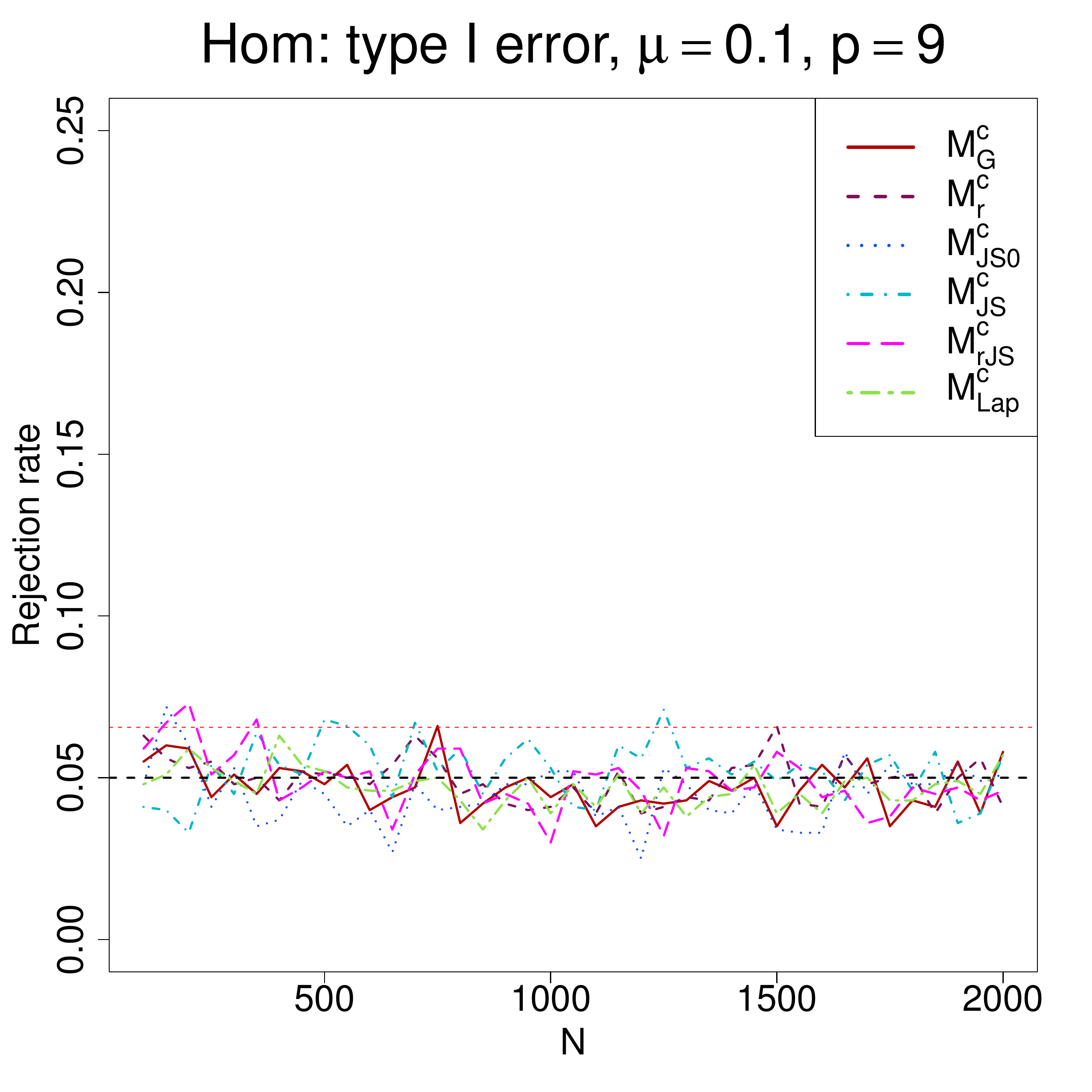}
\end{subfigure}
\begin{subfigure}{0.4\textwidth}
    \centering
    \includegraphics[width=.8\textwidth]{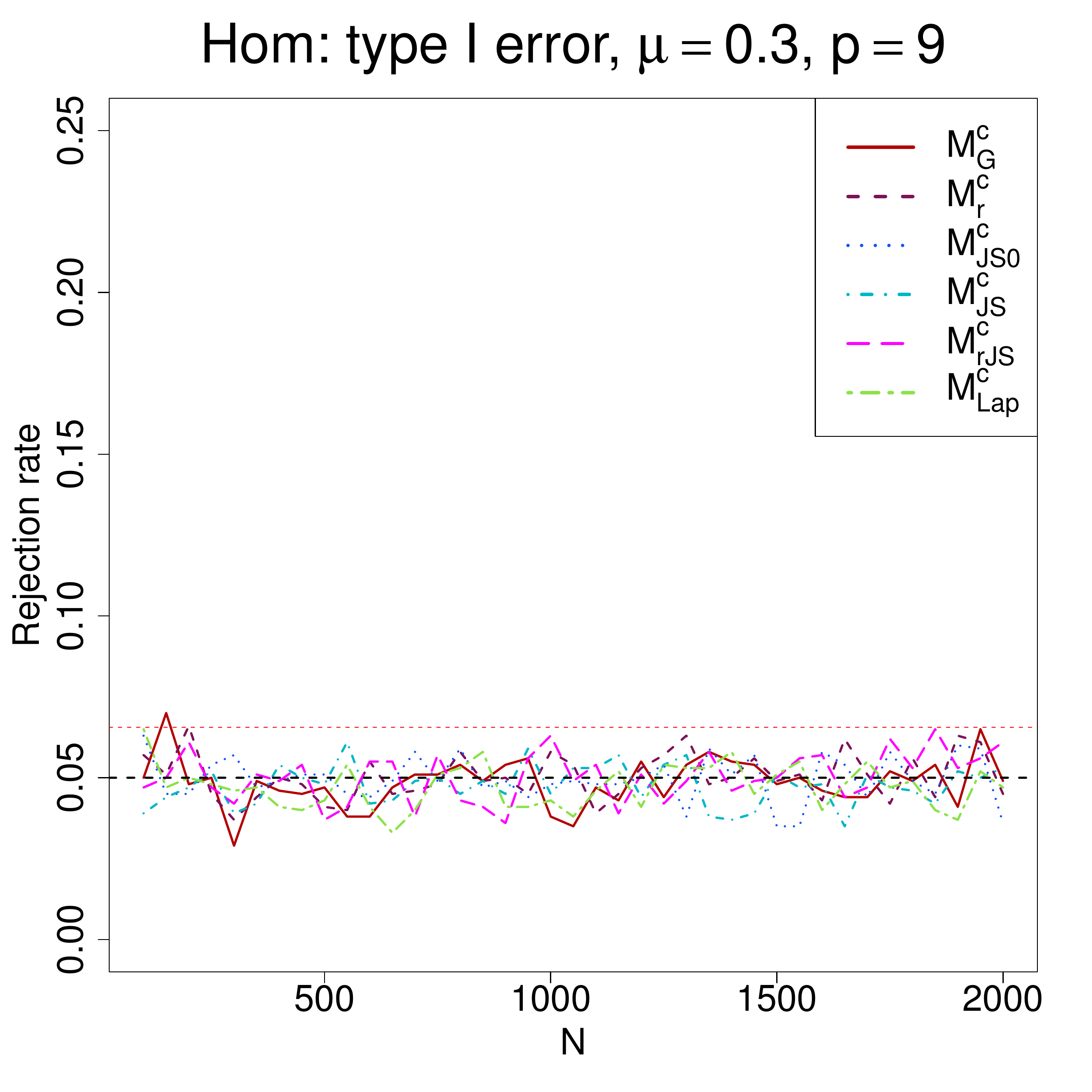}
\end{subfigure}
\begin{subfigure}{0.4\textwidth}
    \centering
    \includegraphics[width=.8\textwidth]{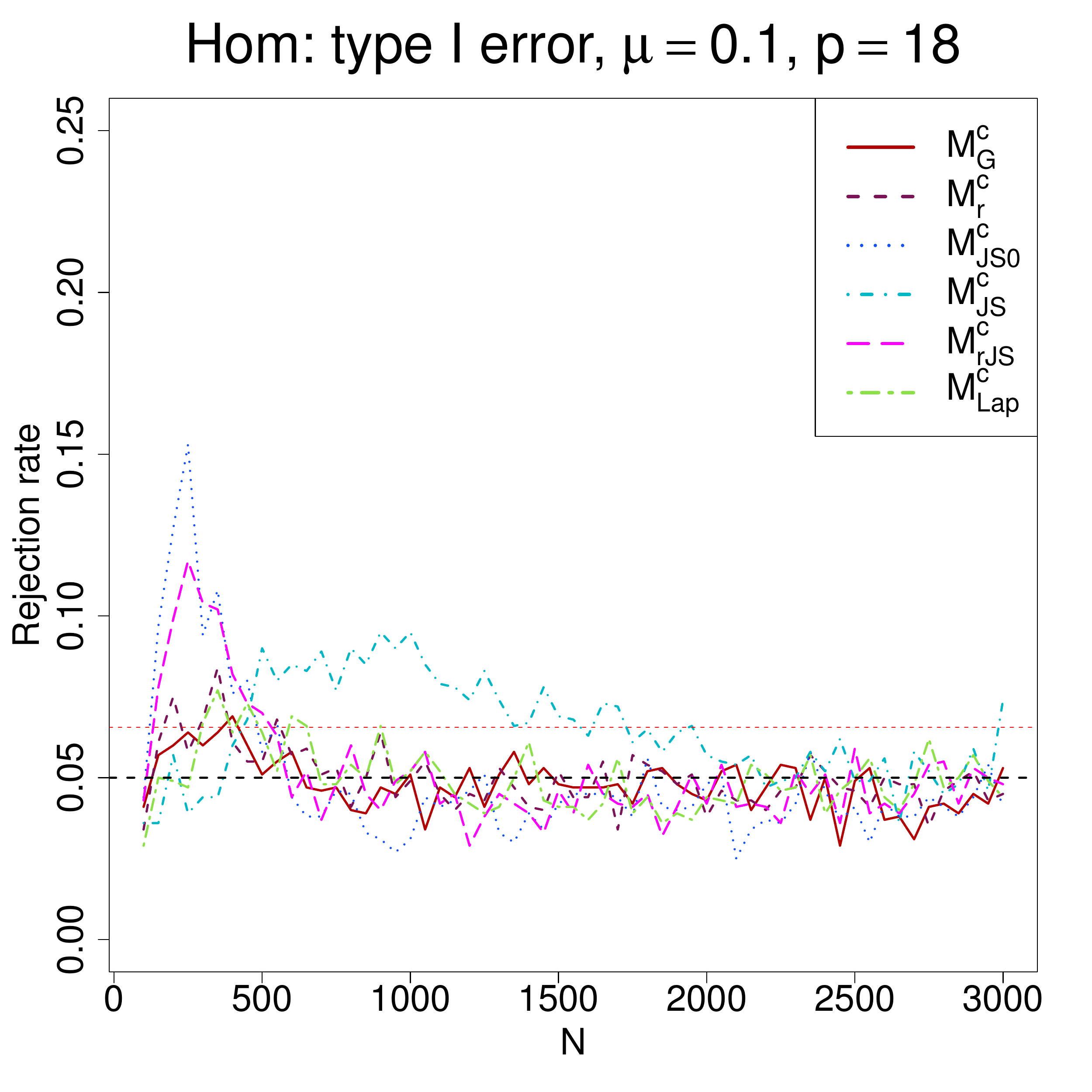}
\end{subfigure}
\begin{subfigure}{0.4\textwidth}
    \centering
    \includegraphics[width=.8\textwidth]{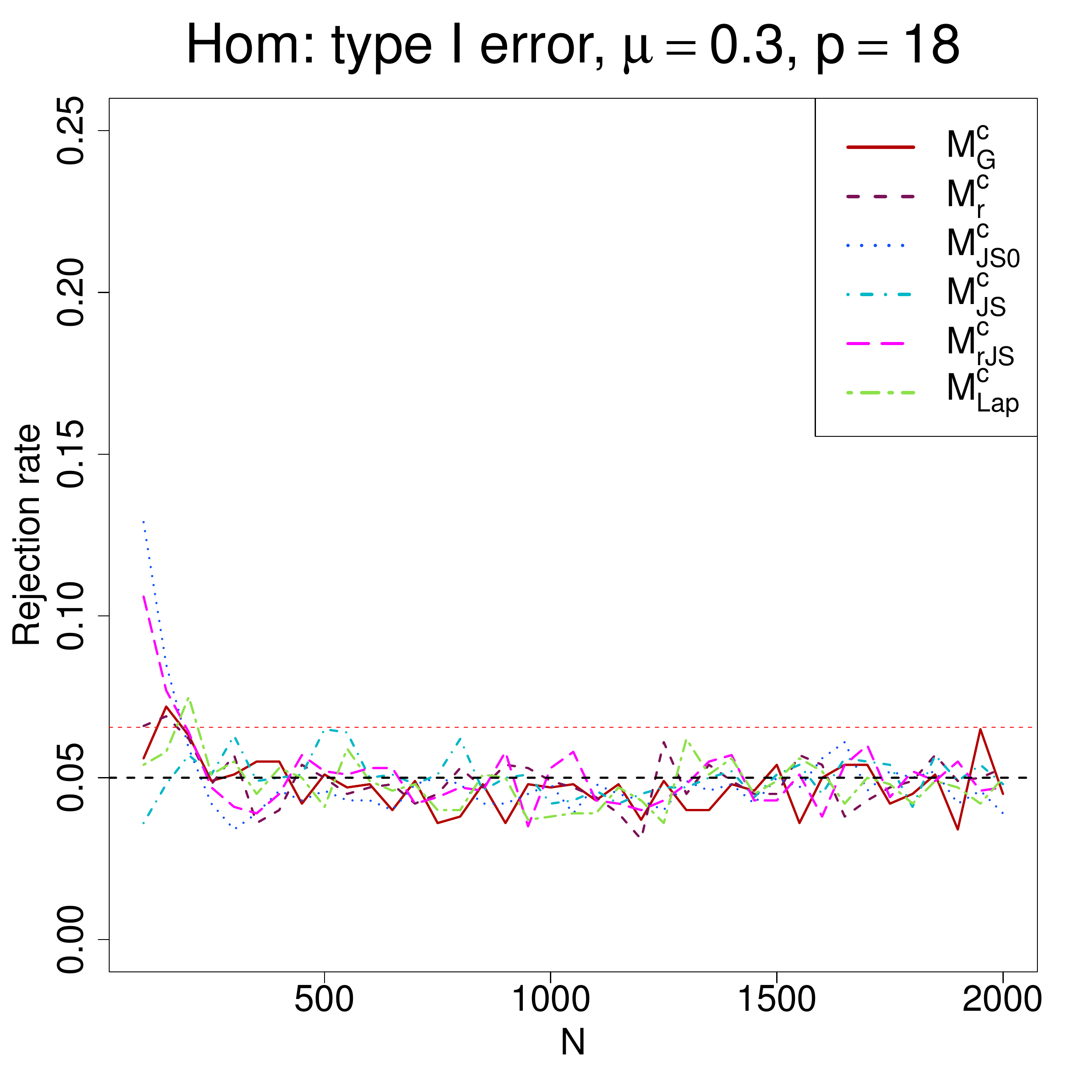}
\end{subfigure}
\caption{Graph of empirical type I error rates of homogeneity tests for each mechanism. The black dotted line is the level of the test, $\alpha = 0.05$. The red dotted line is the maximum empirical rejection rate value that contains $0.05$ within two times of standard error.}
\label{fig:hom-type1}
\end{figure}

\begin{figure}
\centering
\begin{subfigure}{0.4\textwidth}
    \centering
    \includegraphics[width=.8\textwidth]{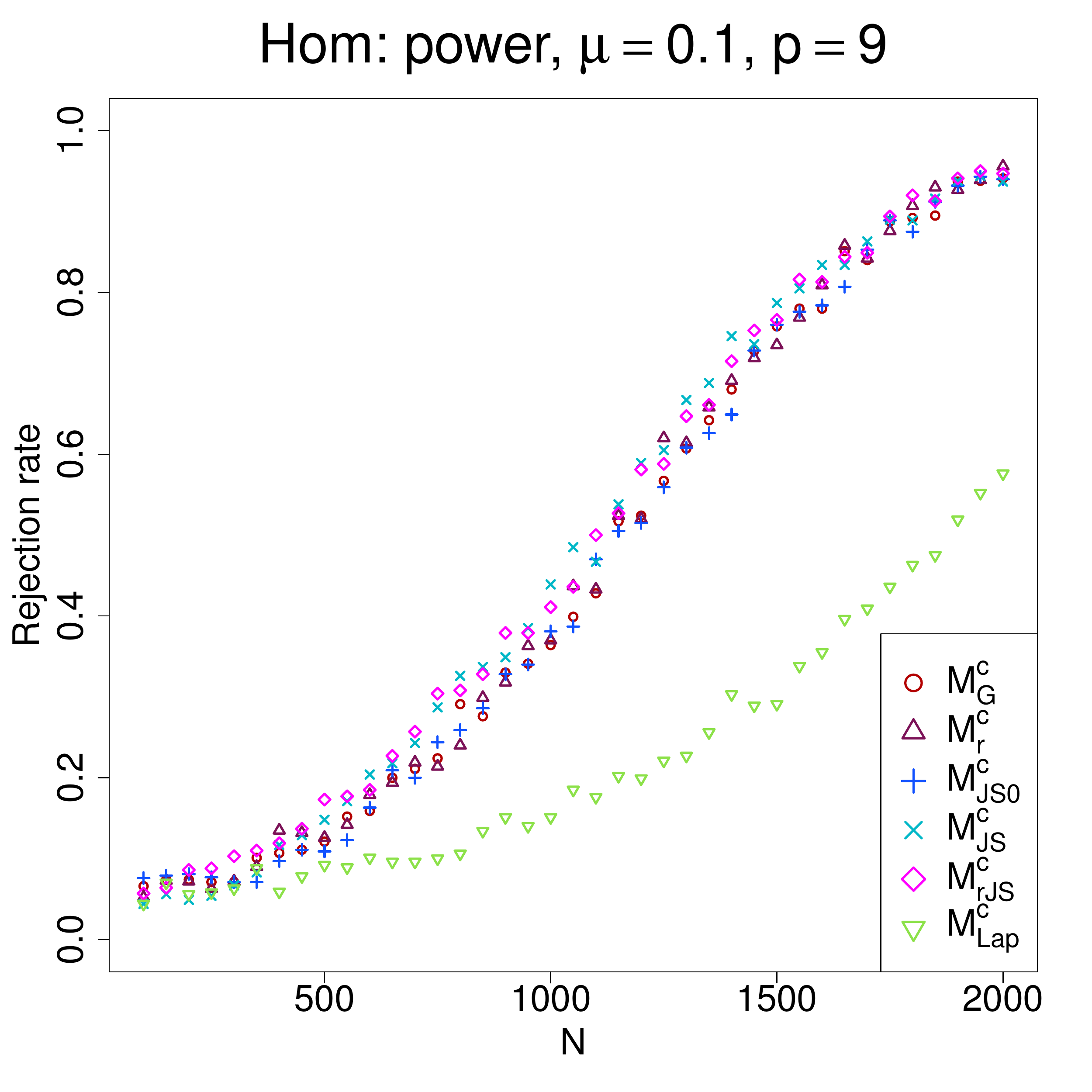}
\end{subfigure}
\begin{subfigure}{0.4\textwidth}
    \centering
    \includegraphics[width=.8\textwidth]{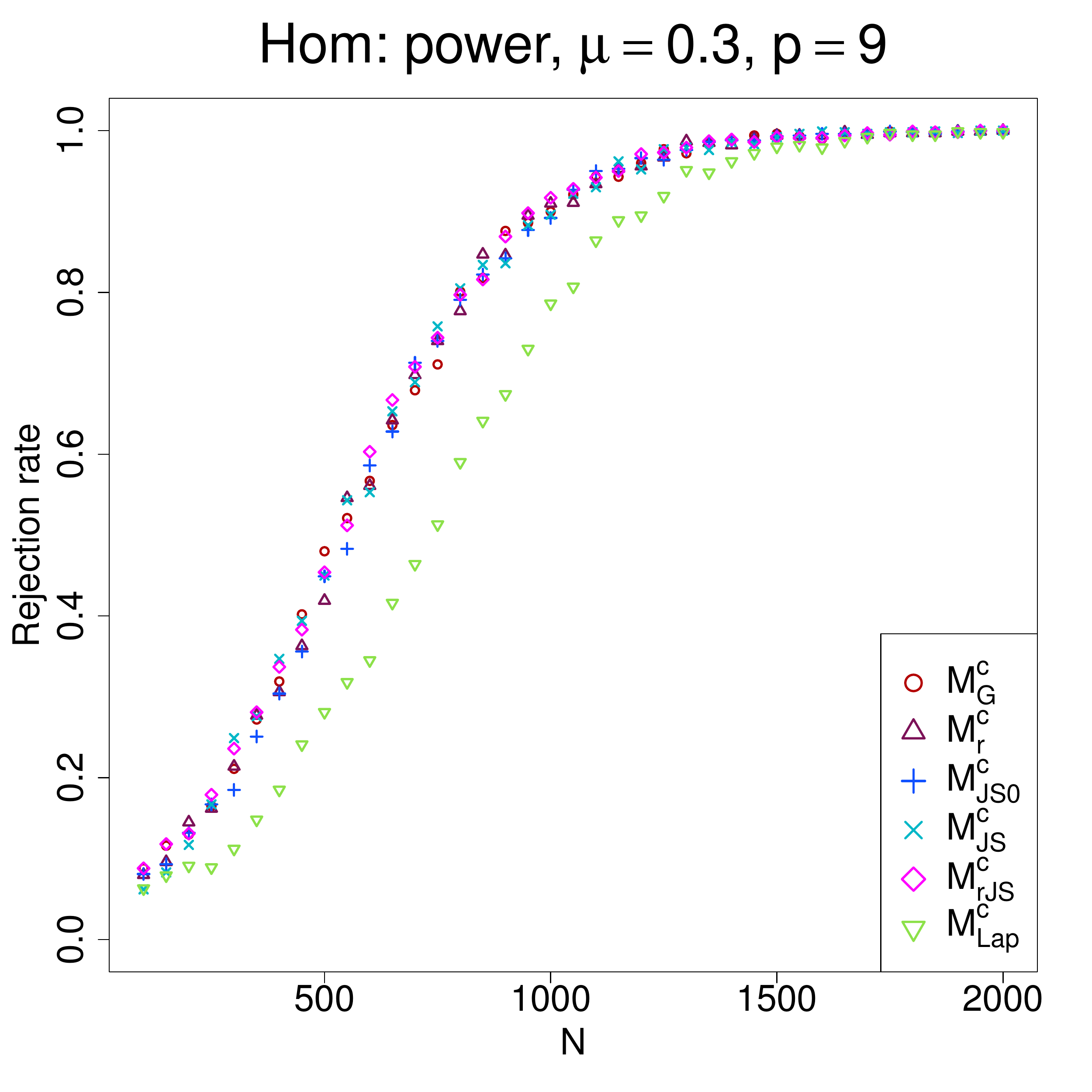}
\end{subfigure}
\begin{subfigure}{0.4\textwidth}
    \centering
    \includegraphics[width=.8\textwidth]{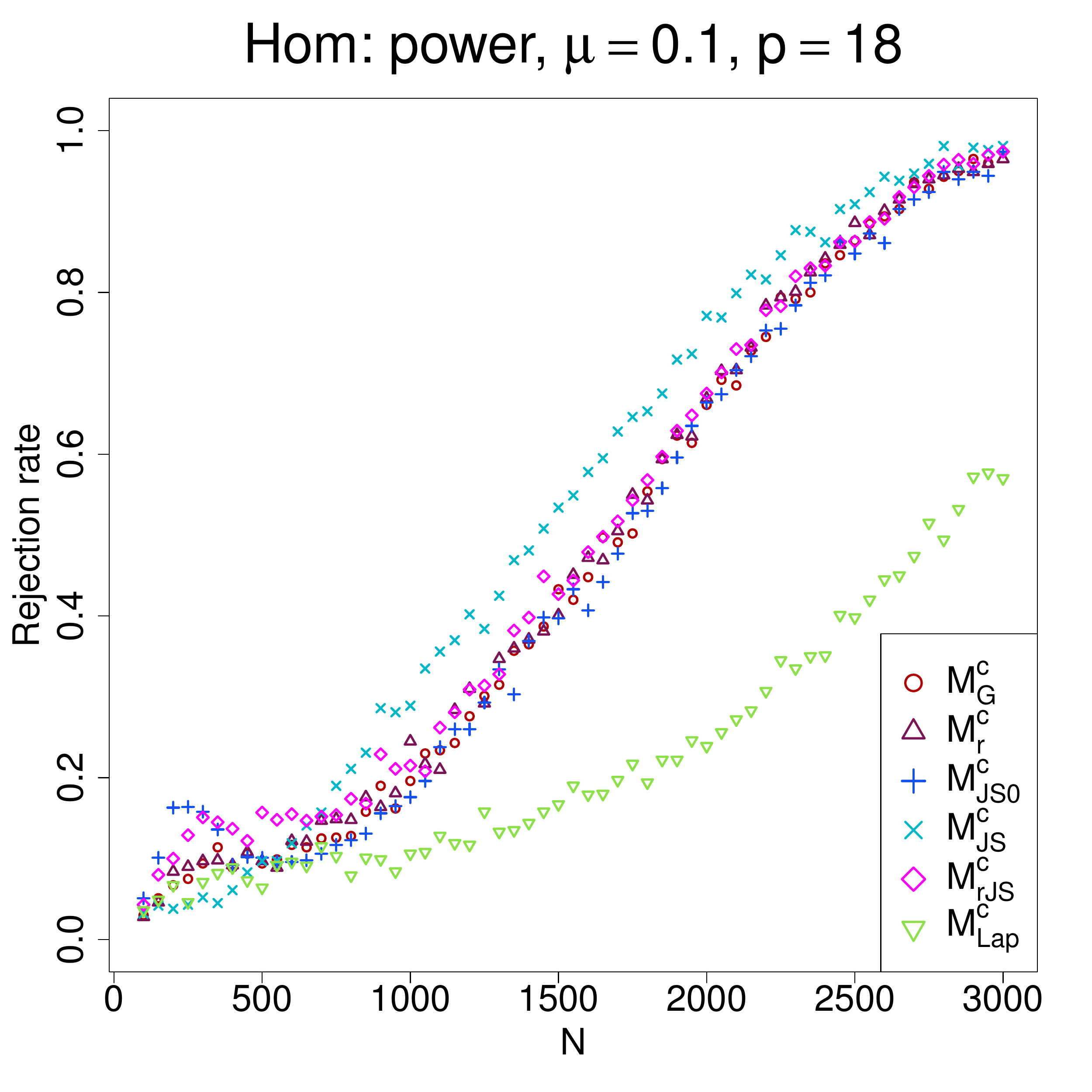}
\end{subfigure}
\begin{subfigure}{0.4\textwidth}
    \centering
    \includegraphics[width=.8\textwidth]{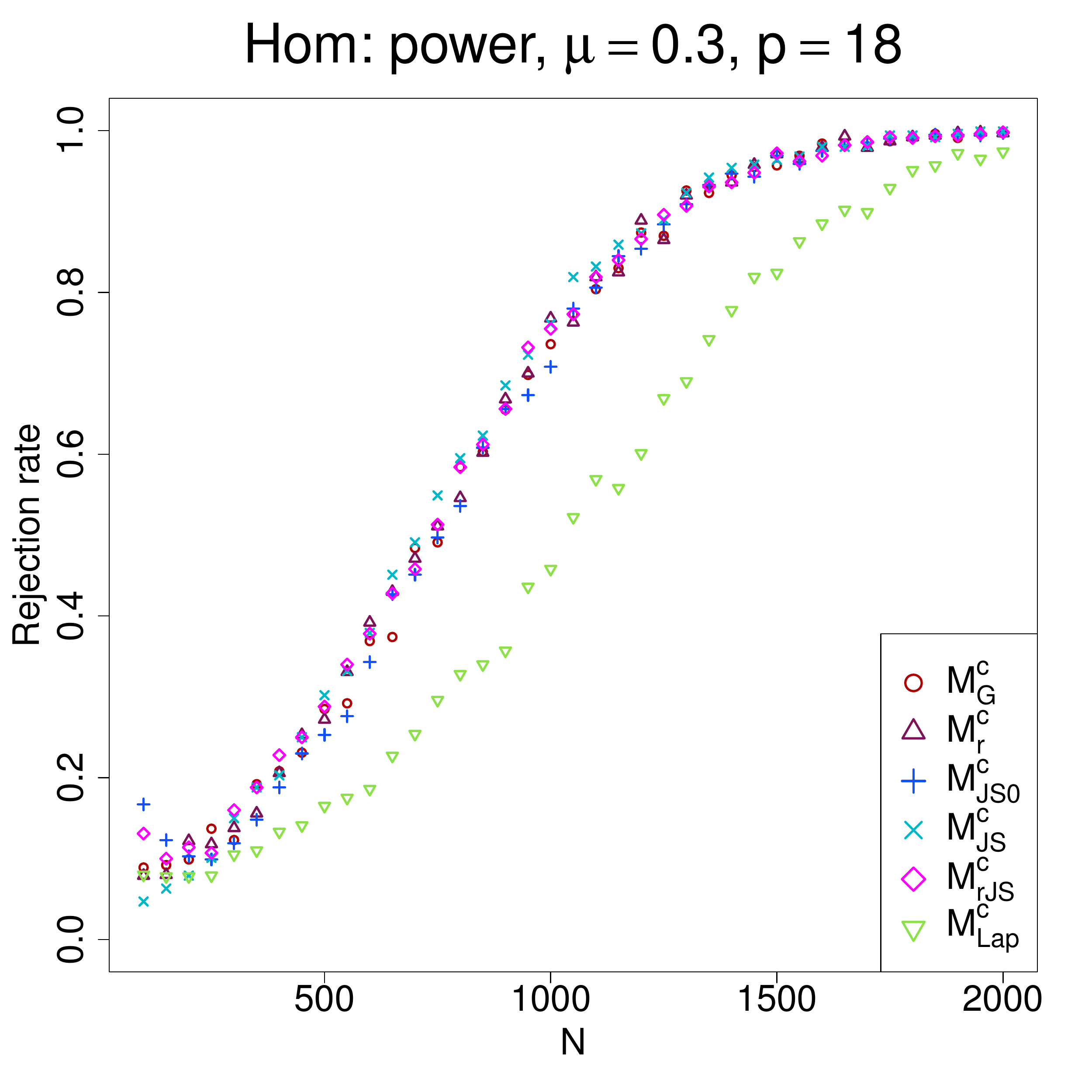}
\end{subfigure}
\caption{Graph of empirical power of homogeneity tests for each mechanism.}
\label{fig:hom-power}
\end{figure}

 As in the GOF test, except for the case of $\mu = 0.1$ and $p=18$, the type I errors are well controlled in  homogeneity test (Figure \ref{fig:hom-type1}).
Converse to the GOF test, the James--Stein processed mechanisms do not control type I error, and one remarkable observation is that $\Mv_{JS}^c$ does not control type I error in the range of $500 \le n \le 1500$.
We speculate that this phenomenon occurs because the shrinkage and truncation steps break the uniform distribution of empirical p-values.

As in the GOF test, Figure \ref{fig:hom-power} reveals that the Laplace mechanism has the lowest power compared to the Gaussian-type mechanisms.
When $\mu = 0.3$, it is hard to say that one mechanism dominates all others with respect to the power of the test.
For the analysis of the power of homogeneity test of $\mu=0.1$ and $p=18$, we need some care.
Since $\Mv_{JS}^c$ does not control the type I error unless $n > 1500$, we cannot conclude that $\Mv_{JS}^c$ dominates the other mechanisms in terms of power when $n \le 1500$ because type I error is not controlled.
However, $\Mv_{JS}^c$ still performs better than other methods for $n > 1500$.
Finally, in the scenario of $\mu=0.1$ and $p=9$, $\Mv_{JS}^c$ and $\Mv_{rJS}^c$ perform better than other mechanisms.
 
\end{document}